\documentclass[12pt]{article}
\usepackage{amsmath,enumerate,amsbsy,amsfonts,amssymb,mathabx,amscd,graphicx,multirow,xcolor,subcaption}
\usepackage{caption}
\usepackage[round,authoryear]{natbib}
\usepackage{authblk}
\usepackage{mathrsfs}
\usepackage[flushleft]{threeparttable}
\usepackage{booktabs}
\RequirePackage[colorlinks,citecolor=blue,urlcolor=blue]{hyperref}
\usepackage{enumerate}
\usepackage{float}
\usepackage{amsmath}
\usepackage{setspace}

\def\spacingset#1{\renewcommand{\baselinestretch}%
	{#1}\small\normalsize} \spacingset{1}

\usepackage{amsthm,amsmath,enumerate,amsbsy,amsfonts,amssymb,mathabx,amscd,graphicx,algorithm, indentfirst}
\usepackage{natbib}
\usepackage{geometry}
\usepackage{authblk, caption, subcaption}
\usepackage{mathrsfs}
\usepackage{algorithmic}
\usepackage{appendix}

\usepackage{epsfig}
\newtheorem{definition}{Definition}
\newtheorem{theorem}{Theorem}
\newtheorem{corollary}{Corollary}
\newtheorem{lemma}{Lemma}

\newtheorem{remark}{Remark}
\newtheorem{example}{Example}
\newtheorem{condition}{Condition}

\usepackage{multirow}

\newcommand{\bzero}{\boldsymbol 0}
\newcommand{\bLambda}{\boldsymbol \Lambda}

\newcommand{\bSigma}{\boldsymbol \Sigma}

\newcommand{\bGamma}{\boldsymbol \Gamma}
\newcommand{\bTheta}{\boldsymbol \Theta}

\newcommand{\bDelta}{\boldsymbol \Delta}
\newcommand{\bXi}{\boldsymbol \Xi}
\newcommand{\bpsi}{\boldsymbol \psi}

\newcommand{\bphi}{\boldsymbol \phi}
\newcommand{\bPhi}{\boldsymbol \Phi}
\newcommand{\bmu}{\boldsymbol \mu}
\newcommand{\bxi}{\boldsymbol \xi}
\newcommand{\btheta}{\boldsymbol \theta}

\newcommand{\bvarepsilon}{\boldsymbol \varepsilon}

\newcommand{\bepsilon}{\boldsymbol \epsilon}
\newcommand{\bbbeta}{{\boldsymbol \eta}}

\newcommand{\bbf}{{\mathbf f}}
\newcommand{\bx}{{\mathbf x}}

\newcommand{\bv}{{\mathbf v}}

\newcommand{\br}{{\mathbf r}}

\newcommand{\bg}{{\mathbf g}}

\newcommand{\bbb}{{\mathbf b}}

\newcommand{\bG}{{\bf G}}
\newcommand{\bW}{{\bf W}}
\newcommand{\bD}{{\bf D}}
\newcommand{\bA}{{\bf A}}
\newcommand{\bB}{{\bf B}}
\newcommand{\bC}{{\bf C}}

\newcommand{\bI}{{\bf I}}
\newcommand{\bK}{{\bf K}}

\newcommand{\bT}{{\bf T}}
\newcommand{\bX}{{\bf X}}
\newcommand{\bY}{{\bf Y}}
\newcommand{\bZ}{{\bf Z}}

\newcommand{\bR}{{\bf R}}

\newcommand{\bQ}{{\bf Q}}

\newcommand{\bUps}{{\boldsymbol{\Upsilon}}}

\newcommand{\boZ}{{\boldsymbol{Z}}}

\newcommand{\cL}{{\cal L}}

\newcommand{\cU}{{\cal U}}
\newcommand{\cV}{{\cal V}}
\newcommand{\cS}{{\cal S}}

\newcommand{\cN}{{\cal N}}

\newcommand{\eR}{\mathbb{R}}

\newcommand{\eE}{\mathbb{E}}

\newcommand{\tF}{\text{F}}

\newcommand{\tdiag}{\text{diag}}
\newcommand{\tcorr}{\text{Corr}}
\newcommand{\tswap}{\text{swap}}

\newcommand{\cov}{\text{Cov}}
\newcommand{\var}{\text{Var}}

\newcommand{\tr}{\mbox{tr}}
\newcommand{\diag}{\mbox{diag}}
\newcommand{\FDR}{\text{FDR}}

\newcommand{\mFDR}{\text{mFDR}}
\newcommand{\Pow}{\text{Power}}
\newcommand{\tAnd}{\text{A}}
\newcommand{\tOr}{\text{O}}

\def\T{{ \mathrm{\scriptscriptstyle T} }}

\makeatletter
\newcommand*{\rom}[1]{\expandafter\@slowromancap\romannumeral #1@}
\makeatother

\makeatletter
\DeclareRobustCommand\widecheck[1]{{\mathpalette\@widecheck{#1}}}
\def\@widecheck#1#2{%
	\setbox\z@\hbox{\m@th$#1#2$}%
	\setbox\tw@\hbox{\m@th$#1%
		\widehat{%
			\vrule\@width\z@\@height\ht\z@
			\vrule\@height\z@\@width\wd\z@}$}%
	\dp\tw@-\ht\z@
	\@tempdima\ht\z@ \advance\@tempdima2\ht\tw@ \divide\@tempdima\thr@@
	\setbox\tw@\hbox{%
		\raise\@tempdima\hbox{\scalebox{1}[-1]{\lower\@tempdima\box
				\tw@}}}%
	{\ooalign{\box\tw@ \cr \box\z@}}}
\makeatother
\RequirePackage[OT1]{fontenc}
\usepackage[round,authoryear]{natbib}\RequirePackage[colorlinks,citecolor=blue,urlcolor=blue]{hyperref}

\providecommand{\keywords}[1]{\textbf{\textit{Keywords: }} #1}

\newcommand{\blind}{1}

\begin{document}

\def\spacingset#1{\renewcommand{\baselinestretch}%
{#1}\small\normalsize} \spacingset{1}

	\if1\blind
	{
  \spacingset{1.25}
		\title{\bf \Large
	Functional knockoffs selection with applications to functional data analysis in high dimensions
		}
            \author[1,2]{Xinghao Qiao}
            \author[1]{Mingya Long}
		\author[3]{Qizhai Li}
         \affil[1]{\it \small Department of Statistics, London School of Economics, London, U.K.}
  \affil[2]{\it \small Faculty of Business and Economics, The University of Hong Kong, Hong Kong}
       		\affil[3]{\it \small Academy of Mathematics and Systems Science CAS and University of Chinese Academy of Sciences, Beijing, China}
		\setcounter{Maxaffil}{0}
		
		\renewcommand\Affilfont{\itshape\small}
		\date{\vspace{-5ex}}
		\maketitle
	} \fi
	\if0\blind
	{\spacingset{2}
		\bigskip
		\bigskip
		\bigskip
		\begin{center}
			{\Large\bf Functional knockoffs selection with applications to functional data analysis in high dimensions}
		\end{center}
		\medskip
	} \fi

\spacingset{1.69}
\begin{abstract}
The knockoffs is a recently proposed powerful framework that effectively controls the false discovery rate (FDR) for variable selection. However, none of the existing knockoff solutions are directly suited to handle multivariate or high-dimensional functional data, which has become increasingly prevalent in various scientific applications. In this paper, we propose a novel functional model-X knockoffs selection framework tailored to sparse high-dimensional functional models, and show that our proposal can achieve the effective FDR control for any sample size. Furthermore, we illustrate the proposed functional model-X knockoffs selection procedure along with the associated theoretical guarantees for both FDR control and asymptotic power using examples of commonly adopted functional linear additive regression models and the functional graphical model. In the construction of functional knockoffs, we integrate essential components including the correlation operator matrix, the Karhunen-Lo\`eve expansion, and semidefinite programming, and develop executable algorithms. We demonstrate the superiority of our proposed methods over the competitors through both extensive simulations and the analysis of two brain imaging datasets.


\end{abstract}


\newpage
\spacingset{1.69} 
\vspace{-0.3cm}
\section{Introduction}
\label{sec:intro}
\vspace{-0.3cm}
Selecting important covariates associated with a response from a pool of potential candidates holds paramount importance across various scientific fields. At the same time, controlling the false discovery rate (FDR) offers an effective means to control error rates, ensuring replicable discoveries. A large body of literature on FDR control focuses on multiple testing approaches based on the $p$-values for assessing the significance of individual covariates; see, e.g., the seminal papers of \cite{benjamini1995,benjamini2001}. Yet, the high-dimensionality in covariates often renders many traditional approaches for $p$-value calculations inapplicable. Furthermore, none of the existing works along this vein directly address the problem of variable selection while simultaneously controlling the FDR.

The recent paper of \cite{barber2015controlling} introduced a fixed-X knockoffs inference framework that effectively controls the FDR for variable selection in Gaussian linear model with the dimensionality $p$ no larger than the sample size $n$ under fixed designs. The key idea is to construct knockoff variables that mimic the dependence structure of original covariates while maintaining independence from the response conditional on the original covariates. It then compares the importance statistics (e.g., lasso coefficients) between the original covariates and knockoffs for variable selection. 
The fixed-X knockoffs inference has been extended to many settings, such as group-variable selection and multitask learning \cite[]{dai2016knockoff}, high-dimensional linear model using data-splitting and feature screening \cite[]{barber2019knockoff} and Gaussian graphical model through a node-based local and a graph-based global procedure \cite[]{li2021ggm}.

More recently, \cite{candes2018} proposed a model-X knockoffs extension that accommodates random design and allows for arbitrary and unknown conditional distribution of the response given the covariates, and for arbitrarily large $p$ compared to $n.$
The model-X knockoffs framework has witnessed a plethora of advancements. For instance, \cite{fan2020rank} developed a graphical nonlinear knockoffs method to handle the unknown covariate distribution, providing theoretical guarantees on the power and robustness.
\cite{Fan2020IPAD} applied knockoffs inference to high-dimensional latent factor models, enabling stable and intepretable forecasting.
\cite{dai2022kernel} introduced a kernel knockoffs procedure for nonparametric additive models, employing subsampling and random feature mapping. See also \cite{Romano2020deep,ren2023} and the references therein. Nevertheless, the existing efforts are primarily devoted to addressing scalar data. As a result, it remains unclear whether the model-X knockoffs framework is applicable to functional data.

The rapid development of data collection technology has led to an increased availability of multivariate or high-dimensional functional data datasets. Examples include time-course gene expression data and different types of neuroimaging data, where brain signals are measured over time at a multitude of regions of interest (ROIs). 
Building upon recent proposals \citep[e.g.,][]{zhu2016,li2018nonparametric,fang2022}, these signals are modelled as multivariate random functions with each ROI represented by a random function, where, under high-dimensional scaling, the number of functional variables $p$ 
can be comparable to, or even exceed, the number of subjects $n.$ 
To overcome the difficulties caused by high-dimensionality, various functional sparsity assumptions are commonly imposed on the model parameter space. E.g., scalar-on-function linear additive regression (SFLR) \cite[]{fan2015,kong2016,xue2021}, function-on-function linear additive regression (FFLR) \cite[]{fan2014SFLR,luo2017,chang2023}, functional linear discriminant analysis \cite[]{xue2023opti} as well as functional Gaussian graphical model (FGGM) \cite[]{qiao2019a} and its various extensions \cite[]{solea2022copula,zapata2019,Lee2023Conditional}. These models involve the development of sparse function-valued estimates in the sense of selecting important functional variables. However, none of the existing work has achieved the essential task of FDR control.

The major goal of our paper is to establish a methodological and theoretical foundation of model-X knockoffs selection for functional data and apply it to concrete examples of sparse high-dimensional functional models, thereby bridging an important gap in the respective fields. Specifically, we propose a functional model-X knockoffs selection framework that begins with dimension reduction via functional principal components analysis (FPCA), thus effectively converting infinite-dimensional curves into vector-valued FPC scores. We then compare group-lasso-based importance statistics between the estimated FPC scores of original and knockoff variables for variable selection. We demonstrate that our proposal is guaranteed to control the FDR below the nominal level regardless of $n.$ We then showcase the proposed framework through three useful examples, i.e., SFLR, FFLR and FGGM, and, additionally, establish that the power for each model asymptotically approaches one as $n$ goes to infinity. In constructing functional knockoffs, we integrate key ingredients: the correlation operator matrix, the Karhunen-Lo\`eve expansion, and semidefinite programming. We also develop executable algorithms using a coordinate representation system within finite-dimensional Hilbert spaces. 


The main contribution of our paper is threefold. First, we propose a general functional model-X knockoffs selection framework. Arising from the initial dimension reduction, we have to deal with estimated FPC scores and truncation errors, whereas the conventional knockoffs is applied directly to observed data. Accounting for both estimation and truncation errors is a major undertaking in theoretical analysis. We show that the estimated vector-valued FPC scores possess two crucial properties: exchangeability and coin-flipping, which are pivotal in ensuring the effective FDR control of our proposal. By comparison, \cite{dai2016knockoff} developed a ``truncation first'' strategy that constructs group-knockoffs based on truncated FPC scores followed by group-lasso for variable selection. However, their FDR is limited to staying below the target level within the fixed-X rather than model-X knockoffs framework. Regarding the statistical power, our ``knockoffs first" proposal constructs functional knockoffs before dimension reduction and group-variable selection, thus capturing more feature information and leading to improved power, as evidenced by  simulation results in Section~\ref{sec:sim}.


Second, we apply our proposal to two sparse functional linear additive regression models, i.e., SFLR and FFLR, effectively achieving the FDR controls. We also delve into the associated power properties, which pose greater challenges compared to non-functional models due to the inclusion of aforementioned additional errors. Furthermore, we integrate the fixed-X GGM knockoff filter \cite[]{li2021ggm} into our functional model-X knockoffs framework to accommodate FGGM. Specifically, our proposal begins by locally constructing functional knockoffs and group-lasso coefficients for each node, and then solves a global optimization problem to determine nodewise thresholds for graph estimation. Compared to \cite{qiao2019a}, we circumvent the estimation of unbounded inverse covariance functions by reformulating the graph estimation through penalized functional regressions on each functional variable against the remaining ones. Unlike \cite{li2021ggm} which lacks theoretical power analysis,  we establish the power guarantee for the more challenging task of functional model-X selection for FGGM in a high-dimensional regime.

Third, constructing knockoff variables is a pivotal step in implementing the knockoffs procedure, and our functional extension is far from incremental. One challenge lies in characterizing the dependence across infinite-dimensional objects and choosing suitable functional norm to quantify the strength of dependence. The natural functional extension, which seeks to minimize the average covariance operators in certain norms between the functional original and knockoff variables, is inappropriate, since the minimum eigenvalues of covariance operators converge to zero, thus making it exceedingly difficult to distinguish the original variables from knockoff counterparts. Motivated by the result that, under mild conditions, all eigenvalues of associated correlation operators remain bounded away from zero and infinity, we propose to minimize the average correlation operators in operator norm as opposed to other unbounded norms, which largely enhances the distinguishability and ensures good power in signal detection. The other challenges are to solve semidefinite programming problems and specify conditional distributions at the operator level with the aid of Karhunen-Lo$\grave{\hbox{e}}$ve expansions. To develop executable algorithms, we reformulate their sample versions by representing operators as matrices using the coordinate mapping.

Our paper is set out as follows. Section~\ref{sec:method} proposes the functional model-X knockoffs selection framework. Section~\ref{sec:app} applies the proposal to three examples, i.e., SFLR, FFLR and FGGM, and establishes the associated theoretical guarantees on the FDR and power.
Section~\ref{sec:cons} presents the construction of functional knockoffs with executable algorithms. We demonstrate the superior finite-sample performance of our methods through simulations in Section~\ref{sec:sim} and the analysis of two brain imaging datasets in Section~\ref{sec:real}. 
All technical proofs are relegated to the Supplementary Material.

{\bf Notation}. For a positive integer $p,$ denote $[p]=\{1, \dots, p\}$ and $\bI_p$ as $p \times p$ identity matrix.  
For any vector $\bbb = (b_{1},\dots,b_p)^\T,$ define $\|\bbb\| = (\sum_{j} b^2_{j})^{1/2}.$
For any matrix $\bB = (B_{ij})_{p\times q}$, denote $\|\bB\|_{\tF} = (\sum_{i,j} B^2_{ij})^{1/2}$ its Frobenius norm and $\bB^\dag$ its Moore-Penrose inverse.  
Let $L_2(\cU)$ be a Hilbert space of square-integrable functions on a compact interval $\cU$ with inner product $\langle f, g\rangle =  \int_\cU f(u) g(u) du$ and norm  $\|\cdot\| = \langle \cdot , \cdot \rangle^{1/2}$ for $f,g \in L_2 (\cU).$ 
For $j\in[p]$, we take a separable Hilbert space 
$\mathcal{H}_j \subseteq L_2 (\cU)$.  
For a compact linear operator $K$ from $\mathcal{H}_j$ to $\mathcal{H}_k$ induced from the kernel function $K$ with $K(f)(u)=\int_{\cU}K(u,v)f(v){\rm d}v \in \mathcal{H}_k$ for $f \in \mathcal{H}_j,$ there exist two orthonormal bases $\{\phi_l\}$ and $\{\psi_l\}$ of $\mathcal{H}_j$ and $\mathcal{H}_k,$, respectively, and a sequence $\{\lambda_l\}$ in $\mathbb R$ tending to zero, such that $K$ has the spectral decomposition $K=\sum_{l=1}^{\infty} \lambda_l \phi_l \otimes \psi_l,$ 
where $\otimes$ denotes the tensor product. For notational economy, we will use $K$ to 
denote both the operator and kernel function. We denote its Hilbert--Schmidt norm by $\|K \|_{\cS} = (\sum_{l=1}^{\infty} \lambda_{l}^2)^{1/2}=\{\int\int K^2(u,v) {\rm d}u {\rm d}v\}^{1/2},$ nuclear norm by $\| K\|_{\cN} = \sum_{l=1}^{\infty} |\lambda_{r}| $ and operator norm by $\|K\|_{\cL} = \sup_{\|f\|\leq 1,f \in \mathcal{H}_j} \|K(f)\|.$
Let $\boldsymbol{\mathcal{H}}$ be the Cartesian product of $\mathcal{H}_1, \dots, \mathcal{H}_p$ with inner product $\langle \bbf, \bg \rangle = \sum _{j=1}^{p}\langle f_j,g_j\rangle$ for $\bbf  = (f_1,\dots,f_p)^{\T}$ and $\bg= (g_1,\dots,g_p)^{\T} \in \boldsymbol{\mathcal{H}}.$
An operator matrix $\bK= (K_{jk})$ is a $p \times p$ matrix of operators with its $(j,k)$th operator-valued entry $K_{jk},$ and can be thought of an operator from $ \boldsymbol{\mathcal{H}}$ to $\boldsymbol{\mathcal{H}}$ with $\bK(\bbf) = \big(\sum_{j=1}^{p} K_{1j}( f_{j}),\dots, \sum_{j=1}^{p} K_{pj} (f_{j})\big)^\T$ for $\bbf \in \boldsymbol{\mathcal{H}}.$ 
We use $\bK \succeq 0$ to denote a positive semidefinite operator matrix satisfying $\langle\bK(\bbf), \bbf\rangle \geq 0$ for any $\bbf \in \boldsymbol{\mathcal{H}}$. For $a,b \in \eR$, we use $a\vee b = \max\{a,b\}$. For two sequences of positive numbers $\{a_n\}$ and $\{b_n\}$, we write $a_n \lesssim b_n$ or $b_n \gtrsim a_n$ if there exists some constant $c>0$ such that $a_n \leq c b_n$. 

\vspace{-0.3cm}
\section{Functional model-X knockoffs selection framework}
\label{sec:method}
\vspace{-0.3cm}
\subsection{Definition} 
\vspace{-0.3cm}
\label{sec:def}
Let $\bX(\cdot) = \big(X_1(\cdot),\dots, X_p(\cdot)\big)^{\T}$ be a random element in $\boldsymbol{\mathcal{H}}.$ Before framing the variable selection in the context of sparse functional models with $Y$ being a scalar or functional response, we define a functional covariate $X_j(\cdot)$ as null if and only if $Y$ is independent of $X_j(\cdot)$ conditional on the remaining functional covariates $X_{-j}(\cdot)=\{X_1(\cdot), \dots, X_p(\cdot)\} \setminus \{X_j(\cdot)\}$ and as nonnull otherwise. Let $S^c$ denote the index set of null functional covariates, i.e., 
\begin{equation} 
\label{f.nullset}
    S^c = \big\{j\in[p]: X_{j}(\cdot) \hbox{ is independent of } Y \hbox{ conditional on } X_{-j}(\cdot)\big\}, 
\end{equation}
and hence the index set of nonnull functional covariates is given by $S,$ the complement of $S^c.$ This formulation naturally establishes the equivalence between the selection of null functional covariates and functional variable selection. 
E.g., it follows from Lemmas~\ref{lemmaeqSFLR} and \ref{lemmaeqFFLR} of the Supplementary Material that $S^c$ is the same as the set $\{j\in [p]: \|\beta_j\|=0\}$ in (\ref{SFLR.eq}) for SFLR or $\{j\in [p]:\|\beta_j\|_{\cS}=0\}$ in (\ref{FFLR.eq}) for FFLR. 
Our goal is to discover as many nonnull functional covariates as possible while controlling the FDR, defined as
\begin{equation*}
\label{df.FDR}
\FDR = \eE\left[\frac{|\widehat{S}\cap S^c|}{|\widehat{S}|\vee 1}\right],
\end{equation*}
where $\widehat{S}$ represents the index set of functional covariates identified by the variable selection procedure, and $|\cdot|$ denotes the cardinality of a set.


The key ingredient of functional model-X knockoffs selection framework is the construction of functional model-X knockoffs, which is defined as follows. 
\begin{definition}\label{def1}
Functional model-X knockoffs for the family of random functions $\bX(\cdot)$ are a new family of random functions $\widetilde{\bX}(\cdot) = \big(\widetilde{X}_1(\cdot), \dots, \widetilde{X}_p (\cdot)\big)^{\T} \in {\boldsymbol{\mathcal{H}}}$ that satisfies the following two properties:
(i) 
$ \big(\bX(\cdot)^\T,\widetilde{\bX}(\cdot)^\T\big)_{\tswap (G)} \overset{D}{=}  \big(\bX(\cdot)^\T, \widetilde{\bX}(\cdot)^\T\big)$ 
for any subset $G \subseteq [p],$ where $\tswap(G)$ means swapping components $X_j(\cdot)$ and $\widetilde X_j(\cdot)$ for each $j \in G$ and $\overset{D}{=}$  denotes the equality in distribution.
(ii) $\widetilde{\bX}(\cdot)$ and $Y$ are independent conditionally on $\bX(\cdot).$
\end{definition}

Definition~\ref{def1} generalizes the definition of model-X knockoffs \cite[]{candes2018} within Hilbert spaces.
Property~(ii) is fulfilled when constructing the knockoffs $\widetilde\bX(\cdot)$ without any reference to $Y,$ and Property~(i) corresponds to the pairwise exchangability between the original and knockoff variables. 
Before giving an example obeying this property, we introduce some notation. For each $j,k \in [p],$ denote the mean of $X_j$ as $\mu_{j}$ and the covariance operator between $X_{j}$ and $X_{k}$ as $\Sigma_{X_{j}X_{k}} =  \cov(X_j, X_k) = {\mathbb E}\{(X_{j}-\mu_{j})\otimes(X_{k}-\mu_{k})\},$ which has a one-to-one correspondence with the covariance function $\Sigma_{X_jX_k}(u,v)=\cov\{X_j(u),X_k(v)\}$ for $(u,v) \in \cU^2.$ Denote the $p\times p$ covariance operator matrix of $\bX$ as $\bSigma_{XX},$ whose $(j,k)$th entry is operator $\Sigma_{X_{j}X_{k}}.$ The cross-covariance operator matrix $\bSigma_{X\widetilde X}$ between $\bX$ and $\widetilde \bX$ can be defined similarly. We will utilize the example below as the way of constructing functional knockoffs in Section~\ref{sec:cons}. 
\begin{example}
\label{ex.MGP}
Suppose that ${\bX}$ follows a multivariate Gaussian process (MGP) with mean zero and covariance $\bSigma_{XX},$ denoted as $\text{MGP}(\bzero,\bSigma_{XX}).$ Then $(\bX^\T,\widetilde{\bX}^\T)^\T\sim\text{MGP}(\bzero,\bSigma_{(X,\widetilde{X})(X,\widetilde{X})})$ satisfies Property~(i),
where 
\begin{equation}
\setlength{\abovedisplayskip}{5pt}
    \setlength{\belowdisplayskip}{5pt}
\label{Q.cov}
\bSigma_{XX} = \bSigma_{\widetilde X\widetilde X},~~\bSigma_{X\widetilde X} = \bSigma_{\widetilde X X} = \bSigma_{XX}- \bQ_{XX},
\end{equation}
and $\bQ_{XX} = \tdiag(Q_{{X}_1 X_1},\dots, Q_{{X}_p X_p})$ is selected in such a way that $\bSigma_{(X,\widetilde{X})(X,\widetilde{X})}  \succeq 0.$
\end{example}

Suppose that  we observe $n$ i.i.d. realizations $\{\bX_i(\cdot),Y_i\}_{i \in [n]}$ from the population $\{\bX(\cdot),Y\}.$ Due to the infinite-dimensionality of functional data, we adopt the standard dimension reduction approach by performing Karhunen-Lo$\grave{\hbox{e}}$ve expansions of $X_{ij}(\cdot)$ and $\widetilde X_{ij}(\cdot)$ for each $j$ and truncating the expansions to the first $d_j$ terms, which serves as the foundation of FPCA:
\begin{equation}
\setlength{\abovedisplayskip}{8pt}
    \setlength{\belowdisplayskip}{8pt}
\label{exX}
{X}_{ij} (\cdot) -\mu_j(\cdot)= \sum_{l=1}^{\infty} \xi_{ijl} \phi_{jl}(\cdot) \approx  \bxi_{ij}^\T \bphi_{j}(\cdot),~~
\widetilde {X}_{ij} (\cdot) - \mu_j(\cdot) = \sum_{l=1}^{\infty} \tilde \xi_{ijl} \phi_{jl}(\cdot) \approx  \widetilde \bxi_{ij}^\T \bphi_{j}(\cdot),  
\end{equation}
where $\bphi_{j} = (\phi_{j1},\dots, \phi_{jd_j})^\T$, $\bxi_{ij} = (\xi_{ij1},\dots,\xi_{ijd_j})^\T$ and $\widetilde\bxi_{ij} = (\tilde\xi_{ij1},\dots,\tilde\xi_{ijd_j})^\T.$ Here  $\xi_{ijl}=\langle X_{ij} -\mu_j, \phi_{jl} \rangle$ (or $\tilde \xi_{ijl}=\langle \widetilde X_{ij} -\mu_j, \phi_{jl} \rangle$), namely FPC score of original (or knockoff) variables, corresponds to a sequence of random variables with ${\mathbb E}(\xi_{ijl})=0$ and $\cov(\xi_{ijl}, \xi_{ijl'})=\omega_{jl}I(l=l'),$ where $\omega_{j1} \geq \omega_{j2} \geq \dots > 0$ are the eigenvalues of $\Sigma_{X_jX_j}$ and $\phi_{j1}(\cdot), \phi_{j2}(\cdot), \dots$ are the corresponding eigenfunctions. To implement FPCA based on $n$ observations, we compute the sample estimator of $\Sigma_{X_jX_j}$ via $\widehat\Sigma_{X_jX_j} = n^{-1} \sum_{i=1}^n (X_{ij} -\hat \mu_j) \otimes (X_{ij} -\hat \mu_j)$ with $\hat \mu_j =n^{-1}\sum_{i=1}^n X_{ij}$ and then carry out an eigenanalysis of $\widehat\Sigma_{X_jX_j}$ that leads to estimated eigenvalue/eigenvector pairs $\{\hat\omega_{jl},\hat \phi_{jl}(\cdot)\}_{l \in [d_j]}.$ We then obtain estimated FPC scores $\hat \xi_{ijl} =\langle X_{ij} - \hat \mu_j, \hat \phi_{jl} \rangle$  and $\check \xi_{ijl} =\langle \widetilde X_{ij} - \hat \mu_j, \hat \phi_{jl} \rangle$ for $l \in [d_j].$
Let $\widehat\bxi_{ij}=(\hat\xi_{ij1}, \dots, \hat\xi_{ijd_j})^\T,$ $\widecheck \bxi_{ij}=(\check\xi_{ij1}, \dots, \check\xi_{ijd_j})^\T,$ $\widehat\bxi_{i}=(\widehat\bxi_{i1}^\T, \dots, \widehat\bxi_{ip}^\T)^\T$ and $\widecheck\bxi_{i}=(\widecheck\bxi_{i1}^\T, \dots, \widecheck\bxi_{ip}^\T)^\T.$ Resulting from the dimension reduction, the estimation of sparse function-valued parameters based on $\{\bX_i(\cdot)^\T,Y_i\}_{i\in [n]}$ is transformed to the block sparse estimation of parameter vectors/matrices based on vector-valued estimated FPC scores $\{\widehat\bxi_i\}_{i \in [n]}$ and transformed responses $\{\widetilde Y_i\}_{i \in [n]},$ where, e.g., $\widetilde Y_i$ equals $Y_i$ for SFLR in Section~\ref{sec:sflr} and estimated FPC scores of $Y_i(\cdot)$ for FFLR in Section~\ref{sec:fflr}.

We next present the exchangeability condition, i.e., swapping estimated FPC scores of null functional covariates with those of corresponding functional knockoffs will not affect the joint distribution of $\widehat \bxi_i$ and $\widecheck \bxi_i$ conditional on $\widetilde Y_i.$

\begin{condition}\label{conscore}
    For any subset $G  \subseteq S^{c}$ and $i\in[n]$,  $\big({\widehat{\bxi}^\T_i}, {\widecheck{{\bxi}}^\T_i} \big)|   \widetilde{Y}_i   \overset{D}{=} \big({\widehat{\bxi}^\T_i}, {\widecheck{{\bxi}}^\T_i} \big)_{\text{swap}(G)}|   \widetilde{Y}_i $.
\end{condition}
This condition can be validated across three examples we consider by applying the functional exachangeability result in Lemma~\ref{lemmaEx.func} of the Supplementary Material, which is built upon the properties in Definition \ref{def1}. It plays a crucial role in establishing the coin-flipping property in Lemma~\ref{lemmacoin} below. 

\vspace{-0.3cm}
 \subsection{Feature statistics} 
 \label{subsec23}
 \vspace{-0.3cm}
To select the nonnull functional variables, we compute knockoff statistics $W_j = w_j(Z_j,\widetilde{Z}_j)$ for each $j \in [p]$, where $w_j$ is antisymmetric function satisfying $w_j(\cdot,*) = -w_j(*,\cdot)$, and $Z_j$ and $\widetilde{Z}_j$ respectively represent the feature importance measure of estimated FPC scores of $\{X_{ij}\}_{i \in [n]}$ and $\{\widetilde X_{ij}\}_{i \in [n]}.$ For three examples we consider, we can choose $Z_j$ and $\widetilde Z_j$ as the group-lasso coefficient vectors/matrices under the vector $\ell_2$/matrix Frobenius norm of $\{\widehat \bxi_{ij}\}_{i \in [n]}$ and $\{\widecheck \bxi_{ij}\}_{i \in [n]},$ respectively, and a valid knockoff statistic is $w_j(Z_j,\widetilde Z_j)=Z_j-\widetilde Z_j.$ Intuitively, a large positive value of $W_j$ suggests that $X_j(\cdot)$ is an important (i.e., nonnull) feature, while small magnitudes of $W_j$ often correspond to unimportant (i.e., null) features. As noted in the \cite{candes2018}, the knockoffs selection can control the FDR when the feature importance measures $W_j$'s possess the essential coin flipping property below. 

\begin{lemma}
\label{lemmacoin} 
Suppose that Condition~\ref{conscore} holds. Let $(\delta_1, \dots, \delta_p)$ be a sequence of independent random variables such that $\delta_j = \pm 1$ with a probability of $1/2$ if $j \in S^{c}$, and $\delta_j =  1$ otherwise. Then $ (W_1,\dots,W_p) \overset{D}{=}  (\delta_1 W_1,\dots, \delta_p W_p),$ conditional on $(|W_1|,\dots,|W_p|).$ 
\end{lemma}

Hence a large positive value of $W_j$  provides evidence that $j \in S,$ whereas, under $j \in S^c$, $W_j$ is equally likely to be positive and negative. Notice that Lemma~\ref{lemmacoin} is established on the sets of null/nonnull functional covariates, and fails to hold for the corresponding truncated sets specified in (\ref{t.nullset}) below. Nevertheless, this does not place a constraint to control the FDR, as justified in Theorem~\ref{thfdr} below.

The last step of our functional knockoffs selection framework is to apply the knockoff filter 
\cite[]{candes2018} by ranking $W_j$'s from large to small and selecting features whose associated $W_j$'s are at least some threshold $T_{\delta}$ in (\ref{threshold}). This results in estimated nonnull sets 
\begin{equation}\label{est.null}
\widehat{S}_\delta = \{j\in[p]:W_j \geq T_\delta\},~~~\delta=0~\text{or}~1.
\end{equation}
To select a data-driven threshold as permissive as possible while still managing the control over the FDR, we choose the threshold in the following two ways:
\begin{equation} \label{threshold}
T_{\delta} = \min\Big \{t \in \big\{|W_j| > 0:j\in [p] \big\}: \frac{\delta +|\{j:W_j \leq -t\}|}{|\{j:W_j \geq t\}|} \leq q \Big\},
\end{equation}
where $T_0$ is used for knockoff filter and $T_1$ is used for more conservative knockoff+ filter. The false discovery is measured by both FDR based on $T_1$ and modified FDR based on $T_0,$ which are respectively defined as,
\begin{equation*}\label{df.mFDR}
\FDR ={\mathbb E}\left[\frac{|\widehat{S}_1 \cap  S^c |}{|\widehat{S}_1|\vee 1}\right]  ~~\text{and}~~ 
\mFDR ={\mathbb E}\left[\frac{| \widehat{S}_0\cap  S^c  |}{|\widehat{S}_0|+1/q}\right].
 \end{equation*}
We are now ready to present a theorem regarding the effective FDR control.

\begin{theorem}
\label{thfdr}
Suppose that Condition~\ref{conscore} holds. For any sample size $n$ and target FDR level $q\in [0,1]$, the selected set $\widehat S_1$ satisfies $\FDR \leq q$ and the selected set $\widehat S_0$ satisfies $\mFDR \leq q.$.
\end{theorem} 

\begin{remark}\label{errorRmK}
Resulting from the dimension reduction step, we are confronted with estimated FPC scores and bias terms formed by truncation errors. However, Theorem~\ref{thfdr} makes it evident that neither estimation errors nor truncation errors affect the effectiveness of FDR control, which remains valid regardless of truncated dimensions $d_j$'s. This is because Condition~\ref{conscore}, concerning estimated FPC scores, serves as the foundation for proving Theorem~\ref{thfdr} and can be verified for any values of $d_j$'s without regard to truncation errors.
\end{remark}

\begin{remark}\label{fdrRmK}
Our proposed ``knockoffs first" framework begins with constructing functional knockoffs before performing dimension reduction and knockoff filter to the group-lasso coefficients for variable selection. Theorem~\ref{thfdr} ensures that our proposal effectively controls FDR when the null set $S^c$ is defined in (\ref{f.nullset}).
By comparison, an alternative ``truncation first" approach applies fixed-X group knockoffs selection \cite[]{dai2016knockoff} by constructing vector-valued knockoffs based on vector-valued estimated FPC scores, followed by adopting a group-lasso-based knockoff filter for variable selection.
However, it is important to note that the presence of functional versions of the conditional independence in both the null set $S^c$ and Definition~\ref{def1}'s Property~(ii) does not imply the corresponding truncated (sample) versions of the conditional independence, which are essential to ensure FDR control within the model-X framework. We thus formulate the ``truncation first" strategy within the fixed-X framework to accommodate linear model settings with $n \geq \sum_{j=1}^p d_j$. For instance, in Section~\ref{sec:sflr} for SFLR, we define the corresponding null set $S^c$ as the complement of $S$ in (\ref{support.SFLR}), and the truncated version of the null set $\widetilde S^c$ in (\ref{t.nullset}) with $S^c \subseteq \widetilde S^c$.  This allows us to employ the group knockoffs selection \cite[]{dai2016knockoff} that results in the effective FDR control via 
\begin{equation}
\label{fdr.ineq}
\underbrace{{\mathbb E}\left[\frac{|  \widehat{S} \cap S^c |}{|\widehat{S}|\vee 1}\right]}_{\text{functional version of FDR}} \leq \underbrace{{\mathbb E}\left[\frac{|  \widehat{S} \cap \widetilde S^c |}{|\widehat{S}|\vee 1}\right]}_{\text{truncated version of FDR}} \leq q, 
\end{equation}
where $\widehat S$ denotes the set of selected variables.
\end{remark}

Compared to our ``functional knockoffs" proposal with both FDR and power guarantees, it is evident from Remarks~\ref{fdrRmK} and \ref{powerRmK} below that the ``truncation first" strategy ensures FDR control but results in empirical power with possibly slower asymptotic rate of convergence for SFLR. These findings also hold true for FFLR and FGGM, and are consistent with our simulation results in Section~\ref{sec:sim}.

\vspace{-0.3cm}
\section{Applications} 
\label{sec:app}
 \vspace{-0.3cm}

\subsection{High-dimensional SFLR}\label{sec:sflr}
\vspace{-0.3cm}
Consider the high-dimensional SFLR model
\begin{equation}\label{SFLR.eq}
\setlength{\abovedisplayskip}{5pt}
    \setlength{\belowdisplayskip}{5pt}
\begin{array}{cll}
Y_i &=& \sum _{j=1}^{p} \int_{\cU} \beta_{j}(u) X_{ij}(u) du + \varepsilon_i, ~~i\in [n] ,
\end{array}
\end{equation}
where $\{\varepsilon_i\}_{i \in [n]} $ are i.i.d. mean-zero random errors, independent of mean-zero $\{X_{ij}(\cdot)\}_{i\in [n], j\in [p]}.$ The function-valued coefficients $\{\beta_j(\cdot)\}_{j \in [p]}$ are to be estimated and are assumed to be functional $s$-sparse with support
\begin{equation}
\setlength{\abovedisplayskip}{5pt}
    \setlength{\belowdisplayskip}{5pt}
    \label{support.SFLR}
S = \big\{j \in [p] : \|\beta_j\| \neq 0\big\}
\end{equation}
and cardinality $s=|S| \ll p.$ Our target is to select important functional variables (i.e., recovery of support $S$) while simultaneously controlling FDR. To integrate this task into our functional model-X knockoffs selection framework, we give the following condition.

\begin{condition}\label{irrSFLR}
For any $j \in [p]$ and any bivariate functions $\gamma_{k} \in L_2(\cU) \otimes L_2(\cU)$ with $k \in [p] \setminus \{j\},$
$X_{ij}(u) \neq \sum_{k \neq j} \int_\cU \gamma_{k}(u,v)X_{ik}(v)dv$ for $i \in [n].$
\end{condition}

Under Condition \ref{irrSFLR}, Lemma \ref{lemmaeqSFLR} of the Supplementary Material establishes the equivalence between the nonnull set (\ref{f.nullset}) and the support set (\ref{support.SFLR}) for SFLR. 
For each $j\in[p],$ we expand $X_{ij}(\cdot)$ according to (\ref{exX}).
Some specific calculations lead to the representation of (\ref{SFLR.eq}) as
\begin{equation}
\setlength{\abovedisplayskip}{8pt}
    \setlength{\belowdisplayskip}{8pt}
\label{Yscore}
Y_i=\sum_{j=1}^p \bxi_{ij}^\T \bbb_j + \epsilon_i + \varepsilon_i,
\end{equation}
where the $j$th coefficient vector $\bbb_j =\int_{\cU}\bphi_j(u)\beta_j(u){\rm d}u \in {\mathbb R}^{d_j}$ and $\epsilon_i = \sum_{j=1}^p \sum_{l=d_j+1}^{\infty} \xi_{ijl} b_{jl}$ is the truncation error. Hence, we can rely on the group sparsity pattern in $\{\bbb_j\}_{j \in [p]}$ to recover the functional sparsity structure in $\{\beta_j(\cdot)\}_{j \in [p]}.$ Within functional knockoffs framework, we denote the augmented coefficient vectors of FPC scores by $\{\bbb_j\}_{j \in [2p]}$ (the first $p$ coefficient vectors are for the original covariates and the last $p$ are for the knockoffs).

We initiate by adopting FPCA on $\{X_{ij}(\cdot)\}_{i \in [n]}$ for each $j$ and obtain vector-valued estimated FPC scores of original and knockoff covariates (i.e., $\widehat \bxi_{ij}$'s and $\widecheck \bxi_{ij}$'s). We then estimate $\{\bbb_j\}_{j \in [2p]}$ via the group-lasso regression on the augmented set of estimated FPC scores
\begin{equation}
\setlength{\abovedisplayskip}{8pt}
    \setlength{\belowdisplayskip}{8pt}
\label{glasso.SFLR}
\min _{\bbb_1,\dots,\bbb_{2p}} \frac{1}{2n} \sum_{i=1}^n \big(Y_i - \sum_{j=1}^p \widehat{\bxi}_{ij}^\T \bbb_j- \sum_{j=p+1}^{2p} \widecheck{\bxi}_{i(j-p)}^\T \bbb_j \big)^2 + \lambda_n \sum_{j=1}^{2p} \|\bbb_j\|,
\end{equation}
where $\lambda_n \geq 0$ is the regularization parameter. Denote the solution to (\ref{glasso.SFLR}) by $\{\widehat \bbb_j\}_{j \in [2p]}.$ Complying with the knockoff selection step in Section \ref{subsec23}, we choose the $j$th feature importance measures by $Z_j=\|\widehat \bbb_j\|$ and $\widetilde Z_j=\|\widehat\bbb_{j+p}\|$ and the corresponding knockoff statistics is $W_j=\|\widehat \bbb_j\|-\|\widehat\bbb_{j+p}\|.$ Hence we obtain the set of selected functional covariates $\widehat S_{\delta}$ by applying the knockoff filter to $\{W_j\}_{j \in [p]}$ in (\ref{est.null}), where the threshold $T_{\delta}$ is determined by (\ref{threshold}).

With the constructed functional model-X knockoffs satisfying Definition~\ref{def1}, we can show that the estimated FPC scores of 
original and knockoff variables fulfill Condition~\ref{conscore}. Then an application of Theorem~\ref{thfdr} leads to the following theorem, which achieves the valid FDR control for SFLR without any constraint on the dimensionality $p$ relative to the sample size $n.$ As such, our proposal works for both $p<n$ and $p>n$ scenarios.

\begin{theorem}
\label{thm_fdr_sflr}
Suppose that Condition~\ref{irrSFLR} holds. Then for any sample size $n$ and target FDR level $q\in [0,1],$  $\widehat S_1$ satisfies $\FDR \leq q$  and $\widehat S_0$ satisfies $\mFDR \leq q.$
\end{theorem}

\begin{remark}
\label{rmk.truncate}
As discussed in Remark~\ref{fdrRmK}, we formulate the ``truncation first" strategy for SFLR within the fixed-X group knockoffs framework \cite[]{dai2016knockoff}. Referring to (\ref{Yscore}), we can represent $Y_i$ linearly as 
$Y_i=\sum_{j=1}^p \widehat\bxi_{ij}^\T \bbb_j + e_i,$ where the error term $e_i$ encompasses truncation, estimation and random errors. Define the corresponding null set as
\begin{equation}
\setlength{\abovedisplayskip}{5pt}
    \setlength{\belowdisplayskip}{5pt}
\label{t.nullset}
\widetilde{S}^c = \{j\in[p]:\|\bbb_j\|=0\}.
\end{equation} 
Treating $\{\widehat \bxi_{ij}\}_{j \in [p]}$ as original group covariates, we
can follow \cite{dai2016knockoff} to construct group knockoffs, then choose group-lasso-based coefficient vectors under the $\ell_2$ norm as feature importance measures and finally apply the knockoff filter for group-variable selection. According to (\ref{fdr.ineq}), the selected set by the ``truncation first" approach achieves the FDR control. However, it may lead to declined power compared to our proposal, as argued in Remark~\ref{powerRmK}.  

\end{remark}

Before asymptotic power analysis of our approach, we impose some regularity conditions. 
\begin{condition}\label{cond_error_SFLR}
$\{\varepsilon_i\}_{i\in[n]}$ with finite variance $\sigma^2$ are i.i.d. sub-Gaussian variables, i.e., there exists some constant $c$ such that $\eE[e^{x\varepsilon_i}] \leq e^{c^2 \sigma^2 x^2/2}$ for any $x\in\eR.$ 
\end{condition}
\begin{condition} \label{cond_inf_SFLR}
For $(\bX, \widetilde{\bX})  \in \boldsymbol{\mathcal{H}}^2$, we denote a diagonal operator matrix by $\bD_{(X,\widetilde{X})(X,\widetilde{X})} = \tdiag(\Sigma_{X_1X_1},\dots, \Sigma_{X_pX_p},\Sigma_{\widetilde X_1\widetilde X_1},\dots,\Sigma_{\widetilde X_p \widetilde X_p} ).$ The infimum
    $$
    \setlength{\abovedisplayskip}{8pt}
    \setlength{\belowdisplayskip}{8pt}
    \underline{\mu} = \inf_{\bPhi \in \boldsymbol{\mathcal{H}}_0^2}\frac{\langle \bPhi,  \bSigma_{(X,\widetilde{X})(X,\widetilde{X})} (\bPhi)\rangle}{\langle \bPhi,  \bD_{(X,\widetilde{X})(X,\widetilde{X})} (\bPhi)\rangle}$$is bounded away from 0, where 
     $\boldsymbol{\mathcal{H}}_0^2 = \{ \bPhi \in \boldsymbol{\mathcal{H}}^2: \langle \bPhi,  \bD_{(X,\widetilde{X})(X,\widetilde{X})} (\bPhi)\rangle \in (0,\infty)\}$.
\end{condition}

\begin{condition} \label{cond_eigen_SFLR}
For each $j \in [p],$ $\omega_{j1} > \omega_{j2} > \dots >0,$ and $\max_{j\in [p]} \sum_{l= 1}^{\infty} \omega_{jl} = O(1)$. There exists some constant $\alpha >1$ such that $\omega_{j l} -  \omega_{j(l+1)} \gtrsim l^{-\alpha-1}$ for $l=1, \dots, \infty.$
\end{condition}

\begin{condition} \label{cond_coef_SFLR}
(i) For each $j \in S,$ there exists some constant $\tau > \alpha/2+1$ such that $|b_{j l }| \lesssim l^{-\tau}$ for $l\geq 1;$  
(ii) $\min_{j \in S} \|\bbb_j\| \geq \kappa_{n} d^\alpha\lambda_n/ \underline{\mu}$ for some slowly diverging sequence 
$\kappa_n \to \infty $ as $n \to \infty,$ and the regularization parameter $\lambda_n \gtrsim s[ d^{\alpha+2} \{\log(pd)/n\}^{1/2} + d^{1- \tau}];$  
(iii) There exists some constant $ c'_1 \in (2(qs)^{-1},1)$ such that $|S_2| \geq c'_1 s$ with $S_2 = \{j \in [p]: \|\bbb_j\| \gg { s^{1/2} d^\alpha \lambda_n}\}.$
\end{condition}

To simplify notation, we assume the same $d$ across $j \in [p]$ in the power analysis in Sections~\ref{sec:sflr}, \ref{sec:fflr} and \ref{sec:fggm}, but our theoretical results extend naturally to the more general setting where $d_j$'s are different. Condition~\ref{cond_inf_SFLR} can be interpreted as requiring the minimum eigenvalue of the correlation operator matrix of $(\bX^\T, \widetilde{\bX}^\T)^\T$ to be bounded away from zero. See similar conditions in \cite{fan2020rank,dai2022kernel} on the minimal eigenvalue of the corresponding covariance matrix, whose functional extension fails to hold as the infimum of the covariance operator matrix is zero.
Conditions~\ref{cond_eigen_SFLR} and \ref{cond_coef_SFLR}(i) are standard in functional regression literature \citep[e.g.,][]{kong2016} with parameter $\alpha$ capturing the tightness of gaps between adjacent eigenvalues and parameter $\tau$ controlling the level of smoothness in nonzero coefficient functions. 
Condition~\ref{cond_coef_SFLR}(ii) requires the $\ell_2$ norms of nonzero coefficient vectors exceed a certain threshold, which ensures that the selected set contains the majority of important variables. Given that our knockoffs selection is built upon the group lasso, its power is upper bounded by that of the group lasso, which approaches one as $n \rightarrow \infty$ under this condition. 
Specifically, in the proof of Theorem~\ref{thm_power_sflr}, we obtain that, with high probability and $\widehat{S}_{\text{\tiny{GL}}}^c = \{j \in [p]: \|\widehat{\bbb}_j\| = 0\},$
$|\widehat{S}_{\text{\tiny{GL}}}^c \cap S|\min_{j\in S} \|\bbb_j \| \leq \sum_{j\in (\widehat S_{\text{\tiny{GL}}}^c \cap S)} \|\bbb_j \| = \sum_{j\in (\widehat S_{\text{\tiny{GL}}}^c \cap S)} \|\bbb_j -\widehat{\bbb}_j\|\leq \sum_{j=1}^p \|\bbb_j -\widehat{\bbb}_j\|=O(s d^{\alpha}\lambda_n/\underline{\mu}).$
Then
Condition~\ref{cond_coef_SFLR}(ii) implies that $|\widehat S_{\text{\tiny{GL}}}^c \cap S|=O(s \kappa_n^{-1})$ and hence the group lasso exhibits asymptotic power one. Condition~\ref{cond_coef_SFLR}(iii) requires a large enough suitable subset of $S$ that contains relatively strong signals to attain high power. See similar conditions in \cite{Fan2020IPAD,dai2022kernel}. 

We are now ready to characterize the power of our proposal, which is defined as
\begin{equation}\label{defpower}
\setlength{\abovedisplayskip}{5pt}
    \setlength{\belowdisplayskip}{5pt}
    \Pow(\widehat{S}_{\delta})  = {\mathbb E}\left[ \frac{|\widehat{S}_{\delta}\cap S|}{| S| \vee 1}\right]. 
\end{equation}

\begin{theorem}
\label{thm_power_sflr}
Suppose that Conditions~\ref{irrSFLR}--\ref{cond_coef_SFLR} hold and 
$s d^\alpha \lambda_n \to 0.$ 
Then the selected sets $S_{\delta}$'s satisfy $\Pow(\widehat{S}_{\delta}) \rightarrow 1$ as $n \to \infty.$
\end{theorem}

Theorem~\ref{thm_power_sflr} establishes the asymptotic power guarantee for both scenarios of $p<n$ and $p>n,$ more specifically $n<p \lesssim e^{n/d^{4\alpha+4}}/d$ under a high-dimensional regime.
\begin{remark}\label{powerRmK} 
In the proof of Theorem~\ref{thm_power_sflr}, we show that, with high probability, 
\begin{equation}\label{low.pow.kf}
\setlength{\abovedisplayskip}{5pt}
    \setlength{\belowdisplayskip}{5pt}
    \frac{|\widehat{S}_{\delta}\cap S|}{|S|\vee 1} \geq 1-O(\kappa_n^{-1}).
\end{equation}
By comparison, we present a heuristic argument about the power of the ``truncation first" strategy, employing the same sets $S, \widetilde S$ and $\widehat S$ as specified in Remarks~\ref{fdrRmK}--\ref{rmk.truncate}. Under conditions similar to those in \cite{fan2020rank} within the fixed-X framework, it is expected that the 
$\text{truncated version of Power} = {\mathbb E}\left[\frac{|\widehat{S} \cap \widetilde S |}{|\widetilde S|\vee 1}\right] \rightarrow 1,$ which holds in the sense of $|\widehat S \cap \widetilde S| \geq |\widetilde S|\big\{1-O(\kappa_n^{-1})\big\}$ with high probability.  This, together with the fact $\widetilde S \to S$ as $d, n \rightarrow \infty,$
yields that the
$\text{functional version of Power} ={\mathbb E}\left[\frac{|\widehat{S} \cap S |}{|S|\vee 1}\right] \rightarrow 1,$
which holds since, with high probability, 
\begin{equation}\label{low.pow.ef} 
\setlength{\abovedisplayskip}{10pt}
    \setlength{\belowdisplayskip}{10pt}
    \frac{|\widehat{S} \cap S |}{|S|\vee 1}  \geq \frac{|\widetilde S|}{|S|\vee 1} \big\{1-O(\kappa_n^{-1})\big\}\geq 1-O(\kappa_n^{-1})-O\big(h(d)\big).
\end{equation}
Here $h(d) = \{ 1- |\widetilde S|/(|S|\vee 1) \} \rightarrow0$ as $d,n \rightarrow \infty.$ Although both methods are guaranteed with asymptotic power one, the asymptotic rate for the empirical power of our ``knockoffs first" proposal in (\ref{low.pow.kf}) is not slower than that of the ``truncation first" competitor in (\ref{low.pow.ef}).
\end{remark} 
\vspace{-0.3cm}
\subsection{High-dimensional FFLR}
\label{sec:fflr}
\vspace{-0.3cm}
Consider the high-dimensional FFLR model
\begin{equation}\label{FFLR.eq}
\setlength{\abovedisplayskip}{5pt}
    \setlength{\belowdisplayskip}{5pt}
\begin{array}{cll}
Y_i(v) &=& \sum _{j=1}^{p} \int_{\cU} X_{ij}(u)\beta_{j}(u,v) {\rm d}u + \varepsilon_i(v), ~i \in [n] ,~v \in \cV, 
\end{array}
\end{equation}
where $\{\varepsilon_i(\cdot)\}_{i \in [n]} $ are i.i.d. mean-zero random error functions, independent of mean-zero $\{X_{ij}(\cdot)\}_{i\in [n], j\in[p]},$ 
and $\{\beta_j(\cdot,\cdot)\}_{j\in[p]}$ are functional coefficients to be estimated with support 
\begin{equation}\label{support.FFLR}
\setlength{\abovedisplayskip}{5pt}
    \setlength{\belowdisplayskip}{5pt}
    {S}= \{j\in [p]: \|\beta_{j}\|_{\cS} \neq 0\}~~\text{and}~~s=|S| \ll p.
\end{equation} 
Our target is to identify the support $S$ and control FDR at the same time within our proposed framework in Section~\ref{sec:method}. To achieve this, we establish in Lemma~\ref{lemmaeqFFLR} of the Supplementary Material that the nonnull set (\ref{f.nullset}) is equivalent to the support set (\ref{support.FFLR}) for FFLR.


We follow (\ref{exX}) to expand $X_{ij}(\cdot)$ for each $i,j.$ We also approximate $\{Y_i(\cdot)\}$ under the Karhunen-Lo$\grave{\hbox{e}}$ve expansion truncated at $\tilde{d},$ i.e.,
$
Y_i(\cdot)  \approx  {{\bbbeta}}_{i}^\T {\bpsi}(\cdot),
$ where $\bpsi = (\psi_{1},\dots, \psi_{\tilde{d}})^\T$, $\bbbeta_{i} = (\eta_{i1},\dots,\eta_{i\tilde{d}})^\T.$   
Some specific calculations lead to
the representation of (\ref{FFLR.eq}) as 
\begin{equation*}
\setlength{\abovedisplayskip}{5pt}
    \setlength{\belowdisplayskip}{5pt}
\label{Yfunction}
\bbbeta_{i}^\T =\sum_{j=1}^p \bxi_{ij}^\T \bB_j + \bepsilon_i^\T + \bvarepsilon_i^\T,
\end{equation*}
where $\bB_j=\int_{\cU \times \cV}\bphi_j(u) \beta_j(u,v) \bpsi(v)^\T {\rm d}u {\rm d}v \in {\mathbb R}^{d_j \times \tilde d},$ and $\bepsilon_i$ and $\bvarepsilon_i$ represent truncation and random errors, respectively.
In a similar fashion to (\ref{glasso.SFLR}) with multivariate responses formed by estimated FPC scores of $Y_i(\cdot)$'s (i.e., $\widehat{\bbbeta}_i$'s),  we  implement the group lasso to estimate $\{\bB_j\}_{j \in [2p]}$ using the augmented set of estimated FPC scores of functional covariates
\begin{equation}
\setlength{\abovedisplayskip}{5pt}
    \setlength{\belowdisplayskip}{5pt}
\label{glasso.FFLR}
  \min _{\bB_1,\dots,\bB_{2p}} \frac1{2n} \sum_{i=1}^n \|\widehat{\bbbeta}_{i}^\T - \sum_{j=1}^p  {\widehat{\bxi}}_{ij}^\T \bB_j -
 \sum_{j=p+1}^{2p}  \widecheck{{{\bxi}}}_{i(j-p)}^\T\bB_j\|^2
+ \lambda_n \sum_{j=1}^{2p} \|\bB_j\|_{\tF}, 
\end{equation}
where $\lambda_n \geq 0$ is the regularization parameter.  
Let $\{\widehat \bB_j\}_{j \in [2p]}$ be the solution to (\ref{glasso.FFLR}).
Following the knockoffs selection step in Section \ref{subsec23}, we choose the $j$th feature importance measures by $Z_j=\|\widehat \bB_j\|_{\tF}$ and $\widetilde Z_j=\|\widehat\bB_{j+p}\|_{\tF},$ and hence the corresponding knockoff statistic $W_j=\|\widehat \bB_j\|_{\tF}-\|\widehat\bB_{j+p}\|_{\tF}$. Applying the knockoff filter to $\{W_j\}_{j \in [p]},$ 
we obtain the set of selected functional covariates denoted as $\widehat S_{\delta}$.


We can verify Condition~\ref{conscore}  with the choice of $\widetilde Y_i=\widehat \bbbeta_i.$ Applying Theorem~\ref{thm_fdr_sflr}, we then attain valid FDR control in our approach for FFLR.

\begin{theorem}
\label{thm_fdr_fflr}
Suppose that Condition \ref{irrSFLR} holds. Then for any sample size $n$ and target FDR level $q\in [0,1]$, $\widehat{S}_1$  satisfies $\FDR \leq q$ and $\widehat{S}_0$ satisfies $\mFDR \leq q.$
\end{theorem}

Before presenting the power analysis, we need the following regularity conditions, which serve as the FFLR analogs of the conditions imposed for SFLR. 

\begin{condition}\label{cond_error_FFLR}
$\{\varepsilon_i(\cdot)\}_{i\in[n]}$ with covariance operator $\Sigma_{\varepsilon\varepsilon}$ are i.i.d. sub-Gaussian processes in $L_2(\cV),$ i.e., there exists some constant $\tilde{c}$ such that $\eE[e^{\langle x,\varepsilon \rangle}] \leq e^{\tilde{c}^2 \langle x,\Sigma_{\varepsilon\varepsilon} (x) \rangle/2 }$ for all $x\in L_2(\cV).$
\end{condition}

\begin{condition} \label{cond_eigen_FFLR}
The eigenvalues of $\Sigma_{YY}$ satisfy $\tilde\omega_{1} > \tilde\omega_{2} > \dots >0,$ and $ \sum_{l = 1}^{\infty} \tilde\omega_{l} = O(1).$  
There exists some constant $\tilde{\alpha} >1$ such that  $ \tilde\omega_{ l} -  \tilde\omega_{l+1} \gtrsim  l ^{-\tilde{\alpha}-1}$ for $ l \geq 1.$
\end{condition}

\begin{condition} \label{cond_coef_FFLR}
(i) For each $j \in S$, there exists some constant $\tau > (\alpha \vee \tilde{\alpha})/2 +1$ s.t. $|B_{j lm}| \lesssim  (l+m)^{-\tau-1/2}$ for $l,m \geq 1;$ 
(ii) $\min_{j \in S} \|\bB_j\|_\tF \geq \kappa_{n} d^\alpha\lambda_n/ \underline{\mu} $ for some slowly diverging sequence  
$\kappa_n \to \infty $ as $n \to \infty,$ and $\lambda_n \gtrsim s d^{1/2}[( d^{\alpha+3/2} \vee \tilde{d}^{\tilde\alpha+3/2})\{\log(pd\tilde{d})/n\}^{1/2}+   d^{1/2- \tau}];$ 
(iii) There exists some constant $c'_2 \in (2(qs)^{-1},1)$ s.t. $|S_2| \geq c'_2 s$ with $S_2 = \{j \in [p]: \|\bB_j\|_\tF \gg {s^{1/2} d^\alpha \lambda_n}\}.$
 \end{condition} 
Define the power of the proposed approach for FFLR in the same form as (\ref{defpower}). We are now ready to present the following theorem about the asymptotic power one of our proposal. 
 
\begin{theorem}
\label{thm_power_fflr}
Suppose that Conditions~\ref{irrSFLR}, \ref{cond_inf_SFLR}--\ref{cond_eigen_SFLR}, \ref{cond_error_FFLR}--\ref{cond_coef_FFLR} hold 
and $s d^\alpha \lambda_n \to 0.$ Then the selected sets $S_{\delta}$'s satisfy $\Pow(\widehat{S}_{\delta}) \rightarrow 1$ as $n \to \infty.$ 
\end{theorem}

\vspace{-0.3cm}
\subsection{High-dimensional FGGM}
\label{sec:fggm}
\vspace{-0.3cm}
Consider the high-dimensional FGGM, which depicts the conditional dependence structure among $p$ Gaussian random functions $X_1(\cdot), \dots, X_p(\cdot).$ To be specific, let 
$C_{{j}{k}}(u,v)= \cov\big\{X_{j}(u),X_{k}(v)| X_{-\{j,k\}}(\cdot)\big\}$ 
be the covariance between $X_{j}(u)$ and $X_{k}(v)$ conditional on the remaining $p-2$ random functions. Then nodes $j$ and $k$ are connected by an edge if and only if $\|C_{jk}\|_{\cS} \neq 0.$ Let $(V,E)$ be an undirected graph with vertex set $V=[p]$ and edge set
\begin{equation}\label{Edge}
\setlength{\abovedisplayskip}{5pt}
    \setlength{\belowdisplayskip}{5pt}
E = \big\{(j,k) \in [p]^2: \|C_{jk}\|_{\cS} \neq 0, j \neq k\big\}~~\text{with}~~s=|E|.
\end{equation}

Our goal is estimate $E$ based on $n$ i.i.d. observations.
To achieve this, \cite{qiao2019a} proposed a functional graphical lasso approach to estimate a block sparse inverse covariance matrix by treating dimensions of $X_{ij}(\cdot)$'s as approaching infinity. However, their proposal cannot handle truly infinite-dimensional objects due to the unboundedness of the inverse of $\bSigma_{XX}.$ Along the line of \cite{mein2006}, we develop a functional neighborhood selection method to estimate $E$ by identifying the important neighborhoods of each node in a FFLR setup,
\begin{equation} 
\setlength{\abovedisplayskip}{5pt}
    \setlength{\belowdisplayskip}{5pt}
\label{eq.GGM}
X_{ij}(v) = \sum _{k \neq j} \int_{\cU} X_{ik}(u) \beta_{jk}(u,v) {\rm d}u + \varepsilon_{ij}(v),~~i \in [n], j \in [p], v \in \cU,
\end{equation}
where, for each $j,$ $\{\varepsilon_{ij}(\cdot)\}_{i \in [n]} $ are i.i.d. mean-zero Gaussian random errors, independent of $\{X_{i,-j}(\cdot)\}_{i\in [n]}.$ Denote the neighborhood set of node $j$ by
\begin{equation*}\label{neibor.GGM}
\setlength{\abovedisplayskip}{5pt}
    \setlength{\belowdisplayskip}{5pt}
S_j = \big\{ k \in [p] \backslash\{j\} : \|\beta_{jk}\|_{\cS} \neq 0  \big\}~~\text{with}~s_j=|S_j|.
\end{equation*}
Before associating the edge set $E$ with neighborhood sets $S_j$'s, we give a regularity condition.
\begin{condition}
    \label{irrGGM}
    $\bSigma_{X_{-\{j,k\}}X_{-\{j,k\}}}^{-1} \bSigma_{X_{-\{j,k\}}X_{j}}$ and $\bSigma_{X_{-\{j,k\}}X_{-\{j,k\}}}^{-1} \bSigma_{X_{-\{j,k\}}X_{k}}$ are bounded linear operators for each $(j,k) \in [p]^2$ with $j \neq k.$
\end{condition} 

Despite the unboundedness of $\bSigma_{X_{-\{j,k\}}X_{-\{j,k\}}}^{-1},$ it is usually associated with another operator. The composite operators in Condition~\ref{irrGGM} can be viewed as regression operators, 
and hence can reasonably be assumed to be bounded.


\begin{lemma}
\label{lemmaeqGGM}
 Suppose that Condition \ref{irrGGM} holds. Then $E = \{(j,k) \in [p]^2: k \in S_j\}.$
\end{lemma}

Lemma~\ref{lemmaeqGGM} suggests that we can recover $E$ by estimating $S_j$ for each $j.$ To achieve this with FDR control, we build upon the idea of \cite{li2021ggm}, which employs a nodewise construction of knockoffs/feature statistics, and a global procedure to obtain thresholds for different nodes, and incorporate it into our 
framework.
With each $X_{ij}(\cdot)$ expanded according to (\ref{exX}), some specific calculations lead to  
the representation of (\ref{eq.GGM}) as 
\begin{equation*}
\setlength{\abovedisplayskip}{5pt}
    \setlength{\belowdisplayskip}{5pt}
\label{YGGM}
\bxi_{ij}^\T =\sum_{k\neq j} \bxi_{ik}^\T \bB_{jk} + \bepsilon_{ij}^\T + \bvarepsilon_{ij}^\T,
\end{equation*}
where $\bB_{jk}=\int_{\cU^2}\bphi_k(u)\beta_j(u,v)\bphi_j(v)^\T{\rm d}u{\rm d}v \in {\mathbb R}^{d_k \times d_j},$ and $\bepsilon_{ij}$ and $\bvarepsilon_{ij}$ represent truncation and random errors, respectively. As a result, we can rely on the block sparsity pattern in $\{\bB_{jk}\}_{1 \leq j\neq k\leq p}$ to identify neighbourhood sets $S_j$'s. Within functional knockoffs framework, we denote the augmented coefficients of PFC scores as $\{\bB_{jk}\}_{k \in [2p] \backslash \{j,j+p\}}$ (the first $p-1$ coefficient matrices are for original covariates and the last $p-1$ are for knockoffs).


After performing FPCA on $\{X_{ij}(\cdot)\}_{i \in [n]}$ for each $j,$ we solve the following group-lasso-based minimization problem using estimated FPC scores of original and knockoff variables:
\begin{equation} \label{glasso.FGGM} 
\setlength{\abovedisplayskip}{7pt}
    \setlength{\belowdisplayskip}{7pt}
    \min _{\bB_{jk}}
 \frac1{2n} \sum_{i=1}^n \|\widehat{\bxi}_{ij}^\T - \sum_{1\leq k \neq j \leq p}  {\widehat{\bxi}}_{ik}^\T \bB_{jk} -
 \sum_{p < k\neq( p+j )\leq 2p}  \widecheck{{{\bxi}}}_{i(k-p)}^\T\bB_{jk}\|^2
+ \lambda_{nj} \sum_{k\neq j, (p+j)} \|\bB_{jk}\|_{\tF}, 
\end{equation}
where $\lambda_{nj} \geq 0$ is the regularization parameter. Let $\{\widehat{\bB}_{jk}\}_{k\in[2p] \backslash \{j,j+p\}}$ be the solution to (\ref{glasso.FGGM}).
For the $j$th node, we select the importance measures as $Z_{jk} = \|\widehat{\bB}_{jk}\|_\tF$ and $\widetilde{Z}_{jk} = \|\widehat{\bB}_{j(k+p))}\|_\tF$ for $k \neq j,$ which results in the corresponding knockoff statistic
$W_{jk} = \|\widehat{\bB}_{jk}\|_{\tF}-\|\widehat{\bB}_{j(k+p)}\|_{\tF}.$
This forms a $p \times p$ matrix of knockoff statistics $\bW= (W_{jk})_{p \times p}$ with $W_{jj}=0.$
Given a threshold vector $\bT_{\delta} = (T_{\delta,1},\dots,T_{\delta,p})^\T$ with $\delta=1$ or 0 for knockoff or knockoff+ filter, respectively, we obtain the estimated neighborhood set of node $j$ as 
\begin{equation*}\label{neihat}
\setlength{\abovedisplayskip}{5pt}
    \setlength{\belowdisplayskip}{5pt}
\widehat{S}_{\delta, j} = \{k \in [p]\backslash \{j\} :W_{jk} \geq T_{\delta, j}   \}. 
\end{equation*}
We estimate the edge set $E$ by applying either the AND or OR rule to $\widehat S_{\delta, j}$'s,
\begin{equation*} \label{edge_and_or}
\setlength{\abovedisplayskip}{5pt}
    \setlength{\belowdisplayskip}{5pt}
\widehat{E}_{\tAnd,\delta} = \big\{(k,j): k\in \widehat{S}_{\delta,j} ~\hbox{and}~ j\in \widehat{S}_{\delta,k} \big\},~~~
\widehat{E}_{\tOr,\delta} = \big\{(k,j): k\in \widehat{S}_{\delta,j} ~\hbox{or}~ j\in \widehat{S}_{\delta,k}  \big\}.
\end{equation*}

There are two options for selecting $\bT_{\delta},$ the node-based local procedure and the graph-based global procedure. For the local one, we employ the knockoff filter to each row of $\bW$ with corresponding $T_{\delta,j}$'s determined via (\ref{threshold}).
For each $j,$ under verified Condition~\ref{conscore} with $\widetilde Y_i=\widehat \bxi_{ij},$ 
we can apply Theorem~\ref{thm_fdr_sflr} to attain node-based FDR control at level $q/p.$ Whereas such local procedure can be easily verified to achieve graph-based FDR control at level $q,$ it results in substantial power loss as discussed in \cite{li2021ggm}. Inspired by \cite{li2021ggm}, we develop a graph-based global approach by solving the following optimization problems to compute $\bT_{\delta}$ under the AND and OR rules, respectively.
\begin{equation}\label{opt.AND}
\setlength{\abovedisplayskip}{5pt}
    \setlength{\belowdisplayskip}{5pt}
\begin{split}
{\bT}_{\delta} = & \arg\max _{\bT} |\widehat{E}_{\tAnd, \delta}|, \\
               & \text{subject to } \frac{a \delta +|\widehat{S}^{-}_{\delta,j}|}{|\widehat{E}_{\tAnd, \delta}| \vee 1} \leq \frac{2q}{c_ap} \hbox{ and }~T_{\delta,j} \in \big\{|W_{jk}|, k\in [p] \big\} \cup \{\infty\} \backslash \{0\}, j\in [p],
\end{split}
\end{equation}
\begin{equation}\label{opt.OR}
\begin{split}
{\bT}_{\delta} =& \arg\max _{\bT} |\widehat{E}_{ \tOr,\delta}|, \\
               &\text{subject to } \frac{a\delta+|\widehat{S}^{-}_{\delta,j}|}{|\widehat{E}_{\tOr,\delta}| \vee 1} \leq \frac{q}{c_ap} \hbox{ and }~T_{\delta,j} \in \big\{|W_{jk}|, k\in [p] \big\} \cup \{\infty\} \backslash \{0\}, j\in [p],
\end{split}
\end{equation}
where, due to the coin flipping property of each row in $\bW,$  $\widehat{S}^{-}_{\delta,j} = \{k \in [p]\backslash \{j\} :W_{jk} \leq -T_{\delta,j} \}$ is used to approximate the set of false discoveries (i.e., $\widehat{S}_{\delta,j} \cap S_{j}^c$), $c_a > 0 $ is a constant depending on $a,$ and ${\bT}_\delta = \{+\infty,\dots,+\infty\}$ if there is no feasible solution. 
We then define the corresponding FDRs and mFDRs under the AND and OR rules as
\begin{equation*}
\setlength{\abovedisplayskip}{7pt}
    \setlength{\belowdisplayskip}{7pt}
 {\FDR}_{\tAnd} =  {\mathbb E}\left[\frac{|\widehat{E}_{\tAnd,1}  \cap E^c|}{|\widehat{E}_{\tAnd,1}| \vee 1} \right],~~~
 {\FDR}_{\tOr} =   {\mathbb E}\left[\frac{|\widehat{E}_{\tOr,1} \cap E^c|}{|\widehat{E}_{\tOr,1}| \vee 1}\right],
\end{equation*}
\begin{equation*}\label{eq.mfdr}
 {\mFDR}_{\tAnd} =  {\mathbb E}\left[\frac{|\widehat{E}_{\tAnd,0} \cap  E^c|}{|\widehat{E}_{\tAnd,0}| +ac_ap/(2q)} \right],~~~
 {\mFDR}_{\tOr} =   {\mathbb E}\left[\frac{|\widehat{E}_{\tOr,0} \cap E^c|}{|\widehat{E}_{\tOr,0}| +ac_ap/q}\right].
\end{equation*} 

\begin{remark} \label{GGM.FDR.pwr}
The global approach in (\ref{opt.AND}) and (\ref{opt.OR}) not only enables FDR control but also ensures good power. To see this, we use $\FDR_O$ as an illustrative example. Then
$\FDR_{\tOr} \leq \sum_{j=1}^p {\mathbb E}\left[\frac{\widehat S_{1,j} \cap S_j^c}{a + |\widehat S_{1,j}^{-}|}\right] \cdot \frac{q}{c_a p} \leq \sum_{i=1}^p c_a \cdot \frac{q}{c_a p }=q,$
where the first inequality follows from (\ref{opt.OR}) and the second inequality follows from \cite{li2021ggm}. 
Furthermore, the denominators in constraints of (\ref{opt.AND}) and (\ref{opt.OR}) are graph-based global terms instead of node-based local terms. This results in a broader feasible domain for the threshold vector, yielding larger estimated edge sets and increased power.
\end{remark}
 
We formalize the above remark about valid FDR control in the following theorem. 

\begin{theorem}\label{thm_fdr_fggm}
For any $n$ and $q\in [0,1],$ ${\FDR}_{\tAnd}\leq q$ and ${\mFDR}_{\tAnd}\leq q$ under the AND rule, while ${\FDR}_{\tOr} \leq  q$ and ${\mFDR}_{\tOr} \leq  q$ under the OR rule.
\end{theorem}


To present the power theory, we give a condition akin to Condition \ref{cond_coef_SFLR} for SFLR and Condition \ref{cond_coef_FFLR} for FFLR.

\begin{condition} \label{cond_power_GGM}
(i) For each $(j,k) \in E$, there exists some constant $\tau > \alpha/2 +1$ such that $|B_{jk lm }| \lesssim (l+m)^{-\tau-1/2}$ for $l,m\geq 1$;  
(ii) For each $j\in [p],$ $\min_{k \in S_j} \|\bB_{jk}\|_\tF \geq \kappa_{nj} d^\alpha  \lambda_{nj}/ \underline{\mu} $ for some slowly 
diverging sequences $\kappa_{nj} \rightarrow \infty$ as $n \rightarrow \infty,$ and $\lambda_{nj} \gtrsim s_j[ d^{\alpha+2} {\{\log(pd)/n}\}^{1/2} + d^{1- \tau}]$; 
(iii) For each $j,$ there exists some constant $c_j \in \big((1+a) c_a p (qs)^{-1},1\big)$ such that $|S_{j2}| \geq c_j s_j$  with $S_{j2} = \{k \in [p] \backslash \{j\}: \|\bB_{jk}\|_\tF \gg {{s_j}^{1/2}  d^\alpha  \lambda_{nj}} \}$.  
\end{condition}

We define the power of our proposal by
\begin{equation*}
\setlength{\abovedisplayskip}{7pt}
    \setlength{\belowdisplayskip}{7pt}
 \Pow(\widehat{E}_{\tAnd,\delta})   =  {\mathbb E}\left[ \frac{|\widehat{E}_{\tAnd,\delta} \cap E|}{| E| \vee 1}\right] ~~\text{and}~~\Pow(\widehat{E}_{\tOr,\delta})   =  {\mathbb E}\left[ \frac{|\widehat{E}_{\tOr,\delta} \cap E|}{| E| \vee 1}\right].
\end{equation*}
\begin{theorem}
\label{thm_power_fggm}
Suppose that Conditions~\ref{cond_inf_SFLR}--\ref{cond_eigen_SFLR} and \ref{irrGGM}--\ref{cond_power_GGM} hold, 
and $s_j d^\alpha \lambda_{nj} \to 0$ for each $j$. Then the selected edge sets $\widehat{E}_{\tOr,\delta}$'s satisfy  $\Pow(\widehat{E}_{\tOr,\delta}) \rightarrow 1$ as $n\rightarrow \infty$.
\end{theorem}

Theorem~\ref{thm_power_fggm} provides the power guarantee in high dimensions when utilizing the OR rule and 
we leave the power analysis of $\widehat{E}_{\tAnd,\delta}$ as future research. 

\vspace{-0.3cm}
\section{Constructing functional model-X knockoffs}
\vspace{-0.3cm}
\label{sec:cons} 
\subsection{Exact construction}
\label{sec:construct} 
\vspace{-0.3cm}
We consider the second-order functional knockoffs construction by matching the mean and covariance functions of $\bX$ and $\widetilde \bX.$ Assuming that $(\bX^\T, \widetilde \bX^\T)^\T$ follows MGP as specified in Example~$\ref{ex.MGP},$ we achieve alignment of the first two moments, which implies a matching of joint distributions, so that we have an exact construction of functional knockoffs.


To ensure the positive semi-definiteness of $\bSigma_{(X,\widetilde{X})(X,\widetilde{X})}$ and maintain good statistical power, we need to solve the following optimization problem to obtain $Q_{X_jX_j}$'s:
\begin{equation}\label{opt.R.cov}
\setlength{\abovedisplayskip}{5pt}
    \setlength{\belowdisplayskip}{5pt}
 \min _{\{Q_{X_jX_j}\}_{j \in [p]}}\sum_{j} \|\Sigma_{ X_j{X}_j} - Q_{X_j X_j}\|_{\rm{norm}}~~
 \text{subject~to}~~2\bSigma_{XX} - \bQ_{XX}\succeq 0,
\end{equation}  
where $\|\cdot\|_{\text{norm}}$ denotes some proper functional norm.
However, for each $j \in [p],$ the fact $X_j \in {\cal H}_j$ along with the constraint in (\ref{opt.R.cov}) implies that, $\lambda_{\min}(Q_{ X_j{X}_j}) \leq 2\lambda_{\min}(\Sigma_{ X_j{X}_j}) \rightarrow 0,$ where $\lambda_{\min}(\cdot)$ denotes the minimum eigenvalue. When the minimum and maximum eigenvalues of $Q_{X_jX_j}$ are of the same order, 
we have $Q_{X_jX_j} \rightarrow 0.$ This makes the original variables nearly indistinguishable from the knockoff counterparts, leading to substantially declined power.


To address this issue, we leverage the correlation operators between $X_j$'s and $X_k$'s for $j,k \in [p]$ \cite[]{baker1973joint}, i.e., $C_{X_{j}X_{k}}:\mathcal{H}_{k} \to \mathcal{H}_{j}$ such that $\|C_{X_{j}X_{k}}\|_{\cL} \leq 1$  and
   $\Sigma_{X_{j}X_{k}} = \Sigma_{X_{j}X_{j}}^{1/2} C_{X_{j}X_{k}} \Sigma_{X_{k}X_{k}}^{1/2},$  
where $\Sigma_{X_jX_j}^{1/2} = \sum _{l=1}^{\infty} \omega_{jl}^{1/2}\phi_{jl}\otimes\phi_{jl}$ is the square-root of the operator $\Sigma_{X_jX_j}.$
We denote $\bC_{X X} =(C_{X_{j}X_{k}})_{p \times p}$ as the correlation operator matrix of $\bX.$ 
It follows from \cite{solea2022copula} that, 
under mild regularity conditions, $\bC_{ XX} - c^* \bI \succeq 0$ for some positive constant $c^*,$ which implies $\lambda_{\min}(\bC_{ XX}) \geq c^*.$
E.g., when $C_{X_{j}X_{k}} = 2^{-|j-k|} \sum_{l} \phi_{jl} \otimes  \phi_{kl},$ 
it is easy to verify that $\lambda_{\min}(\bC_{XX}) \geq 1/3.$  By utilizing $\bC_{XX}$ instead of $\bSigma_{XX},$ the constraint in (\ref{opt.R.cov}) becomes $2 \bC_{XX} - \bR_{ XX} \succeq 0,$ where $\bR_{XX}=\tdiag(R_{X_1X_1},\dots,R_{X_pX_p}).$ This makes it feasible to determine $R_{X_{j}X_{j}}$ with eigenvalues being bounded away from zero (see e.g., the construction of $R_{X_{j}X_{j}}$ under E\ref{R.ex1} and E\ref{R.ex2} below), which ensures discrepancies between original and knockoff variables with enhanced power. Some simple calculations show that the covariance structure in (\ref{Q.cov}) reduces to the correlation structure
\begin{equation}
\label{R.cor}
 \setlength{\abovedisplayskip}{5pt}
    \setlength{\belowdisplayskip}{5pt}
\bC_{X X}=\bC_{\widetilde X \widetilde X}, ~~ 
\bC_{X \widetilde X}= \bC_{\widetilde X X} =\bC_{XX} -  \bR_{XX},
\end{equation}
where $Q_{X_jX_j}= \Sigma_{X_jX_j}^{1/2}R_{X_jX_j}\Sigma_{X_jX_j}^{1/2}$ for $j \in [p].$ Hence we can achieve the equivalence between $\bSigma_{(X,\widetilde{X})(X,\widetilde{X})} \succeq 0$ and $2\bC_{XX} - \bR_{XX} \succeq 0.$ 
As a result, we propose to obtain $R_{X_jX_j}$'s by solving the following optimization problem instead of (\ref{opt.R.cov})
\begin{equation}\label{opt.R.cor}
\setlength{\abovedisplayskip}{5pt}
    \setlength{\belowdisplayskip}{5pt}
 \min _{\{R_{X_j X_j }\}_{j\in[p]}} \sum_{j} \|C_{ X_j{X}_j} - R_{X_j X_j}\|_{\cL}~~
\text{subject~to}~~ 2\bC_{XX} - \bR_{XX} \succeq 0.
\end{equation}

\begin{remark}
It is crucial to select an appropriate functional norm in the objective function of (\ref{opt.R.cor}). Notice that, for each $j,$ $C_{ X_j{X}_j} - R_{X_j X_j}$ is neither Hilbert--Schmidt nor nuclear with possibly unbounded $\|C_{ X_j{X}_j} - R_{X_j X_j}\|_{\cS}$ and $\|C_{X_j{X}_j} - R_{X_j X_j}\|_{\cN}.$ Therefore, we opt for the operator norm, considering that
$\|C_{ X_j{X}_j} - R_{X_j X_j}\|_{\cL} < \infty$. 
\end{remark}
To solve the optimization problem (\ref{opt.R.cor}), we rely on the expression of the correlation operator $C_{X_{j}X_{k}}$ under the Karhunen-Lo$\grave{\hbox{e}}$ve expansion (\ref{exX}) in the following lemma.

\begin{lemma}
\label{lemma2}
Suppose that $\sum_{l=1}^{\infty} \omega_{jl}^{1/2} < \infty$.  
Then, for each $j, k \in [p]$, 
$$
\setlength{\abovedisplayskip}{5pt}
    \setlength{\belowdisplayskip}{5pt}
C_{X_{j}X_{k}} = \sum _{l=1}^{\infty} \sum _{m=1}^{\infty}
\tcorr(\xi_{jl},\xi_{km}) (\phi_{jl} \otimes \phi_{km}).
$$
\end{lemma}
 
It then holds that $C_{X_{j}X_{j}} =  \sum_{l=1}^{\infty} (\phi_{jl} \otimes \phi_{jl}).$  While Lemma~\ref{lemma2} applies to $\bC_{XX},$ $\bC_{\widetilde X\widetilde X}$ and cross-correlation operators $C_{X_j\widetilde X_k}=C_{X_jX_k}$ for $j \neq k$ under (\ref{R.cor}), we focus on $C_{X_j\widetilde X_j}$ that depends on $R_{X_jX_j},$ i.e., $C_{X_j\widetilde{X}_j} = C_{X_{j}X_{j}}-  R_{{X_j}{X_j}}.$ 
Combining these facts, we derive three expressions of $R_{X_j X_j}$ as follows, leading to the corresponding correlation operator $C_{X_j\widetilde{X}_j}$: 
\begin{enumerate}[E1:] 
\setlength{\itemsep}{0ex}
    \item\label{R.ex1} When $ R_{X_j X_j} = \sum_{l=1}^{\infty}r(\phi_{jl} \otimes \phi_{jl})$ for $r \in [0,1],$    $C_{X_j\widetilde{X}_j} = \sum_{l=1}^{\infty} (1-r) (\phi_{jl} \otimes \phi_{jl})$;
    \item\label{R.ex2} When $ R_{X_j X_j} = \sum_{l=1}^{\infty} r_{j} (\phi_{jl} \otimes \phi_{jl})$ for each $r_j \in [0,1],$ $C_{X_j\widetilde{X}_j} = \sum_{l=1}^{\infty} (1-r_{j}) (\phi_{jl} \otimes \phi_{jl})$;
    \item\label{R.ex3} When $ R_{X_j X_j} = \sum_{l=1}^{\infty} r_{jl} (\phi_{jl} \otimes \phi_{jl})$ for each $r_{jl} \in [0,1],$ $C_{X_j\widetilde{X}_j} = \sum_{l=1}^{\infty} (1-r_{jl}) (\phi_{jl} \otimes \phi_{jl}).$
\end{enumerate}

In (\ref{opt.R.cor}), we present the optimization problem at the operator level. To facilitate this optimization task, we establish in Section \ref{obj.dedu} of the Supplementary Material that the objective functions of (\ref{opt.R.cor}), subject to E\ref{R.ex1}, E\ref{R.ex2} and E\ref{R.ex3}, can be simplified to the corresponding equivalent forms as presented in Section \ref{obj.dedu} of the Supplementary Material. 
\vspace{-0.3cm}
\subsection{Implementation}
\label{sec:implement}
\vspace{-0.3cm}
Based on i.i.d. observations $\bX_1, \dots, \bX_n,$ 
we can compute the sample versions $\widehat\bC_{XX}$ and $\widehat\bR_{XX}$ in the constraint of (\ref{opt.R.cor}). To do this, we replace relevant terms in Lemma~\ref{lemma2} by their sample counterparts, resulting in the sample correlation matrix operator $\widehat{\bC}^{\text{S}}_{XX} = (\widehat C^{\text{S}}_{{X}_{j}{X}_{k}})_{p\times p},$ where $\widehat C_{{X}_{j}{X}_{k}}^{\text{S}} = \sum _{l=1}^{\infty} \sum _{m=1}^{\infty} \widehat\Theta_{jklm}^{\text{S}} (\hat \phi_{jl}\otimes  \hat \phi_{km})$
with $\widehat \Theta_{jklm}^{\text{S}}=n^{-1}\sum_{i=1}^n(\widehat\xi_{ijl}-n^{-1}\sum_{i=1}\widehat\xi_{ijl})(\widehat\xi_{ikm}-n^{-1}\sum_{i=1}\widehat\xi_{ikm})/(\hat\omega_{jl}^{1/2}\hat\omega_{km}^{1/2}).$ However, the sample correlation estimator performs poorly in high-dimensional settings. Inspired by \cite{Schfer2005ASA}, we propose a shrinkage version of the sample correlation operator matrix as
\begin{equation}
\label{C.est.shrink}
 \setlength{\abovedisplayskip}{5pt}
    \setlength{\belowdisplayskip}{5pt}
\widehat{\bC}_{XX} = (1-\gamma_n) \widehat{\bC}^{\text{S}}_{XX} + \gamma_n \widehat\bI_{XX}, 
\end{equation}
where $\gamma_n \in [0,1]$ is the shrinkage parameter, and $\widehat\bI_{XX} = \tdiag(\hat{I}_{X_1X_1}, \dots, \hat{I}_{X_pX_p})$ with 
$\hat{I}_{X_jX_j} = \sum _{l=1}^{\infty}\hat{\phi}_{jl} \otimes \hat{\phi}_{jl}$. Additionally, we can obtain sample versions of $R_{X_j X_j}$ under E\ref{R.ex1}, E\ref{R.ex2}, E\ref{R.ex3} as $\widehat{R}_{X_jX_j} = \sum _{l=1}^{\infty} r (\hat{\phi}_{jl} \otimes \hat{\phi}_{jl}),$   $\widehat{R}_{X_jX_j} = \sum _{l=1}^{\infty} r_j (\hat{\phi}_{jl} \otimes \hat{\phi}_{jl}),$ $\widehat{R}_{X_jX_j} = \sum _{l=1}^{\infty} r_{jl} (\hat{\phi}_{jl} \otimes \hat{\phi}_{jl}),$ respectively. 

Nevertheless, it is still difficult to solve the optimization problems at the operator level. To make them executable algorithms, we map operators as matrices using a coordinate representing system within finite-dimensional Hilbert spaces \citep[]{solea2022copula}.
The coordinate mapping employs a finite set of functions $\mathcal{B}_j = \{b_{j1},\dots,b_{jk_n}\}$ to approximate $\mathcal{H}_j$ for $j \in [p].$ 
{\color{black}
Hence, each $X_{ij}$ can be expressed as a linear combination of
$b_{j1},\dots,b_{jk_n},$
where the coefficient vector is denoted by $[{X}_{ij}]_{\mathcal{B}_j}$ and 
is called the coordinate of ${X}_{ij}$
with respect to $\mathcal{B}_j.$ For any operator $K: \mathcal{H}_j \to \mathcal{H}_k$, the coefficient matrix $\big([K(b_{{j}1})]_{\mathcal{B}_k},\dots, [K(b_{{j}k_n})]_{\mathcal{B}_k}\big)$ is denoted by $_{{\mathcal{B}}_{k}}[K]_{{\mathcal{B}}_{j}}$ and 
is called the coordinate of $K$ with respect to $\mathcal{B}_j$ and $\mathcal{B}_k.$ 
In this way, we map each $X_{ij} \in \mathcal{H}_j$ to a vector in $\eR^{k_n}$ 
and each operator $K:\mathcal{H}_j \to \mathcal{H}_k$ to a matrix in $\eR^{k_n \times k_n}.$ 
} 
Let $\bG_j=(G_{jlm})_{k_n \times k_n}$ with $G_{jlm} = \langle b_{jl},b_{jm}\rangle$ be the Gram matrix of $\mathcal{B}_j.$
See details of coordinate mapping in Section~\ref{cor.map} of the Supplementary Material.

{\color{black}In Section \ref{cor.map} of the Supplementary Material, we derive that $2\widehat{\bC}_{XX} - \widehat{\bR}_{XX}\succeq 0$ is implied by $2 \widehat\bTheta_C- \widehat\bTheta_{R} \succeq 0,$  where 
    $\widehat\bTheta_C=(1-\gamma_n)(\widehat{\Theta}_{jklm}^{\text{S}})_{pk_n\times pk_n}+  \gamma_n \bI_{p k_n}.$
}
Combing this with Section \ref{obj.dedu} of the Supplementary Material, we propose solving three sample finite-representations of the optimization problem (\ref{opt.R.cor}):
\begin{enumerate}[E1:] \label{Omega.R.hat}
\setlength{\itemsep}{0ex}
    \item With $\widehat\bTheta_R = \diag (r \bI_{k_n}, \dots ,r \bI_{k_n})\in \eR^{pk_n \times pk_n},$ 
\begin{equation}\label{opt.R.mat.1}
    \setlength{\abovedisplayskip}{5pt}
    \setlength{\belowdisplayskip}{5pt}
  \min _{r}(1-r)~~
  \text{subject~to} ~~ r\in [0,1], 
  ~2 \widehat\bTheta_C- \widehat\bTheta_{R} \succeq 0. 
\end{equation}
\item With
$\widehat\bTheta_R = \diag (r_1 \bI_{k_n}, \dots ,r_p \bI_{k_n})\in \eR^{pk_n \times pk_n},$ 
\begin{equation}\label{opt.R.mat.2}
     \setlength{\abovedisplayskip}{5pt}
    \setlength{\belowdisplayskip}{5pt}
  \min _{(r_1,\dots,r_p) }\sum_{j=1}^{p} (1-r_j)~~
  \text{subject~to}~~r_j\in [0,1],  ~2 \widehat\bTheta_C- \widehat\bTheta_{R}  \succeq 0.
\end{equation}
\item With $\widehat\bTheta_R = \diag(r_{11}, \dots,r_{1k_n}, \dots, r_{p1}, \dots, r_{pk_n}
)\in \eR^{pk_n \times pk_n}$ for $j\in [p],$ 
\begin{equation}\label{opt.R.mat.3}
 \setlength{\abovedisplayskip}{5pt}
    \setlength{\belowdisplayskip}{5pt}
 \min _{(\Bar{\br}_1,\dots,\Bar{\br}_p)} \sum_{j=1}^{p} \sum_{l= 1}^{k_n} |1-r_{jl}|~~
\text{subject~to}~~r_{jl}\in [0,1],  ~
  2 \widehat\bTheta_C- \widehat\bTheta_{R} \succeq 0, 
\end{equation} 
where $\Bar{\br}_j = (r_{j1}, \dots,r_{jk_n})^\T.$
\end{enumerate}
Note that the original objective function $\min _{ (\Bar{\br}_1,\dots,\Bar{\br}_p)}\sum_{j=1}^{p}\sup_{l\in[k_n]} |1-r_{jl}|$ corresponding to E\ref{R.ex3} is difficult to handle. To simplify the computation, we consider the objective function in (\ref{opt.R.mat.3}) instead. We establish the equivalence of both optimization tasks in Remark~\ref{eq.obj.E3} of the Supplementary Material.
\vspace{-0.3cm}
\subsection{Algorithms} 
\label{sec:algorithm}
\vspace{-0.3cm}
To simplify notation, we use $[X_{ij}]_{{\mathcal{B}}_{j}}$ to denote the centered version $[X_{ij} - \hat{\mu}_{j}]_{{\mathcal{B}}_{j}}.$ Since constructing functional knockoffs is achieved with the aid  of empirical Karhunen-Lo$\grave{\hbox{e}}$ve expansion, we firstly summarize the algorithm for Karhunen-Lo$\grave{\hbox{e}}$ve expansion using the finite coordinate representation in Algorithm~\ref{Alg1}. We then present the algorithm for constructing functional model-X knockoffs in Algorithm \ref{Alg:2}. 
By the fact that $(\bX_i^\T, \widetilde \bX_i^\T)^\T$ is MGP and the derivations in Section~\ref{smsec:algorithm} of the Supplementary Material, we obtain, in Step~2 that,  
\begin{equation}\label{cond.dist}
\setlength{\abovedisplayskip}{5pt}
    \setlength{\belowdisplayskip}{5pt}
    \widehat\bW^{-1/2}\widehat\bA^\T\big ( [\widetilde X_{i1}]_{\mathcal{B}_1}^\T,\dots,[\widetilde X_{ip}]_{\mathcal{B}_p}^\T\big)^{\T} \Big|   \widehat\bW^{-1/2}\widehat\bA^\T \big([X_{i1}]_{\mathcal{B}_1}^\T,\dots,[X_{ip}]_{\mathcal{B}_p}^\T\big)^{\T} \sim {\cal N}(\widehat\bmu_{\tilde X|X},\widehat\bTheta_{\tilde X|X})
\end{equation} 
for each $i \in [n],$ where the normalization matrix
$\widehat\bW = \tdiag(\hat\omega_{11}, \dots, \hat\omega_{1k_n},\dots, \hat\omega_{p1}, \dots, \hat\omega_{pk_n}),$ {\color{black} $\widehat\bA = \tdiag(\bG_{1}\widehat\bPhi_{1},\dots,\bG_{p}\widehat\bPhi_{p})  
\in \eR^{pk_n \times pk_n}$ is the mapping matrix from the space of FPC scores to that of coordinates, $\widehat\bPhi_j \in {\mathbb R}^{k_n \times k_n}$ with its $l$th column $[\hat{\phi}_{j l}]_{{\mathcal{B}}_{j}}$ for $l\in[k_n]$, } $\widehat\bmu_{\widetilde X|X} =  \big(\bI_{pk_n} - \widehat\bTheta_R \widehat\bTheta_C^{-1} \big)\widehat\bW^{-1/2}\widehat\bA^\T \big([X_{i1}]_{\mathcal{B}_1}^\T,\dots,[X_{ip}]_{\mathcal{B}_p}^\T\big)^{\T}$ and
$\widehat\bTheta_{\widetilde X|X} = 2\widehat\bTheta_{R}- \widehat\bTheta_{R} \widehat\bTheta_C^{-1}\widehat\bTheta_{R}.$
\begin{algorithm}[h]
\caption{Algorithm for Karhunen-Lo$\grave{\hbox{e}}$ve expansion.}
\begin{algorithmic}[1]
   \STATE For each $j \in [p],$ choose a set of functions $\mathcal{B}_j = \{b_{j1},\dots,b_{jk_n }\}$ on $\cU$ and compute $\bG_j.$
   \STATE Compute the coordinates  $[X_{ij}]_{\mathcal{B}_j}$ relative to the basis $\mathcal{B}_j$ by least squares. 
    \STATE Perform spectral decomposition on $n^{-1} \bG_j^{1/2}  \sum_{i=1}^n \big([ X_{ij}  ]_{{\mathcal{B}}_{j}} [ X_{ij} ]_{{\mathcal{B}}_{j}}^{\T}\big) \bG_j^{1/2}$ to obtain eigenpairs $(\hat\omega_{jl}, \widehat \bv_{jl})$ for $l\in [k_n].$
    \STATE Compute $\hat\phi_{jl}  =(b_{j1}, \dots, b_{jk_n}) \bG_j^{\dag 1/2} \widehat\bv_{jl}$ and ${\hat\xi}_{ijl}=[X_{ij} ]_{\mathcal{B}_j}^\T \bG_j^{1/2} \widehat \bv_{jl}$ for $i \in [n]$ and $l \in [k_n].$
 \end{algorithmic}
 \label{Alg1}
\end{algorithm} 
\begin{algorithm}[!htp]
\caption{ Algorithm for constructing functional model-X knockoffs.}
 \label{Alg:2}
\begin{algorithmic}[1]
\STATE\label{step:1} 
Under E\ref{R.ex1}, E\ref{R.ex2} and E\ref{R.ex3}, recover the expressions of {\color{black}$ \widehat\bTheta_{R}$} by replacing $r,$ $(r_1,\dots,r_p)$ and $(\Bar{\br}_1,\dots,\Bar{\br}_p)$ with the corresponding solutions to the optimization problems in (\ref{opt.R.mat.1}), (\ref{opt.R.mat.2}) and (\ref{opt.R.mat.3}), respectively.
   \STATE \label{step:2} 
Sample $\bZ_1, \dots \bZ_n$ independently from ${\cal N}({\bf 0}_{pk_n}, \bI_{pk_n}).$ Based on the conditional distribution (\ref{cond.dist}),  obtain $\widehat\bW^{-1/2}\widehat\bA^\T( [\widetilde X_{i1}]_{\mathcal{B}_1}^\T,\dots,[\widetilde X_{ip}]_{\mathcal{B}_p}^\T)^{\T}  = \widehat \bmu_{\widetilde X|X}  +  \widehat\bTheta_{\widetilde X|X}^{1/2}\bZ_i,$ which turns to be the coefficient-vector $(\check \xi_{i11}, \dots, \check \xi_{i1k_n}, \dots, \check \xi_{ip1}, \dots, \check \xi_{ipk_n})^{\T}$ of $p$-vector of functional knockoffs under respective Karhunen-Lo$\grave{\hbox{e}}$ve expansions. 
 \STATE \label{step:3} Construct functional knockoffs as $\widecheck{X}_{ij}(\cdot)  =  \sum_{l=1}^{k_n}\check\xi_{ijl} \hat \phi_{jl}(\cdot)$ for $i\in [n]$ and $j\in [p].$ 
 \end{algorithmic}
\end{algorithm}

 \vspace{-0.3cm}
\subsection{Partially observed functional data}
\label{sec:partial}
\vspace{-0.3cm}
In this section we consider a practical scenario where each $X_{ij}(\cdot)$ is partially observed, with errors, at $L_{ij}$ random time points $U_{ij1}, \dots, U_{ijL_{ij}}\in \cU.$ Let $W_{ijl}$ be the observed value of $X_{ij}(U_{ijl})$ satisfying
\begin{equation}
\label{model.partial}
W_{ijl} = X_{ij}(U_{ijl}) + e_{ijl},~~l = 1,\dots, L_{ij},
\end{equation}
where $e_{ijl}$'s are i.i.d. mean-zero errors with finite variance, independent of $X_{ij}(\cdot).$ The sampling frequencies $L_{ij}$'s play a crucial role when choosing the smoothing strategy. When $L_{ij}$'s are larger than some order of $n,$ the conventional approach employs nonparametric smoothing on $W_{ij1}, \dots, W_{ijL_{ij}}$ to reduce the noise, see, e.g., local linear smoothers \cite[]{kong2016}. This allows the reconstruction of individual curves, which can be treated as original covariates to construct functional model-X knockoffs. When $L_{ij}$'s are bounded, the pre-smoothing step is no longer viable. In such cases, it is recommended to apply nonparametric smoothers for estimating the mean, marginal- and cross-covariance functions, which are essential terms within the functional model-X knockoffs framework. This can be achieved by pooling data from subjects to build strength across all observations \cite[]{fang2022}.

\vspace{-0.3cm}
\section{Simulations}
\label{sec:sim}
\vspace{-0.3cm}
In this section, we conduct a number of simulations to assess the finite-sample performance of the proposed functional knockoffs selection methods for SFLR, FFLR and FGGM. 
We compare our ``knockoff first" proposals, which include the construction of functional knockoffs under E\ref{R.ex1}, E\ref{R.ex2} and E\ref{R.ex3} (respectively denoted as KF1, KF2 and KF3), with two competing methods. The first competitor follows a ``truncation first" strategy (denoted as TF), as detailed in Remark \ref{fdrRmK}. The second is a group-lasso-based variable selection method (denoted as GL), which involves initial dimension reduction followed by group-lasso for group-variable selection but does not include knockoffs.


Implementing our proposals require choosing the shrinkage parameter $\gamma_n$ in (\ref{C.est.shrink}), the dimension $k_n$ in the coordinate mapping and truncated dimensions $d_j$'s (and $\tilde d$ for FFLR). 
To select $\gamma_n,$ we can either use cross-validation or minimize the mean squared error of $\widehat\bTheta_C$ \cite[]{Schfer2005ASA}, while the latter approach is adopted for its computational efficiency.
In the coordinate mapping, we follow \cite{solea2022copula} to use cubic spline functions with 3 interior nodes, which results in $k_n=7$ spline functions to span each ${\cal H}_j$. 
To determine the truncated dimensions, we take the standard approach by selecting the largest $d_j$ (or $\tilde d$ for FFLR) eigenvalues of $\widehat \Sigma_{X_jX_j}$ (or $\widehat \Sigma_{YY}$ for FFLR) such that the cumulative percentage of selected eigenvalues exceeds $90\%.$ Inspired from \cite{wang2011consistent}, we consider minimizing the following high-dimensional BIC to choose the regularization parameter in the penalized least squares (\ref{glasso.FFLR}) for FFLR
\begin{equation} 
\label{HBIC}
\text{HBIC}(\lambda_n)=n\log\big\{\text{RSS}(\lambda_n)\big\} + 2\hbar \log\Big(\tilde{d} \sum_{j=1}^{2p} d_j\Big) \sum_{j=1}^{2p} \Big\{\frac{(\tilde{d}d_j-1)\|\widehat{\bB} \|_{\tF}}{\|\widehat{\bB} \|_{\tF} + \lambda_n} + \bI(\|\widehat{\bB} \|_{\tF} >0)\Big\},
\end{equation}
where $\text{RSS}(\lambda_n)$ represents the residual sum of squares, and $\hbar$ is a constant in $[0.1,3]$ to maintain comparable power levels. The criterion (\ref{HBIC}) is also applicable for selecting $\lambda_n$ in (\ref{glasso.SFLR}) for SFLR with $\tilde d=1$ and $\|\bB\|_{\tF}$ degenerated to $\|\bbb\|.$ Since our proposal for FGGM involves $p$ FFLRs, we can select the regularization parameters $\lambda_{nj}$'s  
in the same fashion as for FFLR. With (\ref{HBIC}), we can also determine the corresponding regularization parameters for each comparison method within each model.


\label{sim:full}
To mimic the infinite-dimensionality of random functions, we generate functional variables by $X_{ij}(u)=\widetilde\bphi(u)^{\T}{\btheta}_{ij}$ for $i\in [n], j=[p]$ and $u \in \cU=[0,1],$ where $\widetilde\bphi(u)$ is a $25$-dimensional Fourier basis function and $(\btheta_{i1}^{\T},\dots,\btheta_{ip}^{\T} )^{\T} \in \eR^{25p}$ is generated independently from a mean zero multivariate normal distribution with block covariance matrix
$\bLambda \in \eR^{25p \times 25p},$ whose $(j,k)$th block is $\bLambda_{jk} \in \eR^{25 \times 25}$ for $j,k \in [p].$
The functional sparsity pattern in $\bSigma_{XX}= \big(\Sigma_{jk}(\cdot,\cdot)\big)_{p \times p}$ with its $(j,k)$th entry $\Sigma_{jk}(u,v) = \widetilde\bphi(u)^{\T}\bLambda_{jk}\widetilde\bphi(v)$ 
can be captured by the block sparsity structure in $\bLambda.$ 
Define $\bLambda_{jj} = \text{diag}(1^{-2},\dots, 25^{-2})$ and $\bLambda_{jk}=(\Lambda_{jklm})_{25 \times 25},$ where
$\Lambda_{jklm} = 0.5\rho^{|j-k|}l^{-1}m^{-1}$ for $l \neq m$ and $\rho^{|j-k|} l^{-2}$ for $l=m$ with $\rho=0.5.$
We generate $n =100,200$ observations of $p =50,100,150$ functional variables, and replicate each simulation $200$ times. For each of the three models, the data is generated as follows.
  
{\bf SFLR:} 
We generate scalar responses $\{Y_{i}\}_{ i \in [n]}$ from model (\ref{SFLR.eq}), 
where $\varepsilon_i$'s are independent standard normal.
For each $j \in S = \{1, \dots, 10\}$, we generate $\beta_j(u) = \sum_{l =1}^{25} b_{jl} \tilde\phi_l(u) $ for $u \in  \cU,$ where $b_{jl} = (-1)^{l} c_b  l^{-2}$ for $l=1, \dots,25,$ and the strength of signals via $c_b$'s are sampled from $\text{Unif}[4,6].$
For $j \in [p] \setminus S,$ we set $\beta_j(u)  = 0.$

{\bf FFLR:} We generate functional responses $\{Y_{i} (v):v \in \cU\}_{ i \in [n]}$ from model (\ref{FFLR.eq}),
where $\varepsilon_i(v) = \sum_{l= 1}^{5} g_{il} \tilde\phi_l(v)$ with $g_{il}$ being i.i.d. ${\cal N}(0,1).$
For $j \in S$, we generate $\beta_j(u,v) = \sum_{l,m =1}^{25} B_{jlm} \tilde\phi_l(u)\tilde \phi_m(v) $ for $(u,v) \in \cU^2, $ 
where $B_{jlm} = (-1)^{l+m} c_b (l+m)^{-2}$ for $l,m=1, \dots, 25,$ and $c_b$'s are sampled from $\text{Unif}[4,6].$
For $j \in [p] \setminus S,$ we set $\beta_j(u,v)=0.$

{\bf FGGM:}
Different from the above data generating process, we sequentially generate $X_{i1}(\cdot), \dots, X_{ip}(\cdot).$ We firstly generate the functional errors 
$\varepsilon_{ij}(u) =\widetilde \bphi(u)^{\T}\widetilde \btheta_{ij}$ for $i\in [n]$ and $j=[p],$ 
where $\widetilde \btheta_{ij}$ are sampled independently from ${\cal N}({\bf 0}_{25}, \bLambda_{jj}).$ We then adopt the following structural equations to establish a directed acyclic graph, 
$$\setlength{\abovedisplayskip}{5pt}
    \setlength{\belowdisplayskip}{5pt}
    X_{i1}(u)=\varepsilon_{i1}(u)~~\text{and}~~X_{ij}(u) =\sum_{(k,j) \in E_{D}} \int_{\cU} X_{ik}(v) \beta_{jk}(u,v) dv   + \varepsilon_{ij}(u) ~\text{for}~j\in [p]\setminus \{1\},$$ 
where $E_{D}$ represents the directed edge set. A pair $(i,j) \in E_D$ is said to be directed from node $i$ to node $j$ if $(j,i) \notin E_D,$ then node $j$ is a child of node $i.$ Denote the candidate edge set as $E_c=\{(k,j) \in [p]^2:k<j\}.$ To determine $E_D,$ we randomly select one edge from $E_c$ for each child node $j\in \{2, \dots,p\}$ in a sequential way,
and then randomly choose $p/3$ edges from the remaining $E_c.$ 
We generate $\beta_{jk}(u,v) = \sum_{l,m =1}^{25} B_{jklm} \tilde\phi_l(u) \tilde\phi_m(v) $ for $(u,v) \in \cU^2,$ where $B_{jklm} = (-1)^{l+m} c_bs_j^{-1} (l+m)^{-2}$ for $l,m=1,\dots,25,$ and $c_b$'s are sampled from $\text{Unif}[4,6].$  
We finally moralize the directed graph to obtain the undirected graph \cite[]{Cowell2007}.

We present numerical summaries of all comparison methods in terms of empirical power and FDR for SFLR, FFLR and FGGM in Tables~\ref{tab:1}, \ref{tab:2} and \ref{tab:3}, respectively. Given the similar performance of three ``knockoff first" methods for SFLR and FFLR, we only employ KF1 for FGGM due to computational efficiency. We choose to report results under the OR rule, which demonstrates superior performance compared to the AND rule.
Several conclusions can be drawn from Tables \ref{tab:1}--\ref{tab:3}.
First, in all three models whether $p>n$ or $p<n,$ the knockoffs-based methods, including KF1, KF2, KF3 and TF, effectively control the empirical FDR below the target level of $q=0.2$. In contrast, GL results in significantly inflated FDR, especially for SFLR and FFLR.
Second, across all scenarios, our ``knockoff first" methods consistently achieve higher empirical powers compared to ``truncation first" competitors. These empirical findings nicely validate the heuristic arguments presented in Remarks \ref{fdrRmK} and \ref{powerRmK}.
Third, among KF1, KF2 and KF3, they exhibit similar performance in terms of FDR control and power. Due to its lowest computational cost, we recommend using KF1 in practice. 
\begin{table}[htbp]
  \caption{The empirical power and FDR for SFLR.}
  \label{tab:1}
  \centering
  \setstretch{0.8} 
  \setlength{\abovecaptionskip}{0pt} 
\setlength{\belowcaptionskip}{0pt} 
      \resizebox{5.5in}{!}{
\begin{tabular}{*{12}{c}}
  \toprule
  \multirow{2}*{$p$} & \multirow{2}*{$n$} & \multicolumn{2}{c}{KF1} & \multicolumn{2}{c}{KF2} & \multicolumn{2}{c}{KF3} & \multicolumn{2}{c}{TF} &\multicolumn{2}{c}{GL}\\
  \cmidrule(lr){3-4}\cmidrule(lr){5-6}\cmidrule(lr){7-8}\cmidrule(lr){9-10}\cmidrule(lr){11-12}
  &  & FDR &Power & FDR &Power & FDR &Power & FDR &Power& FDR &Power\\
  \midrule
  50 &  100 & 0.16  & 0.96  &0.16   &0.95  & 0.16  
     &0.96  &0.16   &0.82  &0.25  & 1.00 \\


      &  200 & 0.13  &1.00   &   0.14& 1.00 &  0.13 &1.00  & 0.17  &0.96  &0.19  &1.00  \\

  100 &  100 & 0.13  & 0.89  &0.15   &0.89  &0.16   &0.89  &  0.19 &0.68  &0.47  &1.00  \\


      &  200 & 0.18  & 1.00  & 0.19  &1.00  &0.19   &1.00  & 0.16  & 0.89 & 0.28 &1.00  \\
      
 150 &  100 & 0.08  & 0.99  &0.08   &0.99  & 0.07  
      &1.00  &  0.10 &0.79  &0.59  &1.00  \\


      &  200 & 0.19  & 1.00  & 0.17  &1.00  &0.19   &1.00  & 0.18  & 0.86 & 0.43 &1.00  \\
      
  \bottomrule
\end{tabular}
}
	\vspace{-0.4cm}
\end{table}

\begin{table}[htbp]
  \caption{The empirical power and FDR for FFLR.}
  \label{tab:2}
  \centering
  \setstretch{0.8}
  \setlength{\abovecaptionskip}{0pt} 
\setlength{\belowcaptionskip}{0pt} 
    \resizebox{5.5in}{!}{
\begin{tabular}{*{12}{c}}
  \toprule
  \multirow{2}*{$p$} & \multirow{2}*{$n$} & \multicolumn{2}{c}{KF1} & \multicolumn{2}{c}{KF2} & \multicolumn{2}{c}{KF3} & \multicolumn{2}{c}{TF} &\multicolumn{2}{c}{GL}\\
  \cmidrule(lr){3-4}\cmidrule(lr){5-6}\cmidrule(lr){7-8}\cmidrule(lr){9-10}\cmidrule(lr){11-12}
  &  & FDR &Power & FDR &Power & FDR &Power & FDR &Power& FDR &Power\\
  \midrule
   50   &  100 & 0.18  & 0.88  &0.17   &0.86  & 0.19  &0.88 &0.17   &0.69  &0.80  & 1.00 \\


       &  200 & 0.12  &0.93   &   0.13& 0.93 &  0.14 &0.93  & 0.14  &0.83  &0.20  &0.94  \\

  100  &  100 & 0.14  & 0.78  &0.14   &0.78  &0.12   &0.78  &  0.07 &0.71  &0.69  &0.81  \\


       &  200 & 0.08  & 0.96  & 0.08  &0.96  &0.08   & 0.96  & 0.07  & 0.85 & 0.88 &1.00  \\
       
        150 &  100 & 0.17  & 0.99  &0.15   &0.99  &0.19   
      &0.99  &  0.16 &0.70  &0.93  &1.00  \\


      &  200 & 0.16  & 1.00  & 0.17  &1.00  &0.18   &1.00  & 0.14  & 0.82 & 0.93 &1.00  \\
  \bottomrule
\end{tabular}
}
	\vspace{-0.2cm}
\end{table}

\begin{table}[htbp]
  \caption{The empirical power and FDR for FGGM.}
  \label{tab:3}
  \centering
  \setstretch{0.8}
  \setlength{\abovecaptionskip}{0pt} 
\setlength{\belowcaptionskip}{0pt} 
    \resizebox{4in}{!}{
\begin{tabular}{*{8}{c}}
  \toprule
  \multirow{2}*{$p$} & \multirow{2}*{$n$} & \multicolumn{2}{c}{KF1} & \multicolumn{2}{c}{TF} &\multicolumn{2}{c}{GL}\\
  \cmidrule(lr){3-4}\cmidrule(lr){5-6}\cmidrule(lr){7-8}
  &  & FDR &Power & FDR &Power& FDR &Power\\
  \midrule
   50   &  100 & 0.18   &0.74 &0.15   &0.50  &0.23  &0.75  \\

       &  200 & 0.20  &0.94 & 0.20  &0.79  &0.22  &0.94  \\

  100  &  100 & 0.17  & 0.63   &  0.12 &0.48  &0.20  &0.68  \\

       &  200 & 0.19  & 0.84    & 0.18  & 0.73 & 0.22 &0.86  \\
       
  \bottomrule
\end{tabular}
}
	\vspace{-0.3cm}
\end{table}

We also assess the finite-sample performance of competing methods for handling partially observed functional data. We generate $X_{ij}(\cdot)$ for $i=[n]$ and $j=[p]$ following the same procedure as above. We then generate the observed values $W_{ijl}$'s from (\ref{model.partial}), where the observational time points and errors $e_{ijl}$'s are independently sampled from $\text{Unif}[0,1]$ and ${\cal N}(0, 0.5^2),$ respectively. We consider the dense setting $L=51$ since brain signals in neuroimaging data are commonly measured at a dense set of points. We employ the local-linear-based pre-smoothing approach using a Gaussian kernel with the optimal bandwidth proportional to $L^{-1/5}.$ The numerical results for SFLR and FGGM are respectively presented in Tables~\ref{tab:4} and \ref{tab:5} of the Supplementary Material. Similar conclusions can be drawn compared to the results obtained for fully observed functional data from Tables~\ref{tab:1}--\ref{tab:3}.


\vspace{-0.3cm}
\section{Real data analysis}
\label{sec:real}
\vspace{-0.3cm}
\subsection{Emotion related fMRI dataset} 
\label{real1}
\vspace{-0.3cm}
In this section, we illustrate our functional model-X knockoffs selection proposal for SFLR using a publicly available brain imaging dataset obtained from the Human Connectome Project (HCP), \url{http://www.humanconnectome.org/}. This dataset comprises $n=848$ subjects of functional magnetic resonance imaging (fMRI) scans with Blood Oxygenation Level-Dependent (BOLD) signals in the brain.
We follow recent proposals, based on HCP, to model signals as multivariate random functions, thus representing each region of interest (ROI) as
one random function; see, e.g., \cite{xue2021,zapata2019,Lee2023Conditional}.
For each subject, the BOLD signals are recorded every 0.72 seconds, totalling $L=176$ observational time points (2.1 minutes). The response of interest is referred to as {\it Emotion Task Shape Acc}, which represents a continuous score measured by the Penn Emotion Recognition Test and is associated with the brain's processing of negative emotional tasks. We construct an SFLR model via (\ref{SFLR.eq}) by treating $p=34$ ROIs as functional covariates; see Table~\ref{tab:regions} of the Supplementary Material for details on the specific ROIs. Our goal is to identify ROIs that significantly influence the Emotion Task Shape Acc. For this purpose, we apply our proposed KF1 approach for SFLR with target FDR level of $q=0.2.$ To construct ${\cal H}_j$'s, we use cubic spline functions with 11 interior nodes, corresponding to $k_n=15.$ For comparison, we also implement the TF and GL methods. While TF and GL respectively select 9 and 15 ROIs, our KF1 identifies 6 ROIs, indexed by $j \in \{9, 12, 20, 22, 31, 32\}$ with sorted importance levels $W_9 > W_{31} > W_{32} > W_{20} > W_{12} > W_{22}.$ Among three competitors, these six ROIs are consistently selected, and align with findings in the existing literature. Specifically, \cite{xue2021} identified isthmus cingulate ($j=9$), lingual ($j=12$) and frontal pole ($j=31$) as important ROIs associated with Emotion Task Shape Acc. Furthermore, previous studies have found regions like pericalcarine ($j=20$), posterior cingulate ($j=22$) and temporal pole ($j=32$) to be responsible for negative emotions \cite[]{Sabatinelli2006, rolls2019cingulate, Olson2007Brain}.

\vspace{-0.3cm}
\subsection{Functional connectivity analysis}
\label{real2}
\vspace{-0.3cm}

In this section, we investigate the relationship between brain functional connectivity and fluid intelligence (gF), which represents the capacity to think and reason independently of acquired knowledge. The dataset, obtained from HCP, consists of fMRI scans and the corresponding gF scores, determined based on participants’ performance on the Raven’s Progressive Matrices. We focus on $n_{\text{low}}=73$ subjects with low fluid intelligence scores (gF$\leq 8$) and $n_{\text{high}}=85$ subjects with high scores (gF$\geq 23$). In an analogy to Section~\ref{real1}, we treat the BOLD signals at different ROIs as multivariate functional data, considering $p=83$ ROIs across three well-established modules in neuroscience literature \cite[]{finn2015}: the medial frontal module (29 ROIs), frontoparietal module (34 ROIs), and default mode module (20 ROIs). The signals for each subject at each ROI are collected every 0.72 seconds at $L=1200$ measurement locations (14.4 minutes). To exclude unrelated frequency bands in resting-state functional connectivity, we apply ICA$+$FIX preprocessed pipeline and use a standard band-pass filter between 0.01 and 0.08 Hz \cite[]{Glasser2016}. For our analysis, we employ the proposed KF1 method for FGGM under the OR rule to construct brain functional connectivity networks, which depict the conditional dependence structures among respective ROIs within each of the three modules. We use cubic splines with 31 interior nodes, resulting in each ${\cal H}_j$ being spanned by $k_n=35$ spline functions. We continue to set the target FDR level at $q=0.2.$

\begin{figure}[h]
\centering
\vspace{-0.5cm}
\begin{subfigure}{0.49\linewidth}
 \centering
  \includegraphics[width=7cm]{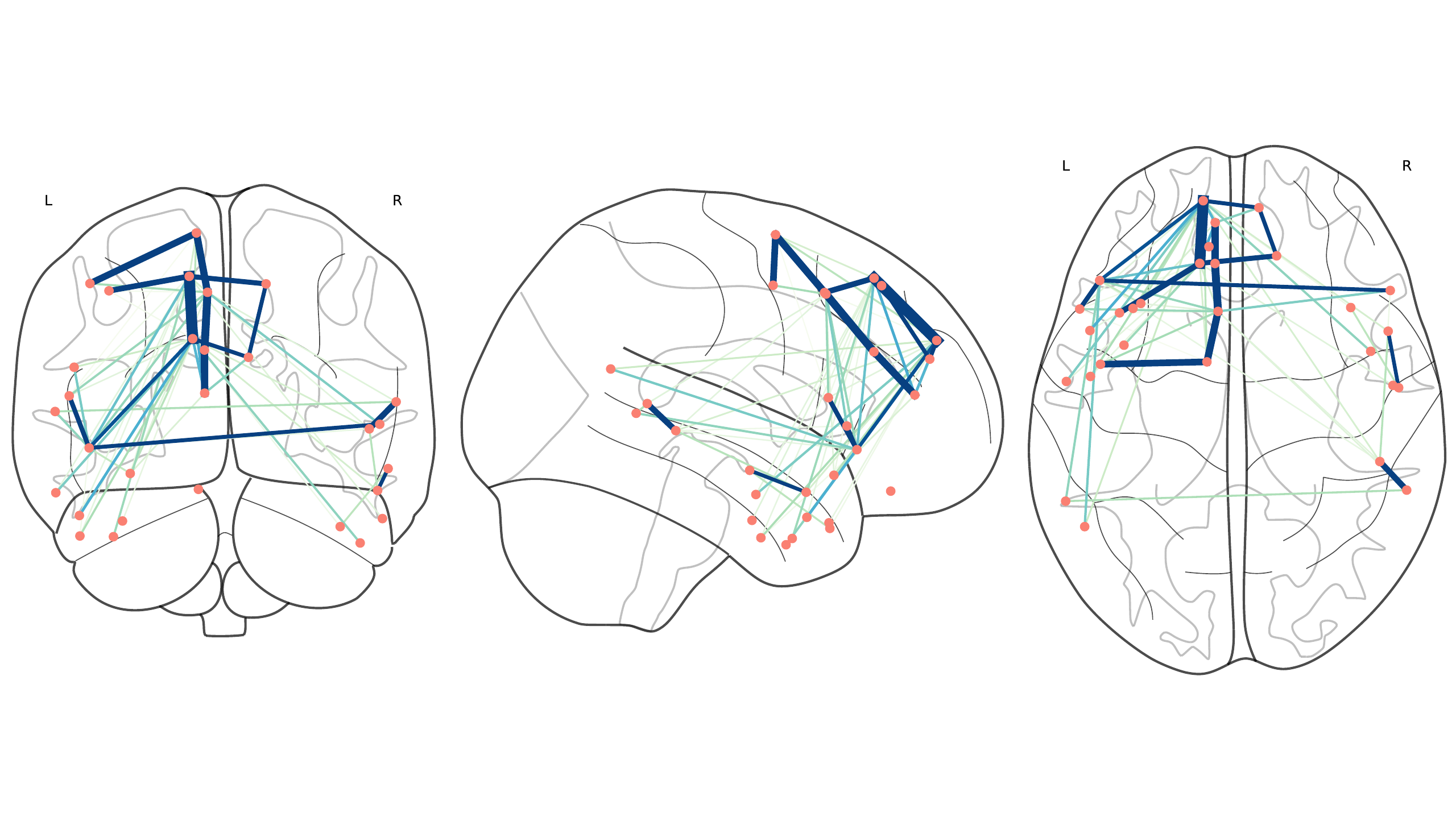}
  \vspace{-0.95cm} 
  \caption{gF$\leq 8$: the medial frontal module } 
 \vspace{-0.09cm}
  \includegraphics[width=7cm]{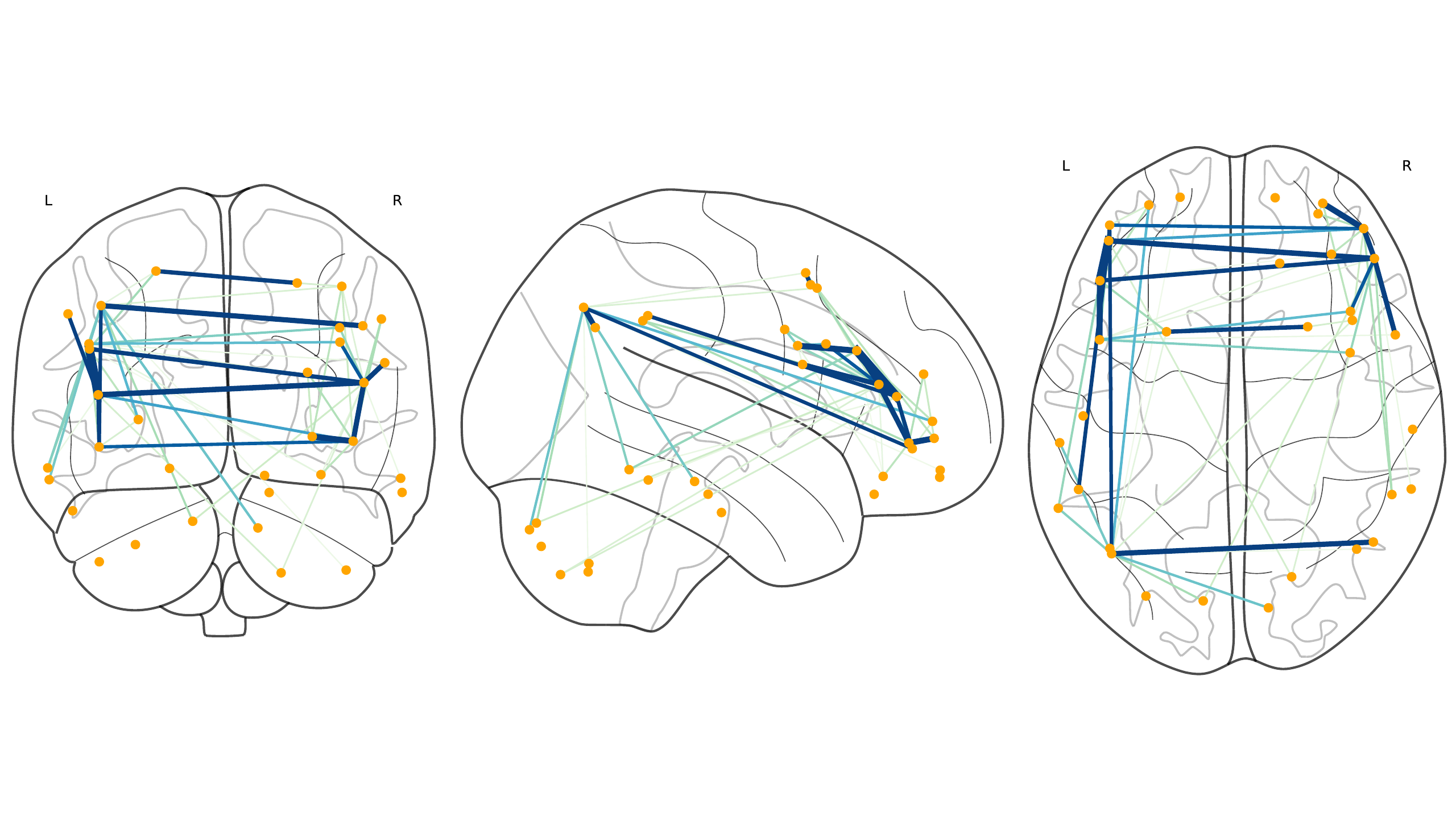}
  \vspace{-0.95cm} 
  \caption{gF$\leq 8$: the frontoparietal module} 
  \vspace{-0.09cm}
  \includegraphics[width=7cm]{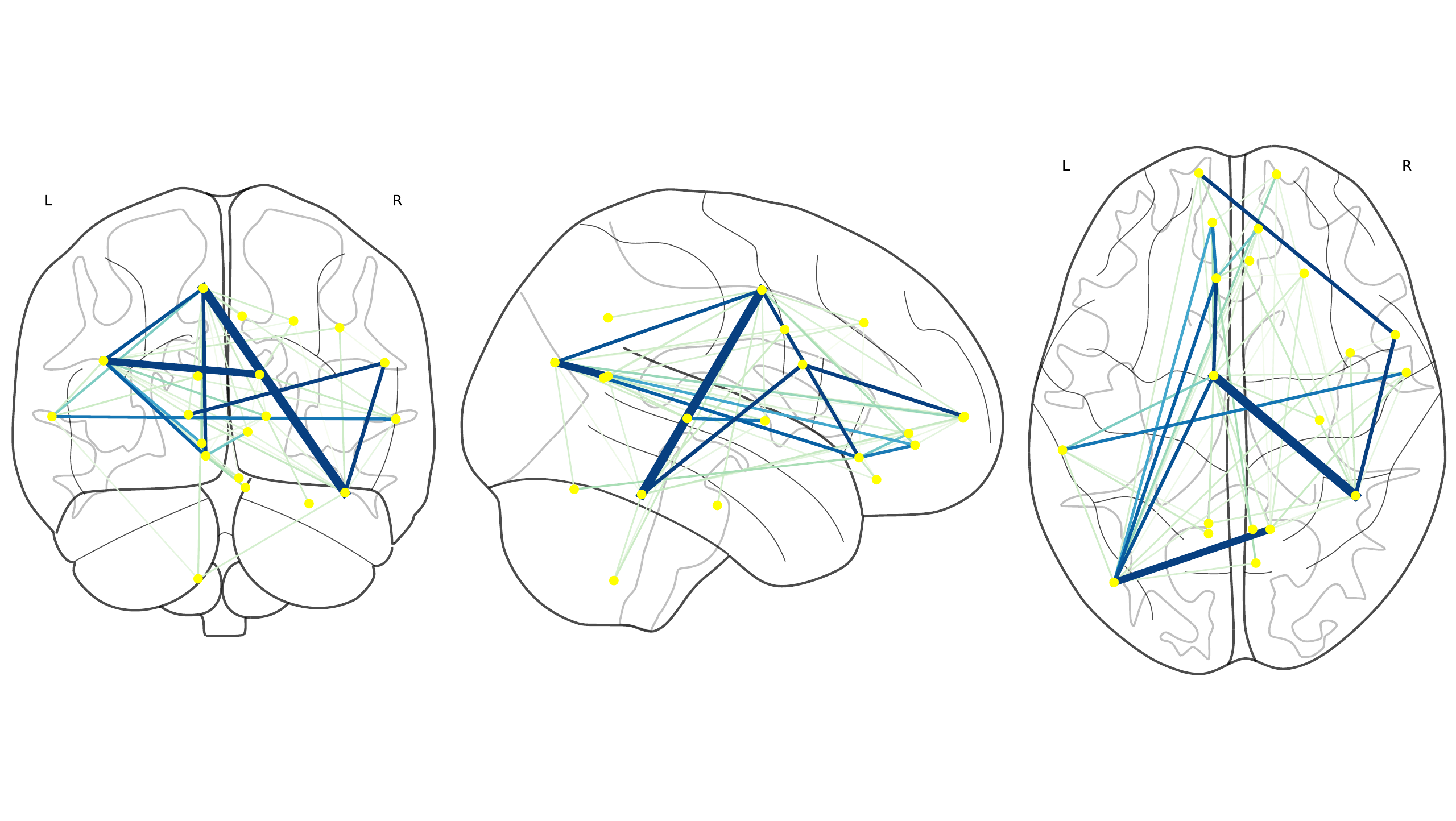}
    \vspace{-0.95cm} 
  \caption{gF$\leq 8$: the default mode module } 
\end{subfigure}
\centering
\begin{subfigure}{0.49\linewidth}
\centering
\vspace{-0.1cm}
  \includegraphics[width=7cm]{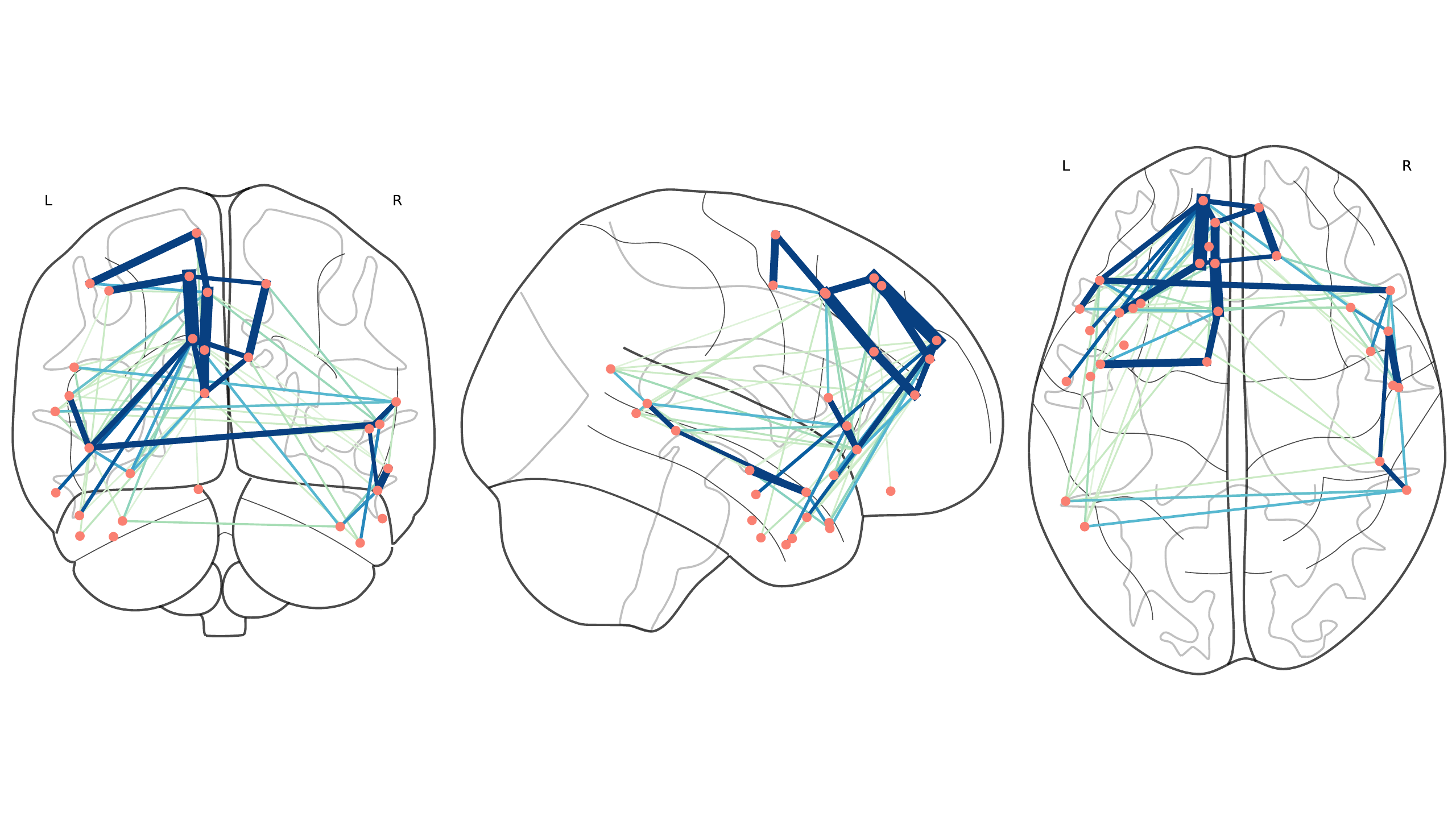}
    \vspace{-0.95cm} 
  \caption{gF$\geq 23$: the medial frontal module } 
   \vspace{-0.09cm}
  \includegraphics[width=7cm]{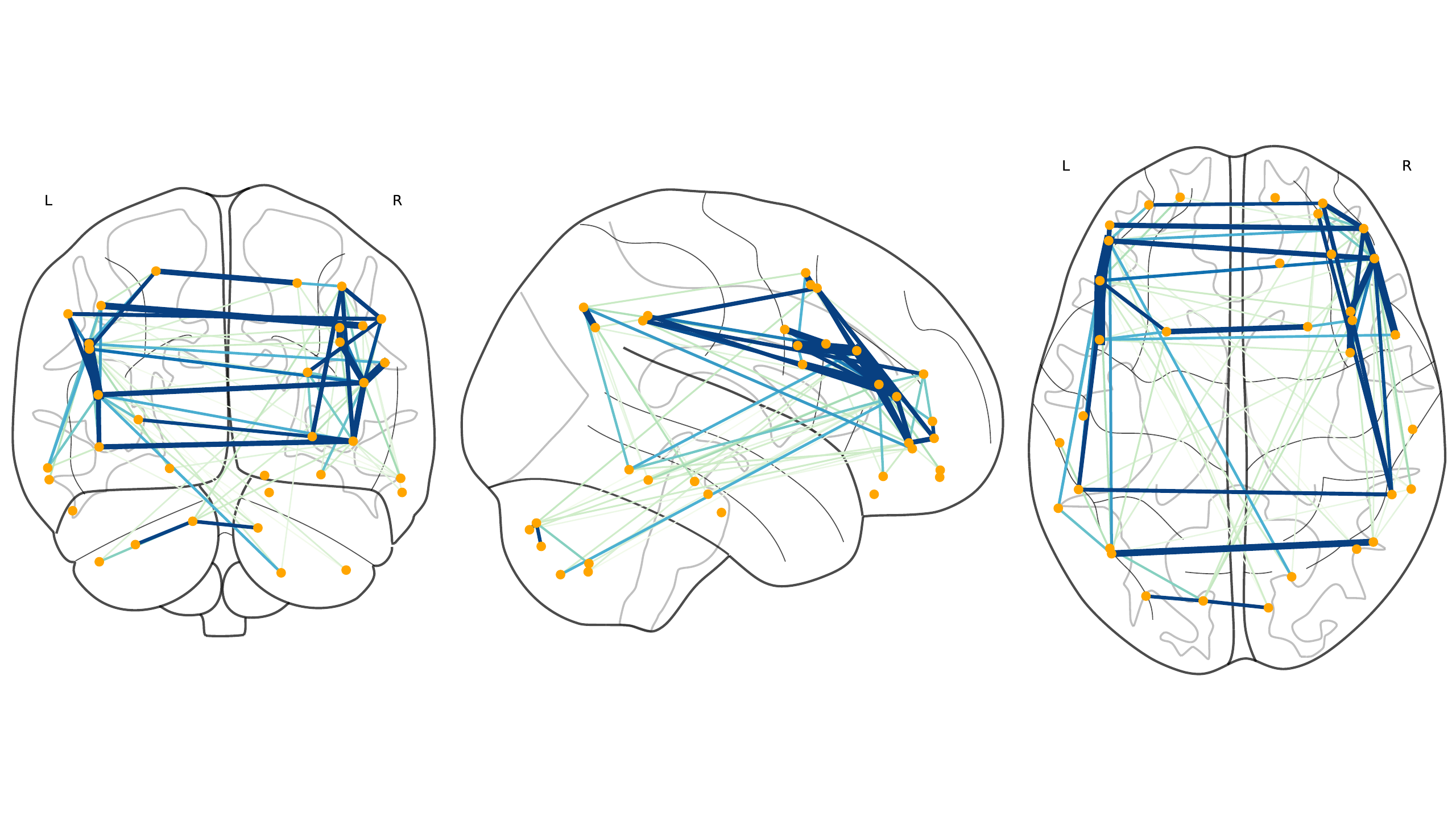}
    \vspace{-0.95cm} 
  \caption{gF$\geq 23$: the frontoparietal module } 
   \vspace{-0.09cm}
  \includegraphics[width=7cm]{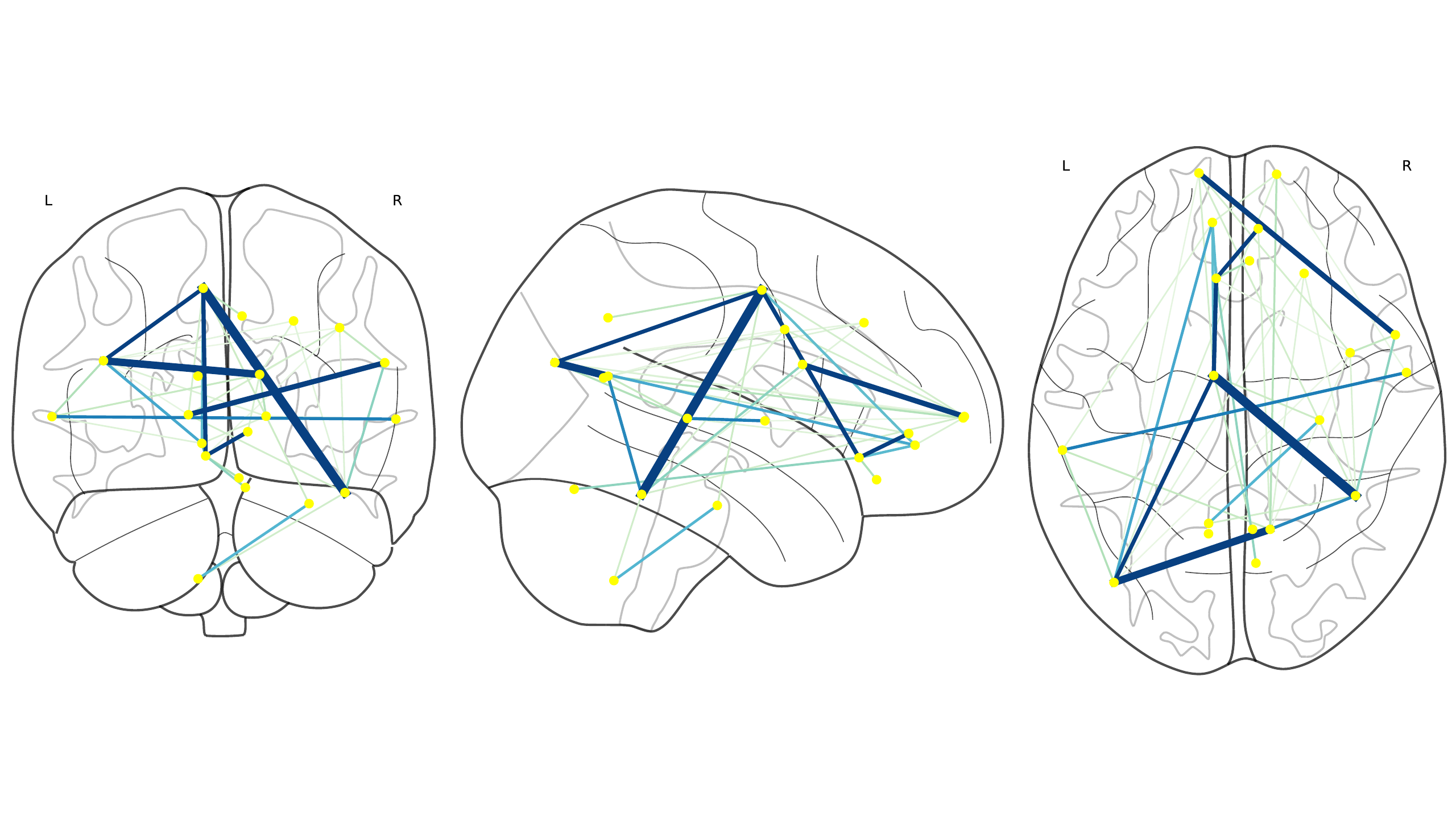}
    \vspace{-0.95cm} 
  \caption{gF$\geq 23$: the default mode module} 
\end{subfigure}
\centering
\caption{\label{hcp_network}{\small The connectivity strengths at fluid intelligence gF$\leq 8$ (left column) and gF$ \geq 23$ (right column).
Salmon, orange and yellow nodes represent the ROIs in the medial frontal, frontoparietal and  default mode modules, respectively. 
The edge color ranging from light to dark corresponds to the value of $W_{jk}$ from small to large.}}
\vspace{-0.1cm}
\end{figure}

Figure~\ref{hcp_network} displays the functional connectivity networks identified for subjects with gF $\leq 8$ and gF $\geq 23.$ To assess the impact of gF on functional connectivity, we measure connectivity strength through $W_{jk}$ for $j,k \in [p]$ with larger values yielding stronger connectivity. A few patterns are apparent. First, we observe increased connectivity and strength in the medial frontal and frontoparietal modules for subjects with gF$ \geq 23.$ This observation aligns with and supports the existing literature, which has reported a strong positive association between the functional connectivity and intellectual performance in these two modules \cite[]{van2009}.
Second, it is evident that the default mode module exhibits declined connectivity and strength for subjects with higher intelligence scores, which is in line with the previous finding in neuroscience that reduced activity in the default mode module is associated with improved cognitive performance \cite[]{anticevic2012}.

\appendix
\spacingset{1}
\bibliography{ref}
\bibliographystyle{dcu}

\newpage
\linespread{1.6}\selectfont
\begin{center}
	{\noindent \bf \large Supplementary material  to ``Functional knockoffs selection with applications to functional data analysis in high dimensions"}\\
\end{center}
	\begin{center}
		{\noindent Xinghao Qiao, Mingya Long and Qizhai Li}
	\end{center}
\bigskip

\setcounter{page}{1}
\setcounter{section}{0}
\renewcommand\thesection{\Alph{section}}
\newtheorem{lemmaS}{Lemma}
\setcounter{lemmaS}{0}
\renewcommand{\thelemmaS}{\Alph{section}\arabic{lemmaS}}
\setcounter{equation}{0}
\renewcommand{\theequation}{S.\arabic{equation}}

This supplementary material contains all technical proofs in Section~\ref{supp.pf.FDR}, 
additional methodological derivations in Section~\ref{supp.details} and additional empirical results in Section~\ref{sec:additional.emp}.

\section{Technical proofs}\label{supp.pf.FDR} 
\subsection{Proof of Theorem~\ref{thm_fdr_sflr} }
To prove Theorem~\ref{thm_fdr_sflr}, we firstly present some technical lemmas with their proofs.

\begin{lemmaS}
\label{lemmaEx.func}
For any subset $G\subseteq S^{c},$  $\big(\bX(\cdot)^\T, \widetilde{\bX}(\cdot)^\T\big) ~\Big|~   Y   \overset{D}{=} \big(\bX(\cdot)^\T, \widetilde{\bX}(\cdot)^\T\big)_{\tswap(G)} ~\Big|~   Y.$
\end{lemmaS}

\begin{proof}
It is  equivalent to show that $\big(\bX(\cdot)^\T, \widetilde{\bX}(\cdot)^\T,  Y\big)   \overset{D}{=} \big(\{\bX(\cdot)^\T, \widetilde{\bX}(\cdot)^\T\}_{\tswap(G)},  Y\big) $.
Moreover, since $\big( \bX(\cdot)^\T, \widetilde{ \bX}(\cdot)^\T\big)  \overset{D}{=} \big( \bX(\cdot)^\T, \widetilde{ \bX}(\cdot)^\T\big)_{\tswap(G)}$ from Property~(i) in Definition~\ref{def1},
it suffices to prove that 
\begin{equation}\label{eq.eX}
   Y  ~\Big|~\big( \bX(\cdot)^\T, \widetilde{ \bX}(\cdot)^\T\big)_{\tswap(G)}   \overset{D}{=}  Y ~\Big|~ \big( \bX(\cdot)^\T, \widetilde{ \bX}(\cdot)^\T\big).
\end{equation}
By Properties~(i) and (ii) in Definition~\ref{def1}, we have 
$$
\begin{array}{cll}
  Y ~\Big|~ \big[\{ \bX(\cdot)^\T, \widetilde{ \bX}(\cdot)^\T\}_{\tswap(G)}= \{\bx(\cdot)^\T, \widetilde\bx(\cdot)^\T\} \big] &\overset{D}{=}& Y ~\Big|~ \big[ \{\bX(\cdot)^\T, \widetilde{ \bX}(\cdot)^\T\}= \{\bx(\cdot)^\T, \widetilde\bx(\cdot)^\T\}_{\tswap(G)} \big] \\  
  &\overset{D}{=}& Y~\Big|~\big\{\bX(\cdot)^\T = \bx'(\cdot)^\T\big\},
\end{array}
$$ where $\bx(\cdot) = 
\big(x_1(\cdot),\dots,x_p(\cdot)\big)^\T$, $\widetilde\bx(\cdot) = 
\big(\widetilde x_1(\cdot),\dots,\widetilde x_p(\cdot)\big)^\T$, and the $j$th element of $\bx'(\cdot)$ is $x'_j(\cdot) = \widetilde{x}_j(\cdot)$ if $j \in G$
 and $x'_j(\cdot) = x_j(\cdot)$ otherwise.
 Without loss of generality, assuming  that $G = \{1,2,\dots, m\} \subseteq S^{c}$, we have
 \begin{equation}\label{seqminus}
 \begin{split}
       Y~\Big|~\big\{\bX(\cdot)^\T = \bx'(\cdot)^\T\big\}& \overset{D}{=}  Y~\Big|~\big\{\bX_{2:p}(\cdot)^\T = \bx_{2:p}'(\cdot)^\T\big\} \\
                          & \overset{D}{=}  Y ~\Big|~ \big\{X_1(\cdot) = x_1(\cdot), \bX_{2:p}(\cdot)^\T = \bx_{2:p}'(\cdot)^\T\big\},
 \end{split}
 \end{equation}
where we use $\bX_{2:p}(\cdot)$ to denote $\big(X_{2}(\cdot), \dots, X_{p}(\cdot)\big)^{\T}$, $\bx_{2:p}'(\cdot) = \big(x'_{2}(\cdot), \dots, x'_{p}(\cdot)\big)^{\T}$ and the above two equalities hold since $Y$ and $X_1(\cdot)$ are independent conditional on $\bX_{2:p}(\cdot)$.
(\ref{seqminus}) shows that $ Y  ~\Big|~\big( \bX(\cdot)^\T, \widetilde{ \bX}(\cdot)^\T\big)_{\tswap(G)}    \overset{D}{=} Y ~\Big|~   \big(\bX(\cdot)^\T, \widetilde{ \bX}(\cdot)^\T \big)_{\tswap (G  \backslash \{1\} ) }.$
Repeating this strategy until $G$ is empty, we obtain that (\ref{eq.eX}) holds, which completes our proof. 
Note that the response $Y$ in our proof can be either scalar or functional, we use the same notation for simplicity. 
\end{proof}

We will next demonstrate that the estimated  FPC scores in SFLR satisfy Condition \ref{conscore}. 

\begin{corollary}
\label{CorSF}
For any subset $G \subseteq S^{c}$, $\big({\widehat{\bxi}^\T_i}, {\widecheck{{\bxi}}^\T_i} \big) ~\Big|~   Y_i   \overset{D}{=} \big({\widehat{\bxi}^\T_i}, {\widecheck{{\bxi}}^\T_i}\big)_{\text{swap}(G)}~\Big|~   Y_i $.
\end{corollary}
\begin{proof} 
It is equivalent to show that $({\widehat{\bxi}_i}^\T, {\widecheck{{\bxi}}_i}^\T,    Y_i) \overset{D}{=} \big\{({\widehat{\bxi}_i}^\T, {\widecheck{{\bxi}}_i} ^\T)_{\text{swap}(G)},   Y_i \big\}$. 
Referring to the result in Lemma~\ref{lemmaEx.func} and considering that $\big\{\big(\bX_i(\cdot)^\T, \widetilde{\bX}_i(\cdot)^\T,   Y_i\big)\big\}_{i\in[n]}$ are i.i.d., we can establish that $\big(\bX_i(\cdot)^\T, \widetilde{\bX}_i(\cdot)^\T\big) ~\Big|~   Y_i   \overset{D}{=} \big(\bX_i(\cdot)^\T, \widetilde{\bX}_i(\cdot)^\T\big)_{\tswap(G)} ~\Big|~   Y_i,$ where the response $Y_i$ is scalar in SFLR.
This implies that $\big(\bX_i(\cdot)^\T, \widetilde{\bX}_i(\cdot)^\T,   Y_i\big)   \overset{D}{=} \big(\{\bX_i(\cdot)^\T, \widetilde{\bX}_i(\cdot)^\T\}_{\tswap(G)}, Y_i\big).$ This further implies that, for any $\mathbf{t} \in (\boldsymbol{\mathcal{H}}^2, \eR)$ and $\iota  = \sqrt{-1}$, 
\begin{equation*}
    \eE \Big[\exp\Big\{\iota \big\langle \mathbf{t} , (\bX_i^\T, \widetilde{\bX}_i^\T,   Y_i)^\T \big\rangle \Big\}\Big] = \eE \Big[\exp\Big\{\iota \big\langle \mathbf{t} , \big\{(\bX_i^\T, \widetilde{\bX}_i^\T)_{\tswap(G)}, Y_i\big\}^\T \big\rangle\Big\}\Big].
\end{equation*}
Given 
${\bf t} = (\mathbf{c}_1^\T \widehat\bphi_{1},\dots, \mathbf{c}_p^\T \widehat\bphi_{p}, \widetilde{\mathbf{c}}_1^\T \widehat\bphi_{1}, \dots, \widetilde{\mathbf{c}}_p^\T \widehat\bphi_{p}, c_Y)^\T$, where 
$\mathbf{c}_{j} = (c_{j1}, \dots, c_{jd_j}, 0 ,0 ,\dots )^\T$, $\widetilde{\mathbf{c}}_{j} = (\tilde c_{j1},  \dots, \tilde c_{jd_j}, 0 ,0 ,\dots )^\T$,  
and $\widehat{\bphi}_j = (\hat\phi_{j1}, \hat\phi_{j2}, \dots )^\T$ for  $j \in [p],$ we can establish that 
\begin{equation}\label{co12}
  \eE \Big[\exp\Big\{\iota \big\langle \mathbf{c} , \big({\widehat{\bxi}_i}^\T, {\widecheck{{\bxi}}^\T_i}  ,Y_i \big)^\T \big\rangle \Big\} \Big| \mathbf{t} \Big]
    = \eE \Big[\exp\Big\{\iota \big\langle \mathbf{c}  , \big\{({\widehat{\bxi}_i^\T}, {\widecheck{{\bxi}}_i^\T})_{\tswap(G)} ,Y_i \big\}^\T \big\rangle\Big\}  \Big| \mathbf{t}  \Big], 
\end{equation}
for any $\mathbf{c} = (\Bar{\mathbf{c}}_{1}^{\T}, \dots, \Bar{\mathbf{c}}_{p}^\T,\widetilde{\Bar{\mathbf{c}}}_{1}^\T, \dots, \widetilde{\Bar{\mathbf{c}}}_{p}^\T,c_Y)^\T$, 
where 
$\Bar{\mathbf{c}}_{j} = (c_{j1}, \dots, c_{jd_j})^\T$,
$\widetilde{\Bar{\mathbf{c}}}_{j} = (\tilde c_{j1},  \dots, \tilde c_{jd_j})^\T$. 
By (\ref{co12}) and the total expectation formula, the joint characteristic function of $({\widehat{\bxi}_i}^\T, {\widecheck{{\bxi}}_i}^\T,    Y_i) $ is equal to that of $ \big\{({\widehat{\bxi}_i}^\T, {\widecheck{{\bxi}}_i}^\T )_{\tswap(G)},   Y_i \big\},$ which implies that $({\widehat{\bxi}_i}^\T, {\widecheck{{\bxi}}_i}^\T,    Y_i) \overset{D}{=} \big\{({\widehat{\bxi}_i}^\T, {\widecheck{{\bxi}}_i} ^\T)_{\tswap(G)},   Y_i \big\},$ 
and thus completes our proof. 
\end{proof}

Following the above corollary, we next validate Lemma~\ref{lemmacoin} in the context of SFLR. 
 \begin{proof}[Proof of Lemma~\ref{lemmacoin} in SFLR]
    Denote $\widehat {\bXi} = (\widehat\bxi_{1},\dots,\widehat\bxi_{n})^\T\in \eR^{n \times \sum_j d_j}$ , $ \widecheck {\boldsymbol \Xi} = (\widecheck\bxi_{1},\dots,\widecheck\bxi_{n})^\T \in \eR^{n \times \sum_j d_j}$ and $\bY= (Y_1, \dots, Y_n)^\T \in \eR^n$. 
    First, note
    $W_j( \widehat {\boldsymbol \Xi}, \widecheck{ {\boldsymbol \Xi}},\bY)$ is a function of $\widehat {\boldsymbol \Xi}, \widecheck{ {\boldsymbol \Xi}}$ and $\bY$.  
    Since $W_j = \| \widehat\bbb_j\|- \| \widehat\bbb_{p+j}\|$ in SFLR for each $j \in [p]$,  
      the flip-sign property of $W_j( \widehat {\boldsymbol \Xi}, \widecheck{ {\boldsymbol \Xi}},\bY)$ holds. That is,  for any subset $G  \subseteq [p],$ 
$$
W_j\big(\{\widehat {\boldsymbol \Xi}, \widecheck{ {\boldsymbol \Xi}} \}_{\tswap(G)},\bY\big)  = \left\{
\begin{array}{cl}
 W_j\big(\{\widehat {\boldsymbol \Xi}, \widecheck{ {\boldsymbol \Xi}} \}_{\tswap(G)},\bY\big) , & j \notin G   \\
 - W_j\big(\{\widehat {\boldsymbol \Xi}, \widecheck{ {\boldsymbol \Xi}}\}_{\tswap(G)},\bY\big), & j \in G,
\end{array}
\right.
$$ 
where $(\widehat {\boldsymbol \Xi}, \widecheck{ {\boldsymbol \Xi}})_{\tswap(G)}$ is obtained from $(\widehat {\boldsymbol \Xi}, \widecheck{ {\boldsymbol \Xi}})$ by swapping the corresponding estimated FPC scores of $X_j(\cdot)$ and $\widetilde{X}_j(\cdot)$ for each $j\in G$. 
{ Second, let $\boldsymbol W = (W_1,\dots,W_p)^\T$, and consider any subset $G \subseteq S^{c}$. By swapping variables in $G$, we define 
$$
\boldsymbol W_{\tswap{(G)}}  \overset{\triangle}{=} \Big(W_1\big(\{\widehat {\boldsymbol \Xi}, \widecheck{ {\boldsymbol \Xi}}\}_{\tswap(G)},\bY\big), \dots, W_p\big(\{\widehat {\boldsymbol \Xi}, \widecheck{ {\boldsymbol \Xi}} \}_{\tswap(G)},\bY\big)\Big)^\T.
$$ 
By Corollary~\ref{CorSF}, it follows that $\big((\widehat\bXi,\widecheck{\bXi}),\bY\big)   \overset{D}{=} \big((\widehat\bXi,\widecheck{\bXi})_{\text{swap}(G)},\bY\big)$, which consequently implies $\boldsymbol W \overset{D}{=} \boldsymbol W_{\tswap{(G)}}$.} 
Finally, consider $S^c_- = \{j \in S^{c}: \delta_j = -1 \}$, where $\boldsymbol \delta = (\delta_1, \dots, \delta_p)^\T$ represents a sequence of independent random variables. These $\delta_j$ variables follow the Rademacher distribution if $j \in S^{c}$ and $\delta_j =1$ otherwise. 
Through the first step, we derive $\boldsymbol W _{\text{swap}(S^c_-)} = (\delta_1 W_1, \dots,\delta_p W_p)^\T.$ Subsequently, from the second step, we obtain $\boldsymbol W _{\text{swap}(S^c_-)} \overset{D}{=} \boldsymbol W.$ 
Combining the above results, we have $(\delta_1 W_1, \dots,\delta_p W_p)^\T \overset{D}{=} \boldsymbol W,$
which completes the proof. 
 \end{proof}

\begin{lemmaS}
    \label{lemmaeqSFLR}
   Suppose that  Condition~\ref{irrSFLR} holds. Then $\|\beta_{j}\| = 0$  if and only if $j\in S^c$, where $S^c$ is defined in (\ref{f.nullset}). 
\end{lemmaS}

Note that the proof of Lemma~\ref{lemmaeqSFLR} follows the same argument as that of Lemma~\ref{lemmaeqFFLR}. We will provide detailed proof of Lemma~\ref{lemmaeqFFLR} in Section~\ref{pf.thm_fflr} and omit the proof of Lemma~\ref{lemmaeqSFLR} here.
 We are now ready to prove Theorem~\ref{thm_fdr_sflr}. 
 \begin{proof}[Proof of Theorem \ref{thm_fdr_sflr}]
  Provided that Lemma~\ref{lemmacoin} holds in SFLR, it implies that the signs of the null statistics are distributed as i.i.d. coin flips. Referring to Theorem 3.4 of \cite{candes2018}, 
  we obtain that
\begin{equation*} \label{eqsflr}
{\mathbb E}\left[\frac{|\widehat{S}_1  \cap S^c |}{|\widehat{S}_1|}\right] \leq  q,~~~
{\mathbb E}\left[\frac{|\widehat{S}_0  \cap S^c|}{| \widehat{S}_0|+1/q}\right] \leq q, 
\end{equation*} 
where $S^c$ is defined in (\ref{f.nullset}).
By Lemma~\ref{lemmaeqSFLR}, we establish that $S^c$ in (\ref{f.nullset}) is equivalent to the set $\{j: \|\beta_j\|=0\},$ whose complement is defined in (\ref{support.SFLR}).
This equivalence implies the effective FDR control in SFLR, which completes our proof.
\end{proof}

\subsection{Proof of Theorem~\ref{thm_power_sflr}} 
Before presenting technical lemmas used in the proof of Theorem~\ref{thm_power_sflr}, we firstly introduce some notation.
For $j, k \in [p]$, denote ${\sigma}_{jklm} = \eE [\xi_{ijl} \xi_{ikm} ]$ and its estimator $\hat{\sigma}_{jk lm} = n^{-1} \sum_{i=1}^n \hat\xi_{ijl} \hat\xi_{ikm}$ for $l,m \in [d].$
For $j \in [p]$, $k \in [2p]\backslash [p]$, denote ${\sigma}_{jklm} = \eE [\xi_{ijl} \tilde\xi_{i(k-p)m} ]$ and its estimator $\hat{\sigma}_{jk lm} = n^{-1} \sum_{i=1}^n \hat\xi_{ijl} \check\xi_{i(k-p)m}.$ 
For $j, k \in [2p]\backslash [p]$, denote ${\sigma}_{jklm} = \eE [\tilde\xi_{i(j-p)l} \tilde\xi_{i(k-p)m} ]$ and its estimator estimator $\hat{\sigma}_{jk lm} = n^{-1} \sum_{i=1}^n \check\xi_{i(j-p)l} \check\xi_{i(k-p)m}.$ 
For $j \in [p]$, $l \in [d]$, denote ${\sigma}^{X,Y}_{j l} = \eE [ \xi_{ijl} Y_{i} ]$ and its estimator $\hat{\sigma}^{X,Y}_{jl} = n^{-1} \sum_{i=1}^n \hat\xi_{ijl} Y_{i}.$ 
For $j \in [2p]\backslash [p]$, $l \in [d]$, denote ${\sigma}^{X,Y}_{j l} = \eE [ \tilde\xi_{i(j-p)l} Y_{i} ]$ and its estimator $\hat{\sigma}^{X,Y}_{jl} = n^{-1} \sum_{i=1}^n \check\xi_{i(j-p)l} Y_{i}$. 
Let ${\bD} = \tdiag({\bD}_{1},\dots, {\bD}_{p}, {\bD}_{1},\dots, {\bD}_{p}) \in \eR^{2pd \times 2pd}$ with ${\bD}_j = \diag(\omega_{j1}^{1/2},\dots,\omega_{jd}^{1/2}) \in \eR^{d\times d}$ for $j \in [p]$ and its estimator $\widehat{\bD} = \tdiag(\widehat{\bD}_{1},\dots, \widehat{\bD}_{p}, \widehat{\bD}_{1},\dots, \widehat{\bD}_{p})$ with $\widehat{\bD}_j = \diag(\hat\omega_{j1}^{1/2},\dots,\hat\omega_{jd}^{1/2})$. For a matrix $\bA = (A_{jk}) \in \eR^{p\times q},$ we denote
its elementwise $\ell_{\infty}$ norm as $\|\bA\|_{\max} = {\max_{j,k}} |A_{j k}|$. 
For a block matrix $\bB = (\bB_{jk}) \in \eR^{p_1 q_1 \times p_2 q_2}$ with its $(j,k)$-th block $\bB_{jk} \in \eR^{q_1 \times q_2}$, we define its block versions of elementwise $\ell_{\infty}$ and matrix $\ell_{1}$ norms by $\|\bB\|^{(q_1 \times q_2)}_{\max} = \max_{j,k} \|\bB_{jk}\|_\tF$ and $\|\bB \|_1^{(q_1 \times q_2)} = \max_{k} \sum_{j} \|\bB_{jk} \|_\tF$, respectively.

\begin{lemmaS} \label{lemmaC1SFLR} 
  Suppose that Condition~\ref{cond_eigen_SFLR} holds. 
  If $n \gtrsim  d^{4\alpha +2} \log(pd)$, then there exist some positive constants $\tilde{c}_1$ and $\tilde{c}_2$ such that,  with probability greater than $1-\tilde{c}_1(pd)^{-\tilde{c}_2}$,  the estimates $\{\hat{\sigma}_{jk lm}\}$ satisfy 
    \begin{equation}\label{sd.XX}
        \max_{j,k \in [2p]\atop l,m \in [d]} \frac{|\hat{\sigma}_{jk lm}  - {\sigma}_{jk lm}|}{(l\vee m)^{\alpha+1} \omega_{jl}^{1/2}\omega_{km}^{1/2}} \lesssim \sqrt{\frac{\log(pd)}{n}}. 
    \end{equation}
     Suppose that Conditions~\ref{cond_error_SFLR} and~\ref{cond_eigen_SFLR} hold .   If $n \gtrsim  d^{3\alpha +2} \log(pd)$, then there exist some positive constants $\tilde{c}_3$ and $\tilde{c}_4$ such that,  with probability greater than $1-\tilde{c}_3(pd)^{-\tilde{c}_4}$,  the estimates $\{\hat{\sigma}_{jl}^{X,Y}\}$ satisfy 
     \begin{equation}\label{sd.XY.SF}
         \max_{j\in [2p]\atop l\in [d]} \frac{| \hat{\sigma}^{X,Y}_{j l} - {\sigma}^{X,Y}_{j l}|}{l^{\alpha+1} \omega_{jl}^{1/2}}   \lesssim  \sqrt{\frac{\log(pd)}{n}} .
     \end{equation}
\end{lemmaS}

\begin{proof} 
(\ref{sd.XX}) is a direct deduction from 
Theorem 4 in \textcolor{blue}{Guo and Qiao} (\textcolor{blue}{2023}) 
\begin{equation*}\label{sd.XX.M}
        \max_{j,k \in [2p]\atop l,m \in [d]} \frac{|\hat{\sigma}_{jk lm} - {\sigma}_{jk lm} |}{(l\vee m)^{\alpha+1} \omega_{jl}^{1/2}\omega_{km}^{1/2}} \lesssim \mathcal{M}_1^X \sqrt{\frac{\log(pd)}{n}},
    \end{equation*}
where, under no serial dependence, the functional stability measure of $\{\bX_{i}(\cdot)\}_{i\in n}$ is $\mathcal{M}_1^X=1$.  
 (\ref{sd.XY.SF}) can be derived from Proposition 1 in \textcolor{blue}{Fang et al.} (\textcolor{blue}{2022}) as 
\begin{equation*}\label{sd.XY.SF.M}
         \max_{j\in [2p]\atop l\in [d]} \frac{|\hat{\sigma}^{X,Y}_{j l}- {\sigma}^{X,Y}_{j l}  |}{l^{\alpha+1} \omega_{jl}^{1/2}}   \lesssim  \mathcal{M}_{X,Y} \sqrt{\frac{\log(pd)}{n}},
\end{equation*}
where the measure of dependence between $\{\bX_{i}(\cdot)\}_{i\in [n]}$ and $\{Y_i\}_{i\in [n]}$ is 
$\mathcal{M}_{X,Y} = \mathcal{M}_1^X + \mathcal{M}_1^Y + \mathcal{M}_{1,1}^{X,Y}.$
Under no serial dependence, $\mathcal{M}_1^X=\mathcal{M}_1^Y=1.$ 
It then suffices to establish the boundedness of $\mathcal{M}_{1,1}^{X,Y}$ to verify the validity of (\ref{sd.XY.SF}). 
Define that $\bSigma_{XY}(\cdot) = \cov(\bX(\cdot),Y)$ and $\Sigma_{YY} = \var(Y)$. By (\ref{exX}), the cross-spectral stability measure $\mathcal{M}_{1,1}^{X,Y}$ satisfies 
\begin{equation}\label{sta.def} 
\begin{split}
     \mathcal{M}_{1,1}^{X,Y} &= \mathop{esssup}\limits_{ \bPhi \in \boldsymbol{\mathcal{H}}_0, \|\bPhi\|_0\leq 1, v\in\eR_0} \frac{\big|\langle  \bPhi, \bSigma_{XY} v   \rangle\big|}{\sqrt{\langle  \bPhi, \bSigma_{XX}(\bPhi)  \rangle} \sqrt{\Sigma_{YY}v^2} }\\
     &= \mathop{esssup}\limits_{\bPhi \in \boldsymbol{\mathcal{H}}_0, \|\bPhi\|_0\leq 1, v\in\eR_0}\frac{\big|\cov\big(\sum_{j=1}^{p}  \sum_{l=1}^\infty \langle \phi_{jl}, \Phi_{j} \rangle \xi_{jl} , Y v\big) \big|}{\sqrt{\var\big(\sum_{j=1}^{p}\sum_{l=1}^\infty \langle \phi_{jl}, \Phi_{j} \rangle \xi_{jl}\big)}\sqrt{\var\big( Y v\big)}}  \leq 1,  
\end{split}
\end{equation}
where $\bPhi = (\Phi_1,\dots,\Phi_p)^{\T}$, 
     $\boldsymbol{\mathcal{H}}_0 = \big\{ \bPhi \in \boldsymbol{\mathcal{H}}: \langle  \bPhi, \bSigma_{XX}(\bPhi)  \rangle \in (0,\infty)\big\}$, 
     $\|\bPhi\|_0 = \sum_{j=1}^p I(\| \Phi_j \|_{\cS} \neq 0)$, and 
      $\eR_0= \big\{ v \in \eR:   \Sigma_{YY}v^2  \in (0,\infty)\big\}$. 
We complete the proof of this lemma. 
\end{proof}


\begin{lemmaS}\label{lemmaC2SFLR}
   Suppose that  Conditions~\ref{cond_inf_SFLR}--\ref{cond_eigen_SFLR} hold. Denote ${\boZ} =\big({{\bXi}}, {{\widetilde{\bXi}}}\big) \in \eR^{n \times 2pd}$ and 
$\widehat{\boZ} =\big({\widehat{\bXi}}, {\widecheck{{\bXi}}}\big) \in \eR^{n \times 2pd}$.  If $n \gtrsim  d^{4\alpha +2} \log(pd)$, then there exist some positive constants $c_z,\tilde{c}_5,\tilde{c}_6$ such that 
      $$
    n^{-1} \btheta^{\T} \big\{\widehat{\bD}^{-1}\big(\widehat{\boZ}^\T\widehat{\boZ}\big)\widehat{\bD}^{-1} \big\} \btheta \geq \underline{\mu} \big\|\btheta \big\|^2 - c_z d^{\alpha+1} \big\{ \log(pd)/n \big\}^{1/2}  \big\| \btheta \big\|_1^2, ~~~ \forall \btheta \in \eR^{2pd}, 
    $$
with probability greater than $1- \tilde{c}_5 (pd)^{-\tilde{c}_6}.$  
\end{lemmaS}
\begin{proof} 
Let $\widehat\bGamma = n^{-1} \widehat{\bD}^{-1} (\widehat{\boZ}^\T\widehat{\boZ})\widehat{\bD}^{-1}$ and $\bGamma = n^{-1} {\bD}^{-1} \eE[{\boZ}^\T{\boZ}]{\bD}^{-1}$.
It is evident that $\btheta^{\T} \widehat{\bGamma} \btheta = \btheta^{\T} \bGamma\btheta + \btheta^{\T} (\widehat{\bGamma} -\bGamma)\btheta$. Consequently, we have
\begin{equation}\label{Gamma.inf}
    \btheta^{\T} \widehat{\bGamma} \btheta \geq \btheta^{\T} \bGamma\btheta -\|\widehat{\bGamma} -\bGamma\|_{\max }\|\btheta\|_1^2.
\end{equation} 
It follows from Condition~\ref{cond_inf_SFLR} 
and (\ref{Gamma.inf}) that
$ \btheta^{\T} \widehat{\bGamma} \btheta \geq  \underline{\mu}\|\btheta\|^2 -  \|\widehat{\bGamma} -\bGamma\|_{\max} \|\btheta\|_1^2.$  
By Lemma 5 of \textcolor{blue}{Guo and Qiao}, (\textcolor{blue}{2023}), we obtain that, if $n \gtrsim  d^{4\alpha +2} \log(pd)$, then with probability greater than $1- \tilde{c}_5 (pd)^{-\tilde{c}_6},$
\begin{equation}
\label{Gamma.max}
\|\widehat{\bGamma} -\bGamma\|_{\max} \leq  c_z d^{\alpha+1} \big\{ \log(pd)/n \big\}^{1/2}.
\end{equation}
Combining the above results, we complete the proof of this lemma. 
\end{proof}


\begin{lemmaS}\label{lemmaeigenSFLR}
     Suppose  that  Condition~\ref{cond_eigen_SFLR} holds. If $n \gtrsim   d^{4\alpha +2} \log(pd)$, then there exist some positive constants $\tilde{c}_5,\tilde{c}_6$ such that 
    \begin{equation*}\label{X.Eigens}
     \max_{j \in [2p]\atop l \in [d]} \left| \frac{\hat\omega_{jl}^{-1/2}  - \omega_{jl}^{-1/2}}{\omega_{jl}^{-1/2}} \right| \lesssim \sqrt{\frac{\log(pd)}{n}},
    \end{equation*}
   with probability greater than $1- \tilde{c}_5 (pd)^{-\tilde{c}_6}$. 
\end{lemmaS}
\begin{proof}
By Proposition 3 in \textcolor{blue}{Guo and Qiao} (\textcolor{blue}{2023}) and $\mathcal{M}_1^Y=1$ under no serial dependence, we complete the proof of this lemma.  
\end{proof}

\begin{lemmaS}\label{lemmaC3SFLR}
     Suppose  that  Conditions~\ref{cond_error_SFLR}, \ref{cond_eigen_SFLR}--\ref{cond_coef_SFLR}  hold. If $n \gtrsim   d^{4\alpha +2} \log(pd)$, then there exist some positive constants $c_e, \tilde{c}_5,\tilde{c}_6$ such that 
    $$
    n^{-1}\|\widehat{\bD}^{-1} \widehat{\boZ}^\T \big(\bY -  \widehat{\boZ} \bbb) \|^{(d \times 1)}_{\max} \leq  c_e s 
    (d^{\alpha+2} \big\{\log(pd)/n \big\}^{1/2} +   d^{1- \tau}\big)
    $$ 
   with probability greater than $1- \tilde{c}_5 (pd)^{-\tilde{c}_6}$. 
\end{lemmaS}
\begin{proof}
Note that 
\begin{equation}\label{SFLR.yminus.dep}
    \begin{split}
        &~~~~ n^{-1}\widehat{\bD}^{-1} \widehat{\boZ}^\T (\bY -  \widehat{\boZ} \bbb)\\
        & =  n^{-1}\widehat{\bD}^{-1} \widehat{\boZ}^\T \bY - {\bD}^{-1}\eE[n^{-1} {\boZ}^\T \bY ] +  {\bD}^{-1}\eE[ n^{-1}{\boZ}^\T \bY ]-  n^{-1}\widehat{\bD}^{-1} \widehat{\boZ}^\T \widehat{\boZ} \bbb \\
        &=n^{-1}\widehat{\bD}^{-1} \widehat{\boZ}^\T \bY - {\bD}^{-1}\eE[n^{-1} {\boZ}^\T \bY ] + {\bD}^{-1}\eE[n^{-1} {\boZ}^\T {\boZ} \bbb]-  n^{-1}\widehat{\bD}^{-1} \widehat{\boZ}^\T \widehat{\boZ} \bbb + {\bD}^{-1}\eE[n^{-1} {\boZ}^\T \bepsilon],
    \end{split}
\end{equation}
where $\bepsilon = (\epsilon_1,\dots,\epsilon_n)^\T$ is the truncation error. \\
First, we show the deviation bounds of $n^{-1}\widehat{\bD}^{-1} \widehat{\boZ}^\T \bY - {\bD}^{-1}\eE[ n^{-1} {\boZ}^\T \bY ]$, which can be decomposed as $\widehat{\bD}^{-1} \big(n^{-1} \widehat{\boZ}^\T \bY - \eE[ n^{-1} {\boZ}^\T \bY ] \big) + (\widehat{\bD}^{-1} -{\bD}^{-1}) \eE[ n^{-1} {\boZ}^\T \bY ]$. By Lemmas~\ref{lemmaC1SFLR} and  \ref{lemmaeigenSFLR}, we have 
\begin{equation}\label{SFLR.yminus.p1}
    \|n^{-1}\widehat{\bD}^{-1} \widehat{\boZ}^\T \bY - {\bD}^{-1}\eE[ n^{-1} {\boZ}^\T \bY ]\|^{(d \times 1)}_{\max} \lesssim d^{\alpha+3/2} \big\{\log(pd)/n \big\}^{1/2}. 
\end{equation} 
Second, we write
${\bD}^{-1}\eE[n^{-1} {\boZ}^\T {\boZ} \bbb]-  n^{-1}\widehat{\bD}^{-1} \widehat{\boZ}^\T \widehat{\boZ} \bbb = (\bGamma -\widehat{\bGamma})\bD\bbb+ \widehat{\bGamma}(\bD-\widehat{\bD})\bbb.$ 
By $\|{\bD} \bbb \|_1^{(d\times1)} = O(s)$ and (\ref{Gamma.max}), 
we obtain that 
\begin{equation}\label{SFLR.yminus.p2}
    \|(\bGamma -\widehat{\bGamma})\bD\bbb\|^{(d \times 1)}_{\max}   \lesssim sd^{\alpha+2}  \big\{\log(pd)/n \big\}^{1/2}.
\end{equation} 
By Lemma~\ref{lemmaeigenSFLR} and $\|{\bD} \bbb \|_1^{(d\times1)} = O(s)$, we have 
\begin{equation}\label{SFLR.yminus.p3}
    \|\widehat{\bGamma}(\bD-\widehat{\bD})\bbb\|^{(d \times 1)}_{\max} =  \|\widehat{\bGamma}(\bD-\widehat{\bD})\bD^{-1}\bD\bbb\|^{(d \times 1)}_{\max}  \lesssim sd \big\{\log(pd)/n \big\}^{1/2}.
\end{equation} 
Note (\ref{SFLR.yminus.p1}), (\ref{SFLR.yminus.p2}) and (\ref{SFLR.yminus.p3}) all hold with probability greater than $1- \tilde{c}_5 (pd)^{-\tilde{c}_6}.$ 
By Condition~\ref{cond_coef_SFLR}(i) and Lemma 23 in \textcolor{blue}{Fang et al.}~(\textcolor{blue}{2022}), we have that $\big\|{\bD}^{-1}\eE[n^{-1} {\boZ}^\T \bepsilon]\big\|_{\max} \leq O(sd^{1/2-\tau})$, which implies that 
    $\big\|{\bD}^{-1}\eE[n^{-1} {\boZ}^\T \bepsilon]\big\|^{(d \times 1)}_{\max}
    \lesssim sd^{1-\tau}.$ 
Combining this with (\ref{SFLR.yminus.dep}), (\ref{SFLR.yminus.p1}), (\ref{SFLR.yminus.p2}) and (\ref{SFLR.yminus.p3}), we complete the proof of this lemma.   
\end{proof}

\begin{lemmaS}
\label{lemmaC4SFLR}
Suppose that  Conditions~\ref{cond_error_SFLR}--\ref{cond_coef_SFLR}  hold, 
$ d^{\alpha} s\lambda_n \to 0$ as $n,p,d \to \infty$, and the regularization parameter $\lambda_n \geq 2  c_e
s \| \widehat\bD\|_{\max}    \big( d^{\alpha+2} \{\log(pd)/n\}^{1/2} +   d^{1- \tau}\big)$. 
Then there exist some positive constants $\tilde{c}_5,\tilde{c}_6$ such that, with probability greater than $1-\tilde{c}_5(pd)^{-\tilde{c}_6},$ 
  $$
 \sum_{j=1}^{2p}\| \bbb_j -  \widehat{\bbb}_j \|\lesssim  d^{\alpha} s\lambda_n. 
  $$
\end{lemmaS}
\begin{proof}
Given that $\widehat{\bbb}_j$ is the solution to the minimization problem in (\ref{glasso.SFLR}), we have
\begin{equation*} \label{eqSFLR0}
    \begin{split}
      & - n^{-1} \bY^\T \widehat{\boZ} \widehat{\bbb}  + \frac12 \widehat{\bbb}^\T n^{-1} \widehat{\boZ}^\T \widehat{\boZ} \widehat{\bbb}   + \lambda_n \| \widehat{\bbb}\|_{1}^{(d\times 1)} \\
      \leq & - n^{-1} \bY^\T \widehat{\boZ} {\bbb}  + \frac12 {\bbb}^\T n^{-1} \widehat{\boZ}^\T \widehat{\boZ} {\bbb}   + \lambda_n \| {\bbb}\|_{1}^{(d\times 1)}.  
    \end{split}
\end{equation*} 
Let $\bDelta = \widehat{\bbb} - {\bbb}$ and $\Bar{S}^c$ represents the complement of $S$ within the set $[2p]$. 
Consequently, we have 
\begin{equation}\label{eqSFLR1}
    \begin{split}
     &\frac12 \bDelta^\T n^{-1} \widehat{\boZ}^\T \widehat{\boZ} \bDelta \\
     \leq&  - \bDelta^\T n^{-1} \widehat{\boZ}^\T \widehat{\boZ} \bbb +   \bDelta^\T n^{-1} \widehat{\boZ}^\T \bY  + \lambda_n (\| {\bbb}\|_{1}^{(d\times 1)} -  \| \widehat{\bbb}\|_{1}^{(d\times 1)} ) \\
      \leq& \bDelta^\T (n^{-1} \widehat{\boZ}^\T \bY - n^{-1} \widehat{\boZ}^\T \widehat{\boZ} \bbb ) + \lambda_n (\| {\bDelta}_S\|_{1}^{(d\times 1)} -  \| {\bDelta}_{\Bar{S}^c}\|_{1}^{(d\times 1)} ).
  \end{split}
\end{equation}
By Lemma~\ref{lemmaC3SFLR} and the choice of $\lambda_n$, we obtain that, with probability greater than $1- \tilde{c}_5 (pd)^{-\tilde{c}_6},$
\begin{equation}\label{eqSFLR2}
    \begin{split}
         \big|\bDelta^\T (n^{-1} \widehat{\boZ}^\T \bY - n^{-1} \widehat{\boZ}^\T \widehat{\boZ} \bbb ) | & = |\bDelta^\T \widehat\bD \{n^{-1} \widehat\bD^{-1} \widehat{\boZ}^\T ( \bY - \widehat{\boZ} \bbb ) \} \big| \\
         &\leq \|n^{-1} \widehat\bD^{-1} \widehat{\boZ}^\T ( \bY - \widehat{\boZ} \bbb )\|^{(d \times 1)}_{\max} 
           \| \widehat\bD\|_{\max}   \|\bDelta\|_1^{(d \times 1)} \\
         &\leq \frac{\lambda_n}{2} (\| {\bDelta}_S\|_{1}^{(d\times 1)} +  \| {\bDelta}_{\Bar{S}^c}\|_{1}^{(d\times 1)} ), 
    \end{split}
\end{equation}
Combining (\ref{eqSFLR1}) and (\ref{eqSFLR2}), we have 
$$ \frac{3\lambda_n}{2} \| {\bDelta}_S\|_{1}^{(d\times 1)} - \frac{\lambda_n}{2} \| {\bDelta}_{\Bar{S}^c}\|_{1}^{(d\times 1)}\geq \frac{1}{2}\bDelta^\T n^{-1} \widehat{\boZ}^\T \widehat{\boZ} \bDelta \geq 0, $$
which indicates that $3 \| {\bDelta}_S\|_{1}^{(d\times 1)} \geq  \| {\bDelta}_{\Bar{S}^c}\|_{1}^{(d\times 1)} $. 
By Condition~\ref{cond_inf_SFLR}, $d^{\alpha} s\lambda_n \to 0$ and Lemma \ref{lemmaC2SFLR}, we can let $\underline{\mu} \geq 32 d s [c_z d^{\alpha+1} \{\log(pd)/n\}^{1/2}]=o(1)$, which ensures that $\bDelta^\T n^{-1} \widehat{\boZ}^\T \widehat{\boZ} \geq {\underline{\mu}} \| \bDelta\widehat\bD\| ^2/2$. 
Combining this with Lemma~\ref{lemmaC2SFLR} 
and  Condition~\ref{cond_eigen_SFLR}, we have 
\begin{equation*}
    \begin{split}
        \bDelta^\T n^{-1} \widehat{\boZ}^\T \widehat{\boZ} \bDelta &\geq \underline{\mu}  \| \bDelta\widehat\bD\| ^2 - 16 c_zs d^{\alpha+2}\{ \log(pd)/n\}^{1/2} \| \bDelta\widehat\bD \|^2\\
        &\geq {\underline{\mu}} \| \bDelta\widehat\bD\| ^2/2 \geq {\underline{\mu}} c_0 \alpha^{-1} d^{-\alpha}  \| \bDelta \| ^2 /2.\\
    \end{split}
\end{equation*}
Note the facts that
$\| {\bDelta}\|_{1}^{(d\times 1)} = 
\| {\bDelta}_S \|_{1}^{(d\times 1)}  +  
\| {\bDelta}_{\Bar{S}^c}\|_{1}^{(d\times 1)} 
\leq 4 \| {\bDelta}_S\|_{1}^{(d\times 1)}  
\leq 4 {s}^{1/2} \| \bDelta\|$ and $ {3\lambda_n}\| {\bDelta}_S\|_{1}^{(d\times 1)} \geq \bDelta^\T n^{-1} \widehat{\boZ}^\T \widehat{\boZ} \bDelta$. Hence, $6 {s}^{1/2} {\lambda_n}  \| \bDelta\|   \geq 3 {\lambda_n} \| {\bDelta}_S\|_{1}^{(d\times 1)} /{2}  \geq  {\underline{\mu}} c_0 \alpha^{-1} d^{-\alpha} \| \bDelta\| ^2 /{4}$. Then we have $$  \| \bDelta\|  \leq {24  \alpha d^{\alpha} s^{1/2} \lambda_n}/({\underline{\mu}}c_0)~~\hbox{and}~ ~\| {\bDelta}\|_{1}^{(d\times 1)}  \leq {96  \alpha d^{\alpha}  s\lambda_n}/({\underline{\mu}}c_0),$$ 
 with probability greater than $1-\tilde{c}_5(pd)^{-\tilde{c}_6}.$ 
The proof of this lemma is completed. 
\end{proof}


\begin{lemmaS}
\label{lemmaC5SFLR}
 Suppose that Condition~\ref{cond_coef_SFLR}(iii) holds. Then there exists some constant
$c_1' \in \big(2(qs)^{-1},1\big)$
such that, with probability greater than $1-\tilde{c}_5(pd)^{-\tilde{c}_6},$ 
$|\widehat{S}_{\delta}| \geq c_1' s $ for $\widehat{S}_{\delta}$ defined in (\ref{est.null}).
\end{lemmaS}

\begin{proof}
By Lemma~\ref{lemmaC4SFLR}, we have that, with probability greater than  $1- \tilde{c}_5(pd
)^{-\tilde{c}_6},$ 
$$\max_{j \in [p]} \| \bbb_j - \widehat{\bbb}_j\| \leq {24   \alpha d^{\alpha}   s^{1/2} \lambda_n}/({\underline{\mu}}c_0)~ \hbox{and}~
\max_{j \in [p]} \| \widehat{\bbb}_{j+p}\| \leq {24   \alpha d^{\alpha}   s^{1/2} \lambda_n}/({\underline{\mu}}c_0).$$ 
 Hence, for any $j \in [p]$, we have that
\begin{equation}
\label{Wj.all}
    W_j = \| \widehat{\bbb}_{j}\| - \|\widehat{\bbb}_{j+p}\| \geq - \|\widehat{\bbb}_{j+p}\| \geq  -{24  \alpha d^{\alpha}   s^{1/2} \lambda_n}/({\underline{\mu}}c_0).
\end{equation} 
This implies $T_{\delta} \leq  {24   \alpha d^{\alpha}    s^{1/2} \lambda_n}/({\underline{\mu}}c_0).$ Otherwise, if $T_{\delta} >  {24   \alpha d^{\alpha}    s^{1/2} \lambda_n}/({\underline{\mu}}c_0)$, by (\ref{Wj.all}), we have $\{j\in[p]: W_j < -T_{\delta}\}$ is a null set. 
Under Condition~\ref{cond_coef_SFLR}(iii), if $j \in S_2 = \big\{j \in [p]: \|\bbb_j\| \gg {24   \alpha d^{\alpha}    s^{1/2} \lambda_n}/({\underline{\mu}}c_0)\big\}$, we have  
$W_j = \| \widehat{\bbb}_{j}\| - \|\widehat{\bbb}_{j+p}\| \geq  \|\bbb_j\|  - \| \widehat{\bbb}_{j} - \bbb_j\|- \|\widehat{\bbb}_{j+p}\|\gg {24   \alpha d^{\alpha}    s^{1/2} \lambda_n}/({\underline{\mu}}c_0).$ 
Therefore, $ S_2 \subseteq \widehat{S}_{\delta} = \big\{j\in[p]: W_j \geq T_{\delta}\big\}.$ Combing this with Condition~\ref{cond_coef_SFLR}(iii), we complete the proof of this lemma.
\end{proof}

The proof strategy for Theorem~\ref{thm_power_sflr} closely resembles that for~Theorem \ref{thm_power_fflr} based on the above technical lemmas. Hence we will only provide detailed proof of Theorem~\ref{thm_power_fflr} in  Section~\ref{pf.power_fflr} and omit the detailed proof of Theorem~\ref{thm_power_sflr} here.

\subsection{Proof of Theorem~\ref{thm_fdr_fflr}}
\label{pf.thm_fflr}
To prove Theorem~\ref{thm_fdr_fflr}, we firstly present some technical lemmas with their proofs. 
\begin{corollary}
\label{CorFF}
For any subset $G   \subseteq S^{c}$,  $\big({\widehat{\bxi}}_i^\T, {\widecheck{{\bxi}}}_i^\T \big)~ \Big|~  \widehat{\bbbeta}_i  \overset{D}{=} \big({\widehat{\bxi}}_i^\T, {\widecheck{{\bxi}}_i^\T}\big)_{\tswap(G)}~\Big|~ \widehat{\bbbeta}_i $.
\end{corollary}
\begin{proof}
In a similar way to the proof of Lemma~\ref{lemmaEx.func}, we obtain that $\big\{\bX_i(\cdot)^\T, \widetilde{\bX}_i(\cdot)^\T ,Y_i(\cdot) \big\}    \overset{D}{=} \big(\{\bX_i(\cdot)^\T, \widetilde{\bX}_i(\cdot)^\T\}_{\tswap(G)},    Y_i(\cdot)\big)$.
This implies that, for any $\mathbf{t}  \in (\boldsymbol{\mathcal{H}}^2,\mathcal{H}_Y),$ 
\begin{equation*}
    \eE \Big[\exp\Big\{\iota \big\langle \mathbf{t} , (\bX_i^\T, \widetilde{\bX}_i^\T,   Y_i)^\T \big\rangle \Big\}\Big] = \eE \Big[\exp\Big\{\iota \big\langle \mathbf{t} , \big\{(\bX_i^\T, \widetilde{\bX}_i^\T)_{\tswap(G)}, Y_i\big\}^\T \big\rangle\Big\}\Big]. 
\end{equation*} 
Given $\mathbf{t} = (\mathbf{c}_1^\T \widehat\bphi_{1},\dots, \mathbf{c}_p^\T \widehat\bphi_{p}, \widetilde{\mathbf{c}}_1^\T \widehat\bphi_{1}, \dots, \widetilde{\mathbf{c}}_p^\T \widehat\bphi_{p}, \mathbf{c}_Y^\T\widehat\bpsi)^\T,$
where $\mathbf{c}_{j} = (c_{j1}, \dots, c_{jd_j}, 0 ,0 ,\dots )^\T$, 
$\widetilde{\mathbf{c}}_{j} = (\tilde c_{j1},  \dots, \tilde c_{jd_j}, 0 ,0 ,\dots )^\T$, 
 ${\mathbf{c}}_{Y} = ( c_{Y,1},  \dots, c_{Y,\tilde{d}}, 0 ,0 ,\dots )^\T$, 
$\widehat{\bphi}_j = (\hat\phi_{j1}, \hat\phi_{j2}, \dots )^\T$ for $j \in [p]$ and $\widehat{\bpsi} = (\hat\psi_{1}, \hat\psi_{2}, \dots )^\T$, it follows that 
\begin{equation}\label{co.FFLR}
  \eE \Big[\exp\big\{\iota \big\langle \mathbf{c} , \big({\widehat{\bxi}_i}^\T, {\widecheck{{\bxi}}^\T_i}  ,\widehat{\bbbeta}_i^\T \big)^\T \big\rangle \big\} \Big| \mathbf{t}  \Big]   
    = \eE \Big[\exp\big\{\iota \big\langle \mathbf{c}  , \big\{({\widehat{\bxi}_i^\T}, {\widecheck{{\bxi}}_i^\T})_{\tswap(G)}, \widehat{\bbbeta}_i^\T \big\}^\T \big\rangle\big\}  \Big| \mathbf{t}  \Big],
\end{equation}
for any $ \mathbf{c} = (\Bar{\mathbf{c}}_{1}^\T, \dots, \Bar{\mathbf{c}}_{p}^\T,\widetilde{\Bar{\mathbf{c}}}_{1}^\T, \dots, \widetilde{\Bar{\mathbf{c}}}_{p}^\T, \Bar{\mathbf{c}}_{Y}^\T)^\T,$ 
where 
$\Bar{\mathbf{c}}_{j} = (c_{j1}, \dots, c_{jd_j})^\T$,
$\widetilde{\Bar{\mathbf{c}}}_{j} = (\tilde c_{j1},  \dots, \tilde c_{jd_j})^\T$, $\Bar{\mathbf{c}}_{Y} = (c_{Y,1}, \dots, c_{Y,\tilde{d}})^\T$. 
By (\ref{co.FFLR}) and the total expectation formula, the joint characteristic function of $({\widehat{\bxi}_i}^\T, {\widecheck{{\bxi}}_i}^\T,\widehat\bbbeta_i^\T)$ is equal to that of $\big\{({\widehat{\bxi}_i}^\T, {\widecheck{{\bxi}}_i} ^\T)_{\tswap(G)}, \widehat\bbbeta_i^\T \big\},$ which implies that
$({\widehat{\bxi}_i}^\T, {\widecheck{{\bxi}}_i}^\T,\widehat\bbbeta_i^\T) \overset{D}{=} \big\{({\widehat{\bxi}_i}^\T, {\widecheck{{\bxi}}_i} ^\T)_{\tswap(G)}, \widehat\bbbeta_i^\T \big\}$, and thus completes our proof of this lemma.
\end{proof}



By Corollary~\ref{CorFF}, we next prove Lemma~\ref{lemmacoin} in the context of FFLR.     

\begin{proof}[Proof of Lemma~\ref{lemmacoin} in FFLR] 
   Denote $\widehat{\bUps}= (\widehat\bbbeta_1, \dots, \widehat\bbbeta_n)^\T \in \eR^{n\times \tilde{d}}$. 
    First, since $W_j  = \| \widehat\bB_j\|_\tF- \| \widehat\bB_{p+j}\|_\tF$ for each $j \in [p]$,  
      the flip-sign property of $W_{j} = W_j( \widehat {\boldsymbol \Xi}, \widecheck{ {\boldsymbol \Xi}},\widehat{\bUps})$ holds.      
{Second, denote $\boldsymbol W = (W_1,\dots,W_p)^\T$ and take any subset $G \subseteq S^{c}$ of null. 
By swapping variables in $G$, we define 
$$
\boldsymbol W_{\tswap{(G)}}  \overset{\triangle}{=} \Big(W_1\big(\{\widehat {\boldsymbol \Xi}, \widecheck{ {\boldsymbol \Xi}} \}_{\tswap(G)},\widehat{\bUps}\big), \dots, W_p\big(\{\widehat {\boldsymbol \Xi}, \widecheck{ {\boldsymbol \Xi}} \}_{\tswap(G)},\widehat{\bUps}\big)\Big)^\T. 
$$
According to Corollary~\ref{CorFF}, we have $(\widehat\bXi,\widecheck{\bXi},\widehat{\bUps})   \overset{D}{=} \big((\widehat\bXi,\widecheck{\bXi})_{\tswap{(G)}},\widehat{\bUps}\big)$, implying $\boldsymbol W \overset{D}{=} \boldsymbol W_{\tswap{(G)}}$.
} 

Lastly, let $S^c_- = \{j \in S^{c}: \delta_j = -1 \}$, where $\boldsymbol \delta = (\delta_1,  \dots, \delta_p)^\T$ is a sequence of independent random variables. Each $\delta_j$ follows a Rademacher distribution if $j \in S^{c}$, and $\delta_j = 1$ otherwise. 
In the first step, we establish $\boldsymbol W _{\tswap(S^c_-)} = (\delta_1 W_1,  \dots,\delta_p W_p)^\T.$ Following the second step, we obtain $\boldsymbol W _{\tswap(S^c_-)} \overset{D}{=} \boldsymbol W.$ Combing the above results, we have $(\delta_1 W_1, \dots,\delta_p W_p)^\T \overset{D}{=} \boldsymbol W,$ which completes the proof of this lemma. 
\end{proof}


\begin{lemmaS}\label{lemmaeqFFLR}
   Suppose that  Condition~\ref{irrSFLR} holds. Then  $\|\beta_{j}\|_\cS = 0$  if and only if $j\in S^c$, where $S^c$ is defined in (\ref{f.nullset}). 
\end{lemmaS} 
\begin{proof}
    On the one hand,  assume that $\|\beta_j\|_\cS = 0$. For any $ \big(t_Y(\cdot),t_j(\cdot)\big) \in (\mathcal{H}_Y,\mathcal{H}_j)$, the joint characteristic function of $\big(Y(\cdot), {X}_j(\cdot)\big)$ conditional on $ X_{-j}(\cdot)$ can be factorized as  
    \begin{equation*}\label{eqinden}
    \begin{array}{cl}
         & \eE\big[ \exp\big\{\iota \big\langle (t_Y,t_j)^\T, (Y, {X}_j)^\T  \big\rangle \big\} ~\big|~ X_{-j}\big]  \\
       =  &  \eE \big[\exp\big\{\iota \big\langle (t_Y,t_j)^\T, \big\{\sum_{k\neq j} \int_{\cU}\beta_k(\cdot,u) X_k(u) du + \varepsilon, {X}_j\big\}^\T \big\rangle \big\} ~\big|~ X_{-j}\big]\\
       =   &   \eE  \big[ \exp\big\{\iota \langle t_Y, \sum_{k\neq j} \int_{\cU}\beta_k(\cdot,u) X_k(u) du + \varepsilon \rangle \big\} ~\big|~ X_{-j}\big]  \eE  \big[ \exp\big\{\iota \langle t_j ,  {X}_j \rangle \big\} ~\big|~ X_{-j}\big]  \\
     =  &   \eE  \big[ \exp\big\{\iota \langle t_Y, Y \rangle \big\} ~\big|~ X_{-j}\big]  \eE  \big[ \exp\big\{\iota \langle t_j ,  {X}_j \rangle \big\} ~\big|~ X_{-j}\big],  
    \end{array}
\end{equation*}
where the second equality comes from the fact that $\sum_{k\neq j}\int_{\cU}\beta_k(\cdot,u) X_k(u) du + \varepsilon$ and ${X}_j$ are independent conditional on $X_{-j}$.  
Hence it implies that $j \in S^c$. 

On the other hand, assume that $Y$ and $X_j$ are conditionally independent, i.e., $j\in S^c$. Then the joint characteristic function conditional on $X_{-j}(\cdot)$ can be factorized as 
$$\eE\big[ \exp\big\{\iota \big\langle (t_Y,t_j)^\T, (Y, {X}_j)^\T \big\rangle \big\} ~\big|~ X_{-j}\big] =     \eE  \big[ \exp\big\{\iota \langle t_Y, Y \rangle \big\} ~\big|~ X_{-j}\big]  \eE  \big[ \exp\big\{\iota \langle t_j ,  {X}_j \rangle \big\} ~\big|~ X_{-j}\big].$$ 
In FFLR, it is worth noting that the left-hand side involves an interaction term, i.e., $\eE\Big[ \exp\big\{\iota \big\langle t_Y, \int_\cU \beta_{j} (\cdot,u)X_{j}(u)du \big\rangle \big\} ~\Big|~ X_{-j}\Big], $
which needs to be a constant.
Condition~\ref{irrSFLR} implies that the interaction term is a constant only when $\|\beta_j\|_\cS=0.$ Combing the above results, we complete the proof of this lemma. 
\end{proof}

\begin{proof}[Proof of Theorem~\ref{thm_fdr_fflr}] 
  Provided that Lemma~\ref{lemmacoin} applies to FFLR, this confirms that the signs of null statistics are distributed as i.i.d. coin flips. 
  By the Theorem~3.4 in \cite{candes2018},  we have 
\begin{equation*} \label{eqfflr}
{\mathbb E}\left[\frac{|\widehat{S}_1  \cap S^c |}{|\widehat{S}_1|}\right] \leq  q~~\text{and}~~
{\mathbb E}\left[\frac{|\widehat{S}_0  \cap S^c|}{| \widehat{S}_0|+1/q}\right] \leq q, 
\end{equation*} 
where $S^c$ is defined as (\ref{f.nullset}). 
By Lemma~\ref{lemmaeqFFLR}, we establish the equivalence between $S^c$ in (\ref{f.nullset}) and the set $\{j: \|\beta_j\|_\cS=0\}$, whose complement is defined in (\ref{support.FFLR}). This equivalence demonstrates that FDR in FFLR is effectively controlled, which completes our proof. 
\end{proof}


\subsection{Proof of Theorem~\ref{thm_power_fflr}} 
\label{pf.power_fflr}
Before presenting technical lemmas used in the proof of Theorem~\ref{thm_power_fflr}, we begin with introducing some notation. For $j \in [p]$, denote ${\sigma}^{X,Y}_{j lm} = \eE [ \xi_{ijl} \eta_{im} ]$ and its estimator $\hat{\sigma}^{X,Y}_{j lm} = n^{-1} \sum_{i=1}^n \hat\xi_{ijl} \hat\eta_{im}$ for $l \in [d]$, $m \in [\tilde{d}].$ For $j \in [2p]\backslash [p],$ denote ${\sigma}^{X,Y}_{j lm} = \eE [ \tilde\xi_{i(j-p)l} \eta_{im}]$ and its estimator  $\hat{\sigma}^{X,Y}_{jlm} = n^{-1} \sum_{i=1}^n  \check\xi_{i(j-p)l} \hat{\eta}_{im}.$ 
Let $\bUps = (\bbbeta_{1},\dots,\bbbeta_{n})^{\T}\in \eR^{n\times \tilde{d}}$ with the estimator $\widehat\bUps = (\widehat\bbbeta_{1},\dots,\widehat\bbbeta_{n})^{\T}$. 


\begin{lemmaS} \label{lemmaC1}
     Suppose that  Conditions~\ref{cond_eigen_SFLR}, \ref{cond_error_FFLR}, \ref{cond_eigen_FFLR} hold. If $n \gtrsim  ( d^{4\alpha +2} \vee  \tilde{d}^{4\tilde{\alpha} +2} )\log(pd\tilde{d})$,  
      then there exist some positive constants $\tilde{c}_7, \tilde{c}_8$ such that, with probability greater than $1-\tilde{c}_7(pd \tilde{d})^{-\tilde{c}_8}$,  the estimates $\{\hat{\sigma}^{X,Y}_{jlm}\}$ satisfy 
    $$\max_{j\in [2p]\atop l\in [d], m\in [\tilde{d}]} \frac{|\hat{\sigma}^{X,Y}_{j lm}  - {\sigma}^{X,Y}_{j lm}|}{(l^{\alpha+1}\vee m ^{\tilde\alpha+1}) \omega_{jl}^{1/2}\tilde\omega_{m}^{1/2}}   \lesssim  \sqrt{\frac{\log(pd \tilde{d})}{n}} .$$
\end{lemmaS}

\begin{proof}
The proof of this lemma is similar to that of Lemma~\ref{lemmaC1SFLR}, we thus omit the proof. 
\end{proof}

     


\begin{lemmaS}\label{lemmaC3}
     Suppose that  Conditions~\ref{cond_eigen_SFLR}, \ref{cond_error_FFLR}--\ref{cond_coef_FFLR}  hold. If $n \gtrsim   ( d^{4\alpha +2} \vee  \tilde{d}^{4\tilde{\alpha} +2} )\log(pd\tilde{d})$, then there exist some positive constants $c_e,\tilde{c}_9,\tilde{c}_{10}$ such that 
    $$
    n^{-1}\|\widehat{\bD}^{-1} \widehat{\boZ}^\T (\widehat{\bUps} -  \widehat{\boZ} \bB) \|^{(d \times \tilde{d})}_{\max} \leq c_es d^{1/2}
    \big(\{ d^{\alpha+3/2} \vee \tilde{d}^{\tilde{\alpha}+3/2} \}\{\log(pd\tilde{d})/n\}^{1/2} +   d^{1/2- \tau}\big)
    $$
   with probability greater than $1- \tilde{c}_9 (pd\tilde{d})^{-\tilde{c}_{10}}.$
\end{lemmaS}
\begin{proof}
Note that 
\begin{equation}\label{FFLR.yminus.dep}
    \begin{split}
        &~~~~ n^{-1}\widehat{\bD}^{-1} \widehat{\boZ}^\T (\widehat{\bUps} -  \widehat{\boZ} \bB)\\
        & =  n^{-1}\widehat{\bD}^{-1} \widehat{\boZ}^\T \widehat{\bUps} - {\bD}^{-1}\eE[n^{-1} {\boZ}^\T {\bUps} ] +  {\bD}^{-1}\eE[ n^{-1}{\boZ}^\T {\bUps} ]-  n^{-1}\widehat{\bD}^{-1} \widehat{\boZ}^\T \widehat{\boZ} \bB \\
        &=n^{-1}\widehat{\bD}^{-1} \widehat{\boZ}^\T \widehat{\bUps} - {\bD}^{-1}\eE[n^{-1} {\boZ}^\T {\bUps}] + {\bD}^{-1}\eE[n^{-1} {\boZ}^\T {\boZ} \bB]-  n^{-1}\widehat{\bD}^{-1} \widehat{\boZ}^\T \widehat{\boZ} \bB + {\bD}^{-1}\eE[n^{-1} {\boZ}^\T \bepsilon],
    \end{split}
\end{equation} 
where $\bepsilon = (\bepsilon_1,\dots,\bepsilon_n)^{\T}\in \eR^{n\times\tilde{d}}$ is the truncation error. \\
First, we show the deviation bounds of $n^{-1}\widehat{\bD}^{-1} \widehat{\boZ}^\T \widehat\bUps - {\bD}^{-1}\eE[ n^{-1} {\boZ}^\T \bUps ]$, which can be decomposed as $\widehat{\bD}^{-1} \big(n^{-1} \widehat{\boZ}^\T\widehat\bUps - \eE[ n^{-1} {\boZ}^\T \bUps ] \big) + (\widehat{\bD}^{-1} -{\bD}^{-1}) \eE[ n^{-1} {\boZ}^\T \bUps ]$. Then, by Condition~\ref{cond_eigen_FFLR}, Lemmas~\ref{lemmaeigenSFLR} and~\ref{lemmaC1}, we have  
\begin{equation}\label{FFLR.yminus.p1}
    \|n^{-1}\widehat{\bD}^{-1} \widehat{\boZ}^\T \widehat\bUps - {\bD}^{-1}\eE[ n^{-1} {\boZ}^\T \bUps ]\|^{(d \times \tilde{d})}_{\max} \lesssim d^{1/2}(d^{\alpha+1}\vee \tilde{d}^{\tilde{\alpha}+1}) \big\{\log(pd\tilde{d})/n \big\}^{1/2}.
\end{equation}  
Second, we write 
${\bD}^{-1}\eE[n^{-1} {\boZ}^\T {\boZ} \bB]-  n^{-1}\widehat{\bD}^{-1} \widehat{\boZ}^\T \widehat{\boZ} \bB = (\bGamma -\widehat{\bGamma})\bD\bB +  \widehat{\bGamma}(\bD-\widehat{\bD})\bB.$
By $\|{\bD} \bB \|_1^{(d\times \tilde{d})} = O(s)$ and (\ref{Gamma.max}), we have 
\begin{equation}\label{FFLR.yminus.p2}
    \|(\bGamma -\widehat{\bGamma})\bD\bB\|^{(d \times  \tilde{d})}_{\max}   \lesssim sd^{\alpha+2}  \big\{\log(pd)/n \big\}^{1/2}. 
\end{equation} 
By Lemma~\ref{lemmaeigenSFLR} and $\|{\bD} \bB \|_1^{(d\times \tilde{d})} = O(s)$, we have 
\begin{equation}\label{FFLR.yminus.p3}
    \|\widehat{\bGamma}(\bD-\widehat{\bD})\bB\|^{(d \times \tilde{d})}_{\max}   \lesssim sd \big\{\log(pd)/n \big\}^{1/2}.
\end{equation}  
Note (\ref{FFLR.yminus.p1}), (\ref{FFLR.yminus.p2}) and (\ref{FFLR.yminus.p3}) all hold with probability greater than $1- \tilde{c}_9 (pd)^{-\tilde{c}_{10}}.$ By Condition~\ref{cond_coef_FFLR}(i) and Proposition 4 in \textcolor{blue}{Guo and Qiao} (\textcolor{blue}{2023}), we have $\big\|{\bD}^{-1}\eE[n^{-1} {\boZ}^\T \bepsilon]\big\|_{\max} \leq O(sd^{1/2-\tau})$, which implies that  
   $ \|{\bD}^{-1}\eE[n^{-1} {\boZ}^\T \bepsilon]\|_{\max}^{(d\times\tilde{d})} 
    \lesssim sd^{1-\tau}.$ 
Combing this with (\ref{FFLR.yminus.dep}), (\ref{FFLR.yminus.p1}), (\ref{FFLR.yminus.p2}) and  (\ref{FFLR.yminus.p3}),  
we complete the proof of this lemma. 
\end{proof}

\begin{lemmaS}
\label{lemmaC4}
Suppose that  Conditions~\ref{cond_inf_SFLR}--\ref{cond_eigen_SFLR}, \ref{cond_error_FFLR}--\ref{cond_coef_FFLR}  hold, 
$s d^\alpha \lambda_n \to 0$ as $n,p,d \to \infty$, and the regularization parameter $\lambda_n \geq 2 c_es  \| \widehat\bD\|_{\max}  d^{1/2}
    \big[( d^{\alpha+3/2} \vee \tilde{d}^{\tilde\alpha+3/2})\{\log(pd\tilde{d})/n\}^{1/2} +   d^{1/2- \tau}\big]$.
Then there exist some positive constants $\tilde{c}_9,\tilde{c}_{10}$ such that,   with probability greater than $1- \tilde{c}_9 (pd\tilde{d})^{-\tilde{c}_{10}}$,
  $$
 \sum_{j=1}^{2p}\| \bB_j -  \widehat{\bB}_j \|_\tF \lesssim   s d^\alpha \lambda_n. 
  $$
  
\end{lemmaS}
\begin{proof}
Since $\widehat{\bB}_j$ is the minimizer of (\ref{glasso.FFLR}), we have
\begin{equation*}
    \begin{split}
        &- \tr (n^{-1} \widehat{\bUps}^\T \widehat{\boZ} \widehat{\bB} ) + \frac12 \tr (\widehat{\bB}^\T n^{-1} \widehat{\boZ}^\T \widehat{\boZ} \widehat{\bB}  ) + \lambda_n \| \widehat{\bB}\|_{1}^{(d\times \tilde{d})}\\
\leq& - \tr (n^{-1} \widehat{\bUps}^\T \widehat{\boZ} {\bB} ) + \frac12 \tr ({\bB}^\T n^{-1} \widehat{\boZ}^\T \widehat{\boZ} {\bB}  ) + \lambda_n \| {\bB}\|_{1}^{(d\times \tilde{d})}.
    \end{split}
\end{equation*}
Let $\bDelta=\widehat{\bB} - {\bB}$ and $\Bar{S}^c$ represents the complement of $S$ in the set $[2p]$. Consequently, we obtain that
\begin{equation}\label{eqFFLR1}
    \begin{split}
    &\frac12 \tr(\bDelta^\T n^{-1} \widehat{\boZ}^\T \widehat{\boZ} \bDelta) \\ 
    \leq&  - \tr(\bDelta^\T n^{-1} \widehat{\boZ}^\T \widehat{\boZ} \bB) +   \tr(\bDelta^\T n^{-1} \widehat{\boZ}^\T \widehat{\bUps} ) + \lambda_n (\| {\bB}\|_{1}^{(d\times \tilde{d})} -  \| \widehat{\bB}\|_{1}^{(d\times \tilde{d})} ) \\
     \leq&  \tr\big\{\bDelta^\T (n^{-1} \widehat{\boZ}^\T \widehat{\bUps} - n^{-1} \widehat{\boZ}^\T \widehat{\boZ} \bB )\big\} + \lambda_n \big(\| {\bDelta}_S\|_{1}^{(d\times \tilde{d})} -  \| {\bDelta}_{\Bar{S}^c}\|_{1}^{(d\times \tilde{d})}\big).
   \end{split}
\end{equation}
By Lemma~\ref{lemmaC3} and the choice of $\lambda_n$, we obtain that, with probability greater than $1- \tilde{c}_9 (pd)^{-\tilde{c}_{10}},$  
\begin{equation}\label{eqFFLR2}
    \begin{split}
        |\tr\{\bDelta^\T (n^{-1} \widehat{\boZ}^\T \widehat{\bUps} - n^{-1} \widehat{\boZ}^\T \widehat{\boZ} \bB )\} | &\leq \|n^{-1} \widehat\bD^{-1} \widehat{\boZ}^\T (\widehat{\bUps} -  \widehat{\boZ} \bB)\|^{(d \times \tilde{d})}_{\max} \|\widehat\bD\|_{\max} \|\bDelta\|_1^{(d \times \tilde{d})}\\
        &  \leq \frac{\lambda_n}{2} (\| {\bDelta}_S\|_{1}^{(d\times \tilde{d})} +  \| {\bDelta}_{\Bar{S}^c}\|_{1}^{(d\times \tilde{d})} ). 
    \end{split}
\end{equation}
Combining (\ref{eqFFLR1}) and (\ref{eqFFLR2}), we have 
$$ \frac{3\lambda_n}{2} \| {\bDelta}_S\|_{1}^{(d\times \tilde{d})} - \frac{\lambda_n}{2} \| {\bDelta}_{\Bar{S}^c}\|_{1}^{(d\times \tilde{d})}\geq \frac{1}{2}\tr(\bDelta^\T n^{-1} \widehat{\boZ}^\T \widehat{\boZ} \bDelta) \geq 0, $$
which indicates that $3 \| {\bDelta}_S\|_{1}^{(d\times \tilde{d})} \geq  \| {\bDelta}_{\Bar{S}^c}\|_{1}^{(d\times \tilde{d})} $. 
By Condition~\ref{cond_inf_SFLR}, $s d^\alpha \lambda_n \to 0$ and Lemma \ref{lemmaC2SFLR}, we can let $\underline{\mu} \geq 32 d \tilde{d} s [c_z d^{\alpha+1} \{\log(pd)/n\}^{1/2}]=o(1)$, which ensures that $\tr(\bDelta^\T n^{-1} \widehat{\boZ}^\T \widehat{\boZ} \bDelta) \geq {\underline{\mu}} \| \bDelta \widehat\bD\|_\tF^2 /{2} $. 
Combining this with  Lemma~\ref{lemmaC2SFLR}  
and Condition \ref{cond_eigen_SFLR}, we have 
 \begin{equation*}
     \begin{split}
         \tr(\bDelta^\T n^{-1} \widehat{\boZ}^\T \widehat{\boZ} \bDelta) &\geq \underline{\mu}  \| \bDelta \widehat\bD\|_\tF^2 - 16 c_z \tilde{d} d^{\alpha+2} \big\{ \log(pd)/n \big\}^{1/2}  \| \bDelta \widehat\bD\|_{\tF}^2 \\
         &\geq {\underline{\mu}} \| \bDelta \widehat\bD\|_\tF^2 /{2} \geq {\underline{\mu}} c_0 \alpha^{-1} d^{-\alpha}  \| \bDelta \|_{\tF}^2 /2.\\
     \end{split}
 \end{equation*}
Note the facts that
$\| {\bDelta}\|_{1}^{(d\times \tilde{d})} = 
\| {\bDelta}_S \|_{1}^{(d\times \tilde{d})}  +  
\| {\bDelta}_{\Bar{S}^c}\|_{1}^{(d\times \tilde{d})} 
\leq 4 \| {\bDelta}_S\|_{1}^{(d\times \tilde{d})}  
\leq 4 s^{1/2} \| \bDelta\|_\tF$ and ${3\lambda_n} \| {\bDelta}_S\|_{1}^{(d\times \tilde{d})} \geq \tr(\bDelta^\T n^{-1} \widehat{\boZ}^\T \widehat{\boZ} \bDelta)$. Hence, $6{s}^{1/2} \| \bDelta\|_\tF  \geq {3\lambda_n \| {\bDelta}_S\|_{1}^{(d\times \tilde{d})} }/{2} \geq {\underline{\mu}} c_0 \alpha^{-1} d^{-\alpha} \| \bDelta\|_\tF^2 /{4} $, which implies that $$  \| \bDelta\|_\tF \leq {24 \alpha  {s}^{1/2} d^{\alpha}\lambda_n}/({\underline{\mu}}c_0)~\hbox{and}~ \| {\bDelta}\|_{1}^{(d\times \tilde{d})} =  \sum_{j=1}^{2p}\| \bB_j -  \widehat{\bB}_j \|_\tF \leq {96  \alpha s d^{\alpha} \lambda_n}/({\underline{\mu}}c_0),$$
with probability greater than $1- \tilde{c}_9 (pd)^{-\tilde{c}_{10}}.$ The proof is completed. 
\end{proof}


\begin{lemmaS}
\label{lemmaC5}
 Suppose that   Condition~\ref{cond_coef_FFLR}(iii) holds. Then there exists some constant 
    $c_2' \in \big(2(qs)^{-1},1\big)$
    such that  
    $|\widehat{S}_{\delta}| \geq c_2' s $ for $\widehat{S}_{\delta}$ defined in  (\ref{est.null}), with probability greater than 
$1- \tilde{c}_9(pd\tilde{d})^{- \tilde{c}_{10}}$.
\end{lemmaS}

\begin{proof}
    By Lemma~\ref{lemmaC4}, we have that, with probability greater than 
$1- \tilde{c}_9(pd\tilde{d})^{- \tilde{c}_{10}},$  
$$\max_{j \in [p]} \| \bB_j - \widehat{\bB}_j\|_{\tF} \leq {24 \alpha  {s}^{1/2} d^{\alpha}\lambda_n}/({\underline{\mu}}c_0)~ \hbox{and}~
\max_{j \in [p]} \| \widehat{\bB}_{j+p}\|_{\tF} \leq {24 \alpha {s}^{1/2}  d^{\alpha}\lambda_n}/({\underline{\mu}}c_0).$$
Hence, for each $j \in [p]$, we have that
\begin{equation}
\label{wj.all.ff}
    W_j = \| \widehat{\bB}_{j}\|_{\tF} - \|\widehat{\bB}_{j+p}\|_{\tF} \geq - \|\widehat{\bB}_{j+p}\|_{\tF} \geq  -{24 \alpha  {s}^{1/2} d^{\alpha}\lambda_n}/({\underline{\mu}}c_0).
\end{equation}
which implies $T_{\delta} \leq  {24 \alpha  {s}^{1/2} d^{\alpha}\lambda_n}/({\underline{\mu}}c_0).$  Otherwise $\{j\in[p]: W_j < -T_{\delta}\}$ constitutes a null set. 
By Condition~\ref{cond_coef_FFLR}(iii), we have that if  $j \in S_2 = \{ j \in [p]:\|\bB_j\|_\tF \gg {24 \alpha  {s}^{1/2} d^{\alpha}\lambda_n}/({\underline{\mu}}c_0)\}$, then 
$W_j = \| \widehat{\bB}_{j}\|_{\tF} - \|\widehat{\bB}_{j+p}\|_{\tF} \geq  \|\bB_j\|_\tF  - \| \widehat{\bB}_{j} - \bB_j\|_{\tF}- \|\widehat{\bB}_{j+p}\|_{\tF}\gg {24 \alpha  {s}^{1/2} d^{\alpha}\lambda_n}/({\underline{\mu}}c_0).$ 
This implies that $ S_2 \subseteq \widehat{S}_{\delta} = \{j\in[p]: W_j \geq T_{\delta} \}$. Combing this with Condition~\ref{cond_coef_FFLR}(iii), we complete the proof of this lemma. 
\end{proof}

We are now ready to prove Theorem~\ref{thm_power_fflr}.

\begin{proof}[Proof of Theorem~\ref{thm_power_fflr}]
 Consider the ordered statistics $|W_{(1)}| \geq |W_{(2)}| \geq \dots \geq |W_{(p)}|$. Let $j^*$ denote the index such that the threshold $T_{\delta} = |W_{(j^*)}|$. By definition of $T_{\delta}$, $-T_{\delta} < W_{(j^*+1)} \leq 0$. We will establish Theorem~\ref{thm_power_fflr} by investigating two scenarios: $-T_{\delta} < W_{(j^*+1)} < 0$ and $W_{(j^*+1)} = 0$.    \\
{\bf Scenario 1.}  For $-T_{\delta} < W_{(j^*+1)} < 0$, given the definition of $T_{\delta}$, we have 
\begin{equation} \label{eqthreshT}
    \frac{|\{ j \in [p]: W_j \leq -T_{\delta} \}| + 2}{|\{ j \in [p]: W_j \geq T_{\delta} \}|} > q.
\end{equation}
This result holds because otherwise $|W_{(j^*+1)}|$ would serve as the new lower threshold for knockoffs. 
 According to (\ref{eqthreshT}) and Lemma~\ref{lemmaC5}, we have
 $$
 \big|\{ j \in [p]:W_j \leq -T_{\delta}  \}\big |  > q \big|\{j \in [p]: W_j \geq T_{\delta}\}\big|  - 2 \geq {q c_2's-2} ,
 $$
 for $c_2' \in \big(2(qs)^{-1},1 \big).$ 
By Lemma~\ref{lemmaC4} with $\sum_{j=1}^{2p}\| \bB_j - \widehat{\bB}_j \|_\tF \leq 96 \alpha  s d^\alpha \lambda_n / (\underline{\mu}c_0)$, and considering $\|\bB_j\|_\tF = 0$ for $j=p+1, \dots, 2p$, we obtain 
\begin{equation}\label{powereq1}
    \begin{split}
         96 \alpha s d^\alpha \lambda_n / (\underline{\mu}c_0) = \sum_{j=1}^{2p}\| \bB_j - \widehat{\bB}_j \|_\tF &\geq \sum_{j \in \{ j \in [p]: W_j \leq -T_{\delta} \}}  \|\widehat{\bB}_{j+p}\|_\tF \\
         &\geq T_{\delta}  \cdot \big|\{j\in [p]: W_j \leq -T_{\delta}\}\big|\geq T_{\delta} \cdot (q c_2' s - 2),
    \end{split}
\end{equation}
where, the second inequality holds because for $j \in \{j \in [p]: W_j \leq -T_{\delta} \}$, we have $\|\widehat{\bB}_j\|_\tF - \|\widehat{\bB}_{j+p}\|_\tF \leq -T_{\delta}$, implying $\|\widehat{\bB}_{j+p}\|_\tF \geq T_{\delta}$. 
By (\ref{powereq1}), we have $T_{\delta} \leq  96 \alpha s d^\alpha \lambda_n \big\{\underline{\mu}c_0 (q c_2' s - 2) \big\}^{-1}$. 
 Similarly, by Lemma~\ref{lemmaC4}, $\|\widehat{\bB}_{j+p}\|_\tF \geq  \|\widehat{\bB}_{j}\|_\tF - T_{\delta}$ for $j \in \widehat{S}^c_{\delta}$, the triangle inequality, and $\min_{j \in S}   \|\bB_j\|_{\tF} \geq  \kappa_{n} d^\alpha \lambda_n/ \underline{\mu},$ we have
\begin{equation}\label{powereq2}
    \begin{split}
         96 \alpha s d^\alpha \lambda_n/ (\underline{\mu}   c_0)   & =\sum_{j=1}^{2p}\| \bB_j -  \widehat{\bB}_j \|_\tF  = \sum_{j=1}^{p} \big( \| \bB_j -  \widehat{\bB}_j \|_\tF +   \|\widehat{\bB}_{j+p}\|_\tF\big)\\
                &  \geq   \sum_{j \in \widehat{S}^c_{\delta} \cap S}\big(\| \bB_j -  \widehat{\bB}_j \|_\tF +   \|\widehat{\bB}_{j+p}\|_\tF\big)  \\
                &\geq   \sum_{j \in \widehat{S}^c_{\delta} \cap S}\big(\| \bB_j -  \widehat{\bB}_j \|_\tF +   \|\widehat{\bB}_{j}\|_\tF - T_{\delta}\big)\\
                & \geq  \sum_{j \in \widehat{S}^c_{\delta} \cap S} \big( \|\bB_{j}\|_\tF  - T_{\delta} \big)  \geq \big( \kappa_{n} d^\alpha \lambda_n/ \underline{\mu} - T_{\delta}\big)  \cdot | \widehat{S}^c_{\delta} \cap S|.
    \end{split}
\end{equation}
For large enough $\kappa_{n}$ such that $T_{\delta} \leq  96 \alpha s d^\alpha \lambda_n \big\{\underline{\mu}c_0 (q c_2' s - 2) \big\}^{-1}\leq\kappa_{n} d^\alpha \lambda_n/ (2\underline{\mu})$ and (\ref{powereq2}),   it holds, with probability greater than $1-\tilde{c}_9(pd \tilde{d})^{-\tilde{c}_{10}},$ that 
$$
 \frac{|\widehat{S}_{\delta}\cap S|}{| S| \vee 1} = 1 -  \frac{|\widehat{S}^c_{\delta}\cap S|}{| S| \vee 1} \geq 1 - \frac{192\alpha  }{c_0  }\kappa_{n}^{-1}. 
$$ 
{\bf Scenario 2.} For $W_{(j^*+1)} = 0, $  we have
$ \widehat{S}_{\delta} = \{ j \in [p]:  W_j  > 0\}$ and $\{j \in [p] : W_j   \leq -T_{\delta}\}  =  \{j \in [p] : W_j   <  0\}.$\\
If   $| \{j \in [p] : W_j   <  0\}|  > c_ns$ with $c_n = 192  \alpha / (c_0 \kappa_{n}  )$, it follows from
$$
  96 \alpha s d^\alpha \lambda_n / (\underline{\mu}c_0)  = \sum_{j=1}^{2p}\| \bB_j -  \widehat{\bB}_j \|_\tF \geq \sum_{j \in \{j \in [p]: W_j \leq -T_{\delta}   \}}  \|\widehat{\bB}_{j+p}\|_\tF \geq T_{\delta}  \cdot|\{j \in [p] :W_j \leq -T_{\delta}\}|,
$$
that
$$
T_{\delta} \leq \frac{ \sum_{j=1}^{2p}\| \bB_j -  \widehat{\bB}_j \|_\tF }{ |\{j \in [p]:W_j \leq -T_{\delta}\}|}
 = \frac{ \sum_{j=1}^{2p}\| \bB_j -  \widehat{\bB}_j \|_\tF }{ |\{j \in [p]:W_j< 0\}|} <  \frac{\kappa_{n} d^\alpha \lambda_n}{2\underline{\mu}}.
$$
Consequently, the argument simplifies to Scenario 1, and the subsequent analysis follows. \\
If   $| \{j \in [p]: W_j   <  0\}|  \leq  c_n  s,$ by $\widehat{S}_{\delta}  = \text{supp}(W_j) \backslash \{ j \in [p]: W_j < 0 \}$ with $ \text{supp}(W_j) = \{j \in [p]: W_j \neq 0 \} $, we have 
\begin{equation}\label{powereq3}
    \big|  \widehat{S}_{\delta} \cap S \big| =  \big|  \text{supp}(W_j) \cap S\big| -  \big|  \{ j \in [p]: W_j < 0 \} \cap S \big|  \geq  \big|   \text{supp}(W_j) \cap S \big| -   c_ns.
\end{equation}
Define $\widehat{S}_{\text{\tiny{GL}}}^c = \big\{j \in [p]:\| \widehat{\bB}_j \|_{\tF} = 0 \big\}$. 
As stated in \cite{fan2020rank}, we have to assume that there are no ties in the magnitude of the nonzero components of the group lasso solution. Then we can conclude that $\big\{j \in[p]: W_j = 0\big \} \subseteq \widehat{S}_{\text{\tiny{GL}}}^c$, which shows that
$[p] \setminus \widehat{S}_{\text{\tiny{GL}}}^c   \subseteq   \text{supp}(W_j).$
By Lemma~\ref{lemmaC4}, we have 
\begin{equation}\label{powereq4}
    \begin{split}
         96 \alpha  s d^{\alpha}\lambda_n / (\underline{\mu}c_0) & =  \sum_{j=1}^{2p}\| \bB_j -  \widehat{\bB}_j \|_\tF \geq \sum_{j\in \widehat{S}_{\text{\tiny{GL}}}^c  \cap S}  \| \bB_j -  \widehat{\bB}_j \|_\tF\\
              & =  \sum_{j\in \widehat{S}_{\text{\tiny{GL}}}^c  \cap S} \| \bB_j\|_\tF \geq | \widehat{S}_{\text{\tiny{GL}}}^c  \cap S| \min_{j\in S}  \| \bB_j\|_\tF.
    \end{split}
\end{equation} 
 Note that $\min_{j\in S} \| \bB_j\|_\tF \geq \kappa_{n} d^\alpha \lambda_n/ \underline{\mu}.$  Then we can get $ | \widehat{S}_{\text{\tiny{GL}}}^c  \cap S| \leq   96 \alpha  s/  (c_0 \kappa_{n} )$ from (\ref{powereq4}). 
 Therefore, we have $$\big|( [p] \setminus \widehat{S}_{\text{\tiny{GL}}}^c ) \cap S\big| \geq s\big\{1- 96\alpha / ( c_0 \kappa_{n} )\big\}.$$
 Combining (\ref{powereq3}) and $[p] \setminus \widehat{S}_{\text{\tiny{GL}}}^c   \subseteq   \text{supp}(W_j)$, we have  $| \widehat{S}_{\delta} \cap S | \geq \big|\big([p] \setminus \widehat{S}_{\text{\tiny{GL}}}^c\big) \cap S\big| -   c_ns = s\big\{1-96 \alpha / ( c_0 \kappa_{n} )- 192\alpha / (c_0 \kappa_{n} )\big\}, 
$
   which shows that $$\frac{|\widehat{S}_{\delta}\cap S|}{| S| \vee 1} \geq  1 -  \frac{192\alpha  }{ c_0} \kappa_{n}^{-1}$$ holds with probability greater than $1-\tilde{c}_9(pd \tilde{d})^{-\tilde{c}_{10}}$. \\
 Combing the above results under two scenarios, we have 
the power 
     $$\text{Power}(\widehat{S} ) = E\left[ \frac{|\widehat{S}_{\delta}\cap S|}{|S| \vee 1}\right] \rightarrow 1,$$
    which completes the proof of Theorem~\ref{thm_power_fflr}. 
\end{proof}

\subsection{Proof of Lemma~\ref{lemmaeqGGM}}
We will show that the set $E$ defined in (\ref{Edge}) is equivalent to the set $\big\{(j,k) \in [p]^2: k \in S_j \big\}$ defined in Lemma~\ref{lemmaeqGGM}. 
\begin{proof}
Let $\beta_{jk}$ and $C_{jk}$ represent operators induced from the coefficient function in (\ref{eq.GGM}) and the conditional covariance function, 
respectively. 
Note that we use ${\beta}_{jk}$ and $C_{jk}$ to denote both the operators and the kernel functions for notational economy. 
To demonstrate Lemma~\ref{lemmaeqGGM}, we will prove 
that $\|C_{jk} \|_{\cS}= 0\Longleftrightarrow \|\beta_{jk}\|_{\cS} = 0$.
 Note the fact that $\|C_{jk}\|_{\cS}= 0$ if and only if $\big\langle f, C_{jk}(g)  \big\rangle = 0$ for any $f\in \mathcal{H}_j$ and $g\in \mathcal{H}_k$.
 By Condition~\ref{irrGGM} and Lemma S5 in 
    \cite{solea2022copula}, we have $\big\langle f, C_{jk}(g)  \big\rangle = 0  \Longleftrightarrow  \cov \big( \langle f, X_{j} \rangle, \langle g, X_{k} \rangle ~|~ X_{-\{j,k\}}\big) =0 \Longleftrightarrow  \cov \big( \langle f, \sum _{l \neq j} \beta_{jl} (X_{l} ) + \varepsilon_{j}  \rangle, \langle g, X_{k} \rangle ~|~ X_{-\{j,k\}}\big) =0$. 
Since $\varepsilon_{j}$ and $X_{k}$ are independent and for any $l \in [p]\backslash  \{j,k\}$ and  $\cov \big( \big\langle f,  \beta_{jl} (X_{l})   \big\rangle, \langle g, X_{k} \rangle| X_{-\{j,k\}}\big) =0,$ we obtain that   
 \begin{equation*} \label{eqeqv} 
 \begin{split}
      \big\langle f, C_{jk}(g)  \big\rangle = 0 
      &\Longleftrightarrow 
 \cov \big( \big\langle f,  \beta_{jk} (X_{k})   \big\rangle, \langle g, X_{k} \rangle~|~ X_{-\{j,k\}}\big) =0\\
 &\Longleftrightarrow \|\beta_{jk} \|_{\cS}= 0,~\text{provided that}~ \beta_{jk} \text{ is a linear operator,}
 \end{split}
 \end{equation*} 
which implies that $\|C_{jk} \|_{\cS}= 0$ if and only if $ \|\beta_{jk}\|_{\cS} = 0$ and thus completes our proof. 
\end{proof}
 
\subsection{Proof of Theorem~\ref{thm_fdr_fggm}}
Note that we have already established the validity of Lemma~\ref{lemmacoin} within the FFLR framework, which can be directly extended to each row of $\bW$ within the FGGM framework. We are now ready to prove Theorem~\ref{thm_fdr_fggm}. 
\begin{proof}
Provided the validity of Lemma~\ref{lemmacoin} in FFLR, it can be similarly demonstrated that the knockoff statistics $\bW$ satisfy the sign-flip property on the neighborhood set $NE_j$ for each $j \in [p]$ at the rowwise level, where $NE_j = \big\{ k \in [p] \backslash\{j\} : \|C_{jk}\|_{\cS} \neq 0  \big\}.$ By Theorem 3.1 in \cite{li2021ggm}, we have 
$$
 {\mathbb E}\left[\frac{|\widehat{E}_{\tAnd,1}\cap E^c|}
{|\widehat{E}_{\tAnd,1}| \vee 1}\right] \leq  q, ~~\text{and}~~
  {\mathbb E}\left[\frac{|\widehat{E}_{\tOr,1} \cap E^c|}
{|\widehat{E}_{\tOr,1}| \vee 1}\right] \leq  q, 
$$
where $E$ is defined in (\ref{Edge}), which implies that 
${\FDR}_{\tAnd} \leq  q,$ and ${\FDR}_{\tOr} \leq  q$ in GGM.  
Similarly, drawing from Theorem 3.2 in \cite{li2021ggm}, we establish that the modified FDR can be controlled, i.e., ${\mFDR}_{\tAnd} \leq  q,$ and ${\mFDR}_{\tOr} \leq  q.$ 
\end{proof}



\subsection{Proof of Theorem~\ref{thm_power_fggm}}
To prove Theorem~\ref{thm_power_fggm}, we firstly present several technical lemmas with their proofs. In the following lemmas, let ${\boZ}_{-j} =({{\bXi}}_{-j}, {{\widetilde{\bXi}}_{-j}}) \in \eR^{n \times 2(p-1)d}$, $ {{\bXi}}_{-j}=(\bxi_{1(-j)},\dots,\bxi_{n(-j)})^\T \in \eR^{n\times(p-1)d}$, $\widetilde{{\bXi}}_{-j}=(\widetilde\bxi_{1(-j)},\dots,\widetilde\bxi_{n(-j)})^\T \in \eR^{n\times(p-1)d}$, $\bxi_{i(-j)} = (\bxi_{i1}^\T,\dots,\bxi_{i(j-1)}^\T,\bxi_{i(j+1)}^\T,\dots,\bxi_{ip}^\T)^\T \in \eR^{(p-1)d}$, $\widetilde\bxi_{i(-j)} = (\widetilde\bxi_{i1}^\T,\dots,\widetilde\bxi_{i(j-1)}^\T,\widetilde\bxi_{i(j+1)}^\T,\dots,\widetilde\bxi_{ip}^\T)^\T \in \eR^{(p-1)d}$, $ {{\bXi}}_{j} = (\bxi_{1j}, \dots, \bxi_{nj})^\T$, ${\bD}_{-j} = \tdiag({\bD}_{1},\dots,{\bD}_{j-1},{\bD}_{j+1},\dots, {\bD}_{p}, {\bD}_{1},\dots,{\bD}_{j-1},{\bD}_{j+1},\dots, {\bD}_{p}) \in \eR^{2(p-1)d \times 2(p-1)d}$, and ${\bB}_{j(-j)} =({\bB}_{j1}^\T, \dots, {\bB}_{j(j-1)}^\T, {\bB}_{j(j+1)}^\T, \dots,{\bB}_{jp}^\T,{\bB}_{j(p+1)}^\T, \dots, {\bB}_{j(p+j-1)}^\T, {\bB}_{j(p+j+1)}^\T, \dots,{\bB}_{j(2p)}^\T)^{\T} \in \eR^{2(p-1)d \times d}$ with estimates $\widehat{\boZ}_{-j} =({\widehat{\bXi}}_{-j}, {\widecheck{{\bXi}}}_{-j})$, $ \widehat{{\bXi}}_{-j}=(\widehat\bxi_{1(-j)},\dots,\widehat\bxi_{n(-j)})^\T $, $ \widecheck{{\bXi}}_{-j}=(\widecheck\bxi_{1(-j)},\dots,\widecheck\bxi_{n(-j)})^\T $, $\widehat\bxi_{i(-j)} = (\widehat\bxi_{i1}^\T,\dots,\widehat\bxi_{i(j-1)}^\T,\widehat\bxi_{i(j+1)}^\T,\dots,\widehat\bxi_{ip}^\T)^\T$, $\widecheck\bxi_{i(-j)} = (\widecheck\bxi_{i1}^\T,\dots,\widecheck\bxi_{i(j-1)}^\T,\widecheck\bxi_{i(j+1)}^\T,\dots,\widecheck\bxi_{ip}^\T)^\T$, $\widehat{{\bXi}}_{j} = (\widehat\bxi_{1j}, \dots, \widehat\bxi_{nj})^\T$, $\widehat{\bD}_{-j} = \tdiag(\widehat{\bD}_{1},\dots,\widehat{\bD}_{j-1},\widehat{\bD}_{j+1},\dots, \widehat{\bD}_{p}, \widehat{\bD}_{1},\dots,\widehat{\bD}_{j-1},\widehat{\bD}_{j+1},\dots, \widehat{\bD}_{p})$, and $\widehat{\bB}_{j(-j)} =(\widehat{\bB}_{j1}^\T, \dots, \widehat{\bB}_{j(j-1)}^\T, \widehat{\bB}_{j(j+1)}^\T, \dots,\widehat{\bB}_{jp}^\T,\widehat{\bB}_{j(p+1)}^\T, \dots, \widehat{\bB}_{j(p+j-1)}^\T, \widehat{\bB}_{j(p+j+1)}^\T, \dots,\widehat{\bB}_{j(2p)}^\T)^{\T}$ for $i\in[n]$ and $j \in [p]$.
\begin{lemmaS}\label{lemmaC7}
    Suppose that  Conditions~\ref{cond_inf_SFLR}--\ref{cond_eigen_SFLR} hold.  If $n \gtrsim  d^{4\alpha +2} \log(pd)$, for each $j\in[p]$, there exist 
    some positive constants $c_{z},\tilde{c}_{11},\tilde{c}_{12}$ such that, with probability greater than $1- \tilde{c}_{11} (pd)^{-\tilde{c}_{12}}$,
    $$
    \btheta^{\T} \big\{n^{-1} \widehat{\bD}_{-j}^{-1} (\widehat{\boZ}_{-j}^\T\widehat{\boZ}_{-j})  \widehat{\bD}_{-j}^{-1}\big\} \btheta \geq \underline{\mu} \big\|\btheta \big\|^2 - c_{z} d^{\alpha+1} \big\{\log(pd)/n\big\}^{1/2}  \big\|\btheta\big\|_1^2,~~~ \forall \btheta \in \eR^{2(p-1)d}. 
    $$
\end{lemmaS}

\begin{proof}  
The proof of this lemma for FGGM is similar to that of Lemma~\ref{lemmaC2SFLR} for FFLR and hence is omitted here. 
It is noteworthy that, in an analogy to the infimum $\underline{\mu}$ defined in Condition \ref{cond_inf_SFLR} for FFLR, we can define the corresponding infimum for FGGM, which is used in our proof of this lemma and is no less than $\underline{\mu}$ for FFLR. Hence,
the result with the presence of $\underline{\mu}$ in this lemma remains valid. 
\end{proof}


\begin{lemmaS}\label{lemmaC8}
   Suppose that  Conditions~\ref{cond_eigen_SFLR} and~\ref{cond_power_GGM} hold. If $n \gtrsim    d^{4\alpha +2} \log(pd)$, for each $j\in[p]$, 
   there exist some positive constants 
   $\tilde{c}_{e},\tilde{c}_{11},\tilde{c}_{12}$ such that 
    $$
   \|n^{-1} \widehat{\bD}_{-j}^{-1} \widehat{\boZ}_{-j}^\T (\widehat{\bXi}_j -  \widehat{\boZ}_{-j} \bB_{j(-j)})\|^{(d \times {d})}_{\max} \leq \tilde{c}_{e}s
    (d^{\alpha+2} \{\log(pd)/n\}^{1/2} +  d^{1- \tau}),
    $$
   with probability greater than $1- \tilde{c}_{11} (pd)^{-\tilde{c}_{12}}.$ 
\end{lemmaS}

\begin{proof}
    The proof of this lemma is similar to that of Lemma~\ref{lemmaC3}, thus being omitted here.
\end{proof}

\begin{lemmaS}
\label{lemmaC9}
Suppose that   Conditions~\ref{cond_inf_SFLR}--\ref{cond_eigen_SFLR} and~\ref{cond_power_GGM} hold,  $s_jd^\alpha \lambda_{nj} \to 0$ as $n,p,d \to \infty$,  
and any regularization parameter $\lambda_{nj} \geq 2 c_{e}s_j\| \widehat\bD_{-j}\|_{\max}(d^{\alpha+2} \{\log(pd)/n\} +  d^{1- \tau})$ for each  $j \in [p]$.
Then for each $j\in [p]$, there exist some positive constants 
$\tilde{c}_{11},\tilde{c}_{12}$ such that
  $$
 \sum_{k \in [2p] \backslash \{j,p+j\}}\| \bB_{j k}-  \widehat{\bB}_{j k} \|_\tF \lesssim   s_jd^\alpha \lambda_{nj} ,
  $$
  with probability greater than $1-\tilde{c}_{11}(pd)^{-\tilde{c}_{12}}$. 
\end{lemmaS}
\begin{proof}
    The proof of this lemma is similar to that of Lemma~\ref{lemmaC4}, thus being omitted here. Specifically, we obtain that 
     $ \sum_{k \in [2p] \backslash \{j,p+j\}}\| \bB_{j k}-  \widehat{\bB}_{j k} \|_\tF \leq {96  \alpha  s_j d^{\alpha} \lambda_n}/({\underline{\mu}}c_{0}).$ 
\end{proof}

\begin{lemmaS}
\label{lemmaC10}
   Suppose  that  Condition~\ref{cond_power_GGM}(iii) holds,  then there exists 
   $c' \in \big((1+a) c_a p (qs)^{-1},1\big)$ such that $|\widehat{E}_{\tOr,\delta} | \geq c'|E|$,  with probability greater than $1-\tilde{c}_{11}(pd)^{-\tilde{c}_{13}}$.  
\end{lemmaS}

\begin{proof}
First, applying Condition~\ref{cond_power_GGM}(iii) along with Lemmas~\ref{lemmaC5} and \ref{lemmaC9}, we obtain that $|\widehat{S}_{\delta,j}| \geq c_j s_j$ holds  with probability greater than $1-\tilde{c}_{11}(pd)^{-\tilde{c}_{12}}$  for each $j \in [p]$, { where $c_j \in \big((1+a) c_a p (qs)^{-1},1\big)$. 
By Bonferroni inequality, $|\widehat{S}_{\delta,j}| \geq c_j s_j$ holds simultaneously across $j\in[p]$ with probability greater than 
$1-\tilde{c}_{11}(pd)^{-\tilde{c}_{13}}$, which can be achieved for sufficiently large $n$. }
Then under the OR rule, we have  
$$  | \widehat{E}_{\tOr,\delta} | \geq \sum_{j=1}^{p} | \widehat{S}_{\delta,j} | 
\geq \sum_{j=1}^{p}  c_j s_j \geq c' \sum_{j=1}^{p}  s_j 
= c' |{E} |, $$
 where $c' = \inf_ {j\in [p]} c_j \in \big((1+a) c_a p (qs)^{-1},1\big).$ 
The proof is completed.  
\end{proof}

We are now ready to prove Theorem~\ref{thm_power_fggm}. 
\begin{proof}[Proof of Theorem~\ref{thm_power_fggm}]
    First, it is essential to note that for each $j$, the selected threshold ${T}_{\delta,j}$ represents the minimum positive number that satisfies the constraints in the optimization problem (\ref{opt.OR}).
Let $|W_{j(1)}| \geq |W_{j(2)} | \geq \dots \geq |W_{j(p)}|$ represent the ordered statistics. The index at which the threshold $T_{\delta, j} = |W_{j(k_j^*)}|$ is reached is denoted by $k_j^*$. Similar to Theorem \ref{thm_power_fflr}, we will prove this theorem in two scenarios: $-T_{\delta, j} <  W_{j(k^*_j+1)} < 0$ and $ W_{j(k^*_j+1)} = 0, $ respectively.\\ 
{\bf Scenario 1.} 
For $-T_{\delta, j} < W_{j(k^*_j+1)} < 0$, by the definition of $T_{\delta, j}$, we obtain that
\begin{equation}\label{thm7.T.neq}
  \frac{\big|\big\{ k\in [p]\backslash\{j\} :W_{jk} \leq -T_{\delta,j} \big\}\big| +1+a}{|\widehat{E}_{\tOr,\delta}|} > \frac{q}{c_ap}. 
\end{equation}
The result in (\ref{thm7.T.neq}) holds because otherwise $ |W_{j(k^*_j+1)}|$ would represent the new lower threshold for knockoffs. 
By (\ref{thm7.T.neq}), we have
 $
 \big|\big\{ k\in [p]\backslash\{j\} :W_{jk} \leq -T_{\delta,j} \big\}\big| > q | \widehat{E}_{\tOr,\delta}| / (c_a p)  - 1-a.
 $
 Combining this with the result in Lemma~\ref{lemmaC10}, we have $q | \widehat{E}_{\tOr,\delta}| / (c_a p)  - 1-a \geq {q c' |E|/ (c_a p)  -1-a}$. 
 Then, we obtain that
 \begin{equation}\label{thm7.E.neq}
   \big|\big\{ k\in [p]\backslash\{j\} :W_{jk} \leq -T_{\delta,j} \big\}\big| >  {q c' |E|/ (c_a p)  -1-a}, 
 \end{equation} 
 where  $c' \in \big((1+a) c_a p (qs)^{-1},1\big)$.
It follows from Lemma~\ref{lemmaC9} that $\sum_{k \in [2p] \backslash \{j,p+j\}}\| \bB_{j k}-  \widehat{\bB}_{j k} \|_\tF\leq  96  s_j d^\alpha \lambda_{nj} /( c_{0}\underline{\mu})$. For each $j\in [p]$ and $k\in\{p+1, \dots, 2p\}\backslash \{ j+p \}$, we have $\|\bB_{jk}\|_\tF = 0$. Then we obtain that 
\begin{equation}\label{thm7.B.neq}
\begin{split}
  96  s_j d^\alpha\lambda_{nj} /(\underline{\mu}c_0 )&= \sum_{k \in [2p] \backslash \{j,p+j\}}\| \bB_{j k}-  \widehat{\bB}_{j k} \|_\tF \\
  &\geq \sum_{k \in \{k \in [p] \backslash \{j\}: W_{jk} \leq -T_{\delta,j}   \}}  \|\widehat{\bB}_{j(k+p)}\|_\tF \\
 &\geq T_{\delta,j}   \cdot | \{k \in [p] \backslash \{j\}: W_{jk} \leq -T_{\delta,j}  \}, 
\end{split}
\end{equation} 
where the last inequality holds since when $W_{jk} \leq -T_{\delta,j}$, it implies that $ \|\widehat{\bB}_{jk}\|_\tF  -  \|\widehat{\bB}_{j(k+p)}\|_\tF  \leq -T_{\delta,j}$ and then follows that $\|\widehat{\bB}_{j(k+p)}\|_\tF \geq T_{\delta, j}.$ 
Combing the results in (\ref{thm7.E.neq}) and (\ref{thm7.B.neq}), we obtain that
\begin{equation}\label{thm7.neq.T}
 T_{\delta,j} \leq \frac{ 96  s_j d^\alpha \lambda_{nj} }{(qc's/c_a p -1-a) \underline{\mu}c_0}.
\end{equation}
 Similarly, by Lemma~\ref{lemmaC9}, the triangle inequality, and Condition~\ref{cond_power_GGM}(ii), we have
\begin{equation}\label{thm7.B.2.neq}
   \begin{split}
    96 \alpha  s_j d^{\alpha} \lambda_{nj} / (\underline{\mu}c_0 )   & = \sum_{k \in [2p] \backslash \{j,p+j\}}\| \bB_{j k}-  \widehat{\bB}_{j k} \|_\tF \\ 
    &= \sum_{k \in [p] \backslash \{j\}} \big\{ \| \bB_{jk} -  \widehat{\bB}_{jk} \|_\tF +   \|\widehat{\bB}_{j(k+p)}\|_\tF\big\}\\
                &  \geq   \sum_{k \in \widehat{S}_{\delta,j}^c \cap S_j}\big\{ \| \bB_{jk} -  \widehat{\bB}_{jk} \|_\tF +   \|\widehat{\bB}_{j(k+p)}\|_\tF\big\}\\
                 &  \geq   \sum_{k \in \widehat{S}_{\delta,j}^c \cap S_j}\big\{\| \bB_{jk} -  \widehat{\bB}_{jk} \|_\tF +   \|\widehat{\bB}_{jk}\|_\tF - T_{\delta,j}\big\}\\
                 &\geq  \sum_{k \in \widehat{S}_{\delta,j}^c \cap S_j}  \|\bB_{jk}\|_\tF  - T_{\delta,j}   \geq \{\kappa_{n} d^{\alpha}\lambda_{nj}/\underline{\mu}  - T_{\delta,j}\}  \cdot | \widehat{S}_{\delta,j}^c \cap S_j|,  
 \end{split}
\end{equation}
where the second inequality holds since $\|\widehat{\bB}_{j(k+p)}\|_\tF \geq  \|\widehat{\bB}_{jk}\|_\tF - T_{\delta, j}$ for $j \in \widehat{S}_{\delta,j}^c$. 
For sufficiently large $\kappa_{n}$, by (\ref{thm7.neq.T}), we obtain that $T_{\delta,j} \leq { 96  s_j d^\alpha \lambda_{nj} }\big\{(qc's/c_a p -1-a) \underline{\mu}c_0\big\}^{-1} \leq \kappa_{n} d^{\alpha}\lambda_{nj}/(2\underline{\mu})$. Combining this with (\ref{thm7.B.2.neq}) yields that $|\widehat{S}_{\delta,j}^c \cap S_j| \leq 192 \alpha s_j/ c_0$, which implies that 
\begin{equation}\label{thm7.neq.S} 
\begin{split}
 \frac{|\widehat{S}_{\delta,j}\cap S_j|}{| E| \vee 1} &=  \frac{| S_j|}{| E| \vee 1} -  \frac{|\widehat{S}_{\delta,j}^c\cap S_j|}{| E| \vee 1} \geq \frac{s_j}{ s \vee 1} -  \frac{192 \alpha s_j }{c_0 s} \kappa_{n}^{-1}
\end{split}
\end{equation} 
 holds with probability greater than $1-\tilde{c}_{11}(pd)^{-\tilde{c}_{12}}$.\\ 
{\bf Scenario 2.} For $W_{j(k^*_j+1)} = 0, $  we have
$ \widehat{S}_{\delta,j} = \big\{ k \in [p]\backslash\{j\} :  W_{jk}  > 0\big\}$ and $\big\{k \in [p]\backslash\{j\} : W_{jk}   \leq -T_{\delta,j}\big\}  =  \big\{k\in [p]\backslash\{j\} : W_{jk}   <  0\big\}.$\\
If  
$\big| \big\{k\in [p]\backslash\{j\} : W_{jk}   <  0 \big\}\big|  > c_{jn} s_j$ for each $j\in [p]$,  where $c_{jn} = 192  \alpha / (c_0 \kappa_{n})$, 
it follows from  
\begin{equation*}\label{thm.7.s2.B}
\begin{split}
    96 \alpha  s_j d^{\alpha}\lambda_{nj} / (\underline{\mu}c_0 )   & \geq  \sum_{k \in [2p] \backslash \{j,p+j\}} \| \bB_{jk} -  \widehat{\bB}_{jk} \|_\tF \\
    & \geq \sum_{\{k\in[p]\backslash\{j\}: W_{jk} \leq -T_{\delta, j}   \}}  \|\widehat{\bB}_{jk+p}\|_\tF \\
     &   \geq T_{\delta,j}  \cdot\big| \big \{k\in[p]\backslash\{j\}:W_{jk} \leq -T_{\delta,j}\big\}\big|,
\end{split}
\end{equation*}
that 
$
T_{\delta,j} 
 \leq  {(\kappa_{n}  d^\alpha \lambda_{nj} )}/{(2\underline{\mu}) }.
$
Consequently, the argument simplifies to Scenario 1, and the subsequent analysis follows.\\
If   $| \{k\in [p]\backslash\{j\} : W_{jk}   <  0\}|  \leq  c_{jn} s_j,$  by $\widehat{S}_{\delta,j}  =  \text{supp}(W_{jk}) \backslash \{k \in [p]\backslash\{j\}:W_{jk} < 0 \}$, where $  \text{supp}(W_{jk}) = \{k \in [p]\backslash\{j\}:W_{jk} \neq 0 \} $, we have
\begin{equation}\label{thm7.supp} 
\begin{split}
    |\widehat{S}_{\delta,j} \cap S_j| &=  |  \text{supp}(W_{jk}) \cap S_j| -  \big|  \big\{ k \in [p]\backslash\{j\}:W_{jk} < 0 \big\} \cap S_j\big| \\
     &  \geq  |   \text{supp}(W_{jk}) \cap S_j| -   c_{jn} s_j\\
     & \geq  \big|\big\{[p] \setminus \widehat{S}_{\text{\tiny{GL}},j}^c\big\} \cap S_j\big|-   c_{jn} s_j,  
\end{split}
\end{equation}
where $\widehat{S}_{\text{\tiny{GL}},j}^c = \big\{k\in [p]\backslash\{j\}:\| \widehat{\bB}_{jk} \|_{\tF} = 0 \big\}$. The last inequality in (\ref{thm7.supp}) holds since $ \big\{k \in [p]\backslash\{j\}: W_{jk} = 0 \big\} \subseteq \widehat{S}_{\text{\tiny{GL}},j}^c,$ i.e., 
$[p] \setminus \widehat{S}_{\text{\tiny{GL}},j}^c   \subseteq    \text{supp}(W_{jk}).$
By Lemma~\ref{lemmaC9}, we have
\begin{equation}\label{thm7.s2.S}
  \begin{split}
   96 \alpha  s_j d^{\alpha}\lambda_{nj} / (\underline{\mu}c_0 )  & =  \sum_{k \in [2p] \backslash \{j,p+j\}} \| \bB_{jk} -  \widehat{\bB}_{jk} \|_\tF \geq \sum_{j\in \widehat{S}_{\text{\tiny{GL}},j}^c  \cap S_j}  \| \bB_{jk} -  \widehat{\bB}_{jk} \|_\tF\\
              & =  \sum_{j\in \widehat{S}_{\text{\tiny{GL}},j}^c  \cap S_j} \| \bB_{jk}\|_\tF \geq | \widehat{S}_{\text{\tiny{GL}},j}^c  \cap S_j| \min_{k\in S_j}  \| \bB_{jk}\|_\tF.
\end{split}
\end{equation}
By (\ref{thm7.s2.S}) and Condition~\ref{cond_power_GGM}(ii), we can get $ | \widehat{S}_{\text{\tiny{GL}},j}^c  \cap S_j| \leq   96 \alpha  s_j  / (c_0 \kappa_n)$, 
which indicates that $\big|\big\{[p] \setminus \widehat{S}_{\text{\tiny{GL}},j}^c\big\} \cap S_j\big| \geq s_j \big\{1- 96 \alpha  / (c_0 \kappa_n)\big\}$.
 Combining this with the result in (\ref{thm7.supp}), we have 
 $
   |  \widehat{S}_j \cap S_j | 
     = s_j \big\{1-96 \alpha  / (c_0\kappa_{n} )- 192 \alpha  / (c_0\kappa_{n} )\big\},
$
 which means that, with probability greater than $1-\tilde{c}_{11}(pd)^{-\tilde{c}_{12}},$
 \begin{equation}\label{thm7.S.E}
   \frac{|\widehat{S}_{\delta,j}\cap S_j|}{|E| \vee 1} \geq  \frac{s_j}{s}\big\{1 -  192 \alpha  / (c_0\kappa_{n} ) \big\}.
 \end{equation}
 Finally, combing results in (\ref{thm7.neq.S}) and (\ref{thm7.S.E}) with Lemma \ref{lemmaC10}, we have that $$\sum_{j=1}^p\frac{|\widehat{S}_{\delta,j}\cap S_j|}{|E| \vee 1} \geq  \sum_{j=1}^p\frac{s_j}{s}\big\{1 -  192 \alpha  / (c_0\kappa_{n} ) \big\}$$
holds with probability greater than {$1-\tilde{c}_{11}(pd)^{-\tilde{c}_{13}}$.}  
    Then it follows that 
    \begin{equation}\label{thm7.S.E.1}
    \begin{split}
         E\left[ \sum_{j=1}^{p} \frac{|\widehat{S}_{\delta,j}\cap S_j|}{| E| \vee 1}\right] &\geq  \Big[\sum_{j=1}^{p}  \frac{s_j}{s}\big\{1 -  192 \alpha  / (c_0\kappa_{n} ) \big\}\Big] \big\{1-\tilde{c}_{11}(pd)^{-\tilde{c}_{13}}\big\} \rightarrow 1. 
    \end{split}
    \end{equation} 
By Lemma~\ref{lemmaeqGGM} and the OR rule, we have $|\widehat{E}_{\tOr,\delta}\cap E| \geq \sum_{j=1}^{p} |\widehat{S}_{\delta,j}\cap S_j|$, which implies that  
\begin{equation}\label{thm7.E.E}
  E\left[  \frac{|\widehat{E}_{\tOr,\delta}\cap E|}{| E| \vee 1}\right]   \geq E\left[ \sum_{j=1}^{p} \frac{|\widehat{S}_{\delta,j}\cap S_j|}{| E| \vee 1}\right].
\end{equation} 
Combining (\ref{thm7.S.E.1}) and (\ref{thm7.E.E}), we complete the proof of Theorem~\ref{thm_power_fggm}. 
\end{proof}

\subsection{Proof of Lemma~\ref{lemma2}} 

\begin{proof}[Proof of Lemma~\ref{lemma2}] 
With $\Sigma_{X_{j}X_{k}} = \Sigma_{X_{j}X_{j}}^{1/2} C_{X_{j}X_{k}} \Sigma_{X_{k}X_{k}}^{1/2}$, $\Sigma_{X_jX_j}^{1/2} = \sum _{l=1}^{\infty} \omega_{jl}^{1/2}\phi_{jl}\otimes\phi_{jl},$ and $C_{X_{j}X_{k}} = \sum _{l=1}^{\infty} \sum _{m=1}^{\infty} \tcorr(\xi_{jl},\xi_{km}) (\phi_{jl} \otimes \phi_{km})$, we can prove Lemma~\ref{lemma2} akin to Theorem 2 in \cite{solea2022copula}. Hence, the proof is omitted.        
\end{proof}

\section{Additional derivations}
\label{supp.details}

\subsection{Simplified objective functions} \label{obj.dedu} 
In Section \ref{sec:construct}, we give the corresponding equivalent forms of the objective function in (\ref{opt.R.cor}), under E\ref{R.ex1}, E\ref{R.ex2} and E\ref{R.ex3}. In this section, we will provide detailed derivations for these simplified forms. Consider any  $\bx = (x_1,\dots,  x_p)^\T \in  {\boldsymbol{\mathcal{H}}}$, where each $x_j(\cdot) = \sum_{l =1}^{\infty} c_{jl} \phi_{jl}(\cdot) \in {\mathcal{H}}_j.$

E\ref{R.ex1}: Consider $R_{X_j X_j} = \sum _{l=1}^{\infty} r (\phi_{jl} \otimes \phi_{jl})$ for $r \in [0,1].$ By Lemma~\ref{lemma2} and the definition of operator norm, we have 
\begin{equation*} \label{opnorm1}
    \begin{split}
        \|C_{ X_j{X}_j} - R_{X_j X_j}\|_{\cL} &= \sup_{\|x_j\| \leq 1} \big\|(C_{X_j X_j}-R_{X_j X_j})(x_j)\big\| = \sup_{\|x_j\| \leq 1} \Big\|\sum_{l=1}^{\infty} (1-r) (\phi_{jl} \otimes \phi_{jl})(x_j) \Big\|\\ 
      &=  \sup_{\|x_j\| \leq 1} \Big\|\sum_{l=1}^{\infty} (1-r) \langle \phi_{jl}, x_j \rangle  \phi_{jl} \Big\| 
      =  \sup_{\|x_j\| \leq 1} \Big\{\sum_{l=1}^{\infty} (1-r)^2 \langle \phi_{jl}, x_j \rangle^2 \Big\}^{1/2} \\
       &= |1-r| \sup_{\|x_j\| \leq 1} \Big\{\sum_{l=1}^{\infty}  \langle \phi_{jl}, x_j \rangle^2 \Big\}^{1/2} = |1-r| \sup_{\|x_j\| \leq 1} \|x_j\|  \\
       &= |1-r| = 1-r,
    \end{split}
\end{equation*} 
which means the objective function can be simplified to $\min _{r} (1-r).$ 

E\ref{R.ex2}: Consider $R_{X_j X_j} = \sum _{l=1}^{\infty} r_{j} (\phi_{jl} \otimes \phi_{jl})$ for $r_j \in [0,1].$ By the similar arguments as above, we can obtain that 
\begin{equation*} \label{opnorm2}
         \|C_{ X_j{X}_j} - R_{X_j X_j}\|_{\cL} = 1-r_j, 
\end{equation*}
which means the the objective functions can be simplified to 
$\min _{ (r_1,\dots,r_p)} \sum_{j} (1-r_j).$ 

E\ref{R.ex3}: Consider $R_{X_j X_j} = \sum _{l=1}^{\infty} r_{jl} (\phi_{jl} \otimes \phi_{jl})$ for $r_{jl} \in [0,1].$ Likewise,  we obtain that
\begin{equation*} \label{opnorm3}
    \begin{split}
      \|C_{ X_j{X}_j} - R_{X_j X_j}\|_{\cL} &= \sup_{\|x_j\| \leq 1} \big\|(C_{X_j X_j}-R_{X_j X_j})(x_j)\big\| = \sup_{\|x_j\| \leq 1} \Big\|\sum_{l=1}^{\infty} (1-r_{jl}) (\phi_{jl} \otimes \phi_{jl})(x_j) \Big\|\\ 
      &=  \sup_{\|x_j\| \leq 1} \Big\|\sum_{l=1}^{\infty} (1-r_{jl}) \langle \phi_{jl}, x_j \rangle  \phi_{jl} \Big\| 
      =  \sup_{\|x_j\| \leq 1} \Big\{\sum_{l=1}^{\infty} (1-r_{jl})^2 \langle \phi_{jl}, x \rangle^2 \Big\}^{1/2}\\
      &=  \sup_{\|x_j\| \leq 1} \Big\{\sum_{l=1}^{\infty} (1-r_{jl})^2 c_{jl}^2 \Big\}^{1/2}\\ 
       &\leq \sup_{l} |1-r_{jl}|\sup_{\|x_j\| \leq 1} \Big(\sum_{l=1}^{\infty} c_{jl}^2 \Big)^{1/2}  = \sup_{l } |1-r_{jl}|, 
    \end{split}
\end{equation*}
which indicates that $ \|C_{ X_j{X}_j} - R_{X_j X_j}\|_{\cL}\leq \sup_{l } |1-r_{jl}|.$ 
Next we will prove that $ \|C_{ X_j{X}_j} - R_{X_j X_j}\|_{\cL} = \sup_{l } |1-r_{jl}|.$  
Let’s consider two scenarios. 
First, when $\sup_{l } |1-r_{jl}| = \max_{l} |1-r_{jl}|$, we set $c_{jl^*} = 1$ for $l^*$ such that $|1-r_{jl^*}| = \max_{l} |1-r_{jl}|$. Consequently, we obtain $\|C_{ X_j{X}_j} - R_{X_j X_j}\|_{\cL} = \max_{l}|1-r_{jl}|=\sup_{l } |1-r_{jl}|$.  
Second, if $\sup_{l } |1-r_{jl}| \notin \big\{|1-r_{jl}|: l\geq 1\big\}$, let $\sup_{l } |1-r_{jl}| = |1-r^*|$. In this case, there exists a subsequence $l_m$ for $m\geq1$ such that $r_{jl_m} \to r^*$ as $m\to \infty$. Under this scenario, we set $c_{jl_m} \to 1$ as $m\to \infty$. Consequently, we conclude that $\|C_{ X_j{X}_j} - R_{X_j X_j}\|_{\cL}=\sup_{l } |1-r_{jl}|$. Combining the results under two scenarios, we complete the proof. 

For (\ref{opt.R.cor}) in E\ref{R.ex3}, the objective function is simplified as:
$ \min _{ (\Bar{\br}_1,\dots,\Bar{\br}_p)}  \sum_{j} \sup_{l} |1-r_{jl}|.$
This corresponds to a non-smooth positive semidefinite programming problem, which poses a computationally challenging task. 
For the sake of computational simplicity, we instead consider solving the following smooth positive semidefinite programming problem: 
\begin{equation}
\label{opt.pro.E3}
    \begin{split}
\min _{ (\Bar{\br}_1,\dots,\Bar{\br}_p)}~&  \sum_{j=1}^{p} \sum_{l=1}^{k_n} (1-r_{jl}) \\
 \rm{subject~to~} & r_{jl}\in [0,1],  ~
 2\widehat{\bC}_{XX} - \widehat{\bR}_{XX}\succeq 0.
\end{split}
\end{equation}

\begin{remark}\label{eq.obj.E3} 
We aim to show that the solution set of the programming problem with objective function $ \min_{ (\Bar{\br}_1,\dots,\Bar{\br}_p)}  \sum_{j} \sup_{l} |1-r_{jl}|$ is equivalent to that of the programming problem with objective function $\min_{ (\Bar{\br}_1,\dots,\Bar{\br}_p)} \sum_{j=1}^{p} \sum_{l=1}^{k_n} (1-r_{jl})$. 
Let $ {\Bar{\br}}^*= ( {\Bar{\br}}_1^*,\dots, {\Bar{\br}}_p^*)$ be the optimal solution to $\min_{ (\Bar{\br}_1,\dots,\Bar{\br}_p)} \sum_j \sup_{l} (1-r_{jl})$ with the same constraints in the optimization problem (\ref{opt.pro.E3}), and $ \Bar{\br}^{\star}= ( \Bar{\br}_1^{\star},\dots, \Bar{\br}_p^{\star})$ be the optimal solution to (\ref{opt.pro.E3}). 
Define the set $F = \{\Bar{\br}= (\Bar{\br}_1,\dots,\Bar{\br}_p) : r_{jl}\in [0,1] \text{ and } 2\widehat{\mathbf{C}}_{XX} - \widehat{\mathbf{R}}_{XX}\succeq 0\}$ as the feasible domain. 
On the one hand, when solving $\min_{ (\Bar{\br}_1,\dots,\Bar{\br}_p)} \sum_j \sup_{l} (1-r_{jl})$, each $r_{jl}$ tends to a higher value within $F$, which implies that $r_{jl}^* \geq r_{jl}^{\star}$ for each $j\in[p]$ and $l\in [k_n]$. 
On the other hand, since $ \Bar{\br}^{\star}= ( \Bar{\br}_1^{\star},\dots, \Bar{\br}_p^{\star})$ is the optimal solution to $\min_{ (\Bar{\br}_1,\dots,\Bar{\br}_p)} \sum_j \sum_{l} (1-r_{jl}),$ we have $ \sum_{j} \sum_{l} (1-r_{jl}^{\star}) \leq \sum_{j}\sum_{l} (1-r_{jl}^{*})$. 
This implies that $\sum_{j}\sum_{l} r_{jl}^{*} \leq  \sum_{j}\sum_{l} r_{jl}^{\star}$. 
Both sides hold if and only if $r_{jl}^* = r_{jl}^{\star}$ for each $j\in[p]$ and $l\in [k_n]$. 

\end{remark}


\subsection{Coordinate mapping} \label{cor.map}  
As demonstrated in Section \ref{sec:implement}, $\mathcal{H}_j$ is spanned by a finite set of functions $\mathcal{B}_j = \{b_{j1},\dots,b_{jk_n }\}$. Each $X_{ij}$ can be expressed as a linear combination: ${X}_{ij} = c_{ij1}b_{j1} +\dots+c_{ijk_n}b_{jk_n}$, and its coordinate can be represented as $[{X}_{ij}]_{\mathcal{B}_j}$. Likewise, the coordinate of any operator $K: \mathcal{H}_j \to \mathcal{H}_k$ is denoted as $_{{\mathcal{B}}_{k}}[K]_{{\mathcal{B}}_{j}}$. This mapping, $K \to~ _{{\mathcal{B}}_{k}}[K]_{{\mathcal{B}}_{j}}$, is referred to as the coordinate mapping.
There are five main properties of the coordinate mapping \citep{li2018nonparametric} that are crucial for subsequent analysis. Here, $K_1$ and $K_2$ represent operators mapping from $\mathcal{H}_j$ to $\mathcal{H}_k$, $a$ and $b$ are real numbers, and ${\boldsymbol{\mathcal{B}}}$ denotes the Cartesian product of ${\mathcal{B}_1},\dots, {\mathcal{B}_p}$. 
\begin{enumerate}[P1.]
    \item\label{cor.pro1} linearity: $_{{\mathcal{B}}_{k}}[a K_1 + bK_2]_{{\mathcal{B}}_{j}} = a\big(\, _{{\mathcal{B}}_{k}}[K_1]_{{\mathcal{B}}_{j}}\,\big) + b\big(\, _{{\mathcal{B}}_{k}}[K_2]_{{\mathcal{B}}_{j}}\,\big)$;  
    \item\label{cor.pro2} tensor product: $_{\mathcal{B}_k}[X_{ik} \otimes X_{ij}]_{\mathcal{B}_j} =[{X}_{ik}]_{\mathcal{B}_k} [{X}_{ij}]_{\mathcal{B}_j}^\T \bG_j $;
    \item\label{cor.pro3} operator calculation: $[K(X_{ij})]_{\mathcal{B}_k} = \big(\, _{{\mathcal{B}}_{k}}[K]_{{\mathcal{B}}_{j}}\,\big)[{X}_{ij}]_{\mathcal{B}_j}$;
    \item\label{cor.pro4} inner product: $\langle X_{ij},\widetilde{X}_{ij} \rangle = [{X}_{ij}]_{\mathcal{B}_j}^\T \bG_j [\widetilde{X}_{ij}]_{\mathcal{B}_j}$;
    \item\label{cor.pro5} operator matrix: $_{{\boldsymbol{\mathcal{B}}}}[{\bC}_{XX}]_{{\boldsymbol{\mathcal{B}}}} = \big(\, _{\mathcal{B}_j}[C_{X_j X_k}]_{\mathcal{B}_k} \, \big)_{j,k \in [p]}$. 
\end{enumerate}

First, we will present the Karhunen-Lo$\grave{\hbox{e}}$ve expansion by coordinate mapping. By P\ref{cor.pro1} and P\ref{cor.pro2}, we can deduce that 
 $   _{{\mathcal{B}}_{j}}[\widehat{\Sigma}_{X_{j}X_{j}}]_{{\mathcal{B}}_{j}} 
    = ({n}^{-1}\sum_{i=1}^n [ X_{ij} ]_{{\mathcal{B}}_{j}} [ X_{ij}]_{{\mathcal{B}}_{j}}^{\T}) \bG_j.$ By P\ref{cor.pro3}, 
$(\hat\omega_{jl}, \hat\phi_{jl}(\cdot))_{l \geq 1}$ 
is the eigenvalue/eigenfunction pair  of $\widehat{\Sigma}_{X_{j}X_{j}}$  if and only if 
$_{{\mathcal{B}}_{j}}([\widehat{\Sigma}_{X_{j}X_{j}}]_{{\mathcal{B}}_{j}})[\hat\phi_{jl}]_{{\mathcal{B}}_{j}}  = \hat\omega_{jl}[\hat\phi_{jl}]_{{\mathcal{B}}_{j}},$
which indicates that  
\begin{equation*}
    {n}^{-1} \bG_j^{1/2}  \sum_{i=1}^n \big([ X_{ij}  ]_{{\mathcal{B}}_{j}} [ X_{ij}   ]_{{\mathcal{B}}_{j}}^{\T}\big) \bG_j^{1/2} \big(\bG_j^{1/2} [\hat\phi_{jl}]_{{\mathcal{B}}_{j}}\big) = \hat\omega_{jl}\big(\bG_j^{1/2}[\hat\phi_{jl}]_{{\mathcal{B}}_{j}}\big),
\end{equation*}
i.e., $(\hat\omega_{jl}, \bG_j^{1/2}[\hat\phi_{jl}]_{{\mathcal{B}}_{j}})_{l \geq 1}$ is the eigenvalue/eigenvector pair of the matrix $$ {n}^{-1} \bG_j^{1/2}  \sum_{i=1}^n \big([ X_{ij}  ]_{{\mathcal{B}}_{j}} [ X_{ij} ]_{{\mathcal{B}}_{j}}^{\T}\big) \bG_j^{1/2}.$$ 
Hence, we can obtain the coordinate of $\hat\phi_{jl}$ as $[\hat\phi_{jl}]_{{\mathcal{B}}_{j}} = \bG_j^{\dag 1/2} \bv_{jl}$, where $\bv_{jl}$  is the eigenvector of ${n}^{-1} \bG_j^{1/2}  \sum_{i=1}^n \big([ X_{ij} ]_{{\mathcal{B}}_{j}} [ X_{ij} ]_{{\mathcal{B}}_{j}}^{\T}\big) \bG_j^{1/2}$. 
Finally, we can obtain the empirical Karhunen-Lo$\grave{\hbox{e}}$ve expansion of (\ref{exX})  as $  X_{ij} -\hat{\mu}_{j} = \sum_{l = 1}^{k_n} {\hat\xi}_{ijl} \hat\phi_{jl}$ by  P\ref{cor.pro4},
 where
\begin{equation}\label{coordinate.score}
    {\hat\xi}_{ijl}  = \langle X_{ij} - \hat{\mu}_{j}, \hat\phi_{jl}\rangle   =   [X_{ij}]_{{\mathcal{B}}_{j}}^\T \bG_j [\hat\phi_{jl}]_{{\mathcal{B}}_{j}}  =   [X_{ij}]_{{\mathcal{B}}_{j}}^\T \bG_j^{1/2} \bv_{jl}.
\end{equation}

Second, we will derive that $2\widehat{\bC}_{XX} - \widehat{\bR}_{XX}\succeq 0$ is implied by $2 \widehat\bTheta_C - \widehat\bTheta_{R}  \succeq 0$.
As shown in Section \ref{sec:implement}, $\widehat{C}_{{X}_{j}{X}_{k}}=(1-\gamma_n) \sum _{l=1}^{k_n} \sum _{m=1}^{k_n} \widehat\Theta_{jklm}^{\text{S}} (\hat \phi_{jl}\otimes  \hat \phi_{km}) + \gamma_n I(j=k) \widehat I_{X_jX_j}$, with $\widehat \Theta_{jklm}^{\text{S}}=n^{-1}\sum_{i=1}^n(\widehat\xi_{ijl}-n^{-1}\sum_{i=1}\widehat\xi_{ijl})(\widehat\xi_{ikm}-n^{-1}\sum_{i=1}\widehat\xi_{ikm})/(\hat\omega_{jl}^{1/2}\hat\omega_{km}^{1/2})$ and $\hat{I}_{X_jX_j} = \sum _{l=1}^{k_n}\hat{\phi}_{jl} \otimes \hat{\phi}_{jl}$.   
By P\ref{cor.pro1} and P\ref{cor.pro2}, 
the coordinate of $\widehat{C}_{{X}_{j}{X}_{k}}$ can be represented as
\begin{equation}\label{opr.map.C}
\begin{split}
      _{{\mathcal{B}}_{j}}[ \widehat{C}_{{X}_{j}{X}_{k}}]_{{\mathcal{B}}_{k}} 
      &=  {_{{\mathcal{B}}_{j}}\Big[}  \sum _{l=1}^{k_n} \sum _{m=1}^{k_n} \{(1-\gamma_n)\widehat\Theta_{jklm}^{\text{S}}+\gamma_n I(l=m)I(j=k)\} (\hat \phi_{jl}\otimes  \hat \phi_{km}) \Big]_{{\mathcal{B}}_{k}} \\
      &= \sum _{l=1}^{k_n} \sum _{m=1}^{k_n}
\big\{(1-\gamma_n) \widehat\Theta_{jklm}^{\text{S}} + {\gamma_n I(l=m)I(j=k)} \big\} 
\big( [\hat{\phi}_{jl}]_{{\mathcal{B}}_{j}} [\hat{\phi}_{km}]_{{\mathcal{B}}_{k}}^\T \big) \bG_k\\
& = \widehat\bPhi_{j}  \widehat\bTheta_{C_{jk}}  \widehat\bPhi_{k}^\T    \bG_{k},
\end{split}
\end{equation}
where $\widehat\bTheta_{C_{jk}}=(1-\gamma_n)(\widehat{\Theta}_{jklm}^{\text{S}})_{k_n\times k_n}+  \gamma_n I(j=k) \bI_{k_n} \in \eR^{k_n \times k_n}$ and $\widehat\bPhi_j \in {\mathbb R}^{k_n \times k_n}$ with its $l$th column $[\hat{\phi}_{j l}]_{{\mathcal{B}}_{j}}$ for $l\in[k_n]$. Similarly, for $j\in [p]$,   we can get the the coordinate of $ \widehat{R}_{{X}_{j}{X}_{j}}$ as 
\begin{equation}\label{opr.map.R}
      _{{\mathcal{B}}_{j}}[ \widehat{R}_{{X}_{j}{X}_{j}}]_{{\mathcal{B}}_{j}} 
 =  \widehat\bPhi_{j}  \widehat\bTheta_{R_{jj}}   \widehat\bPhi_{j}^\T    \bG_{j},
\end{equation}
where $\widehat\bTheta_{R_{jj}} = r \bI_{k_n}\in \eR^{k_n \times k_n}$ under E\ref{R.ex1}, 
 $\widehat\bTheta_{R_{jj}} = r_j \bI_{k_n}\in \eR^{k_n \times k_n}$ under E\ref{R.ex2}, and 
 $\widehat\bTheta_{R_{jj}}  = \diag(r_{j1} ,\dots, r_{jk_n}) \in \eR^{k_n \times k_n}$ under E\ref{R.ex3}.
By (\ref{opr.map.C}), (\ref{opr.map.R}), P\ref{cor.pro3}--P\ref{cor.pro5}, 
we have 
 \begin{equation}\label{opt.R.corcond}
    \begin{split}
 &\big\langle (2\widehat{\bC}_{XX} - \widehat{\bR}_{XX})(\bx), \bx \big\rangle \geq 0 \\
\Longleftrightarrow~  & \big[(2\widehat{\bC}_{XX} - \widehat{\bR}_{XX})(\bx) \big]_{{\mathcal{B}}} ^\T (\oplus_{j\in [p]} \bG_{j}) [\bx]_{{\mathcal{B}}} \geq 0 \\
\Longleftrightarrow~  & [\bx]_{{\mathcal{B}}}^\T~_{{\mathcal{B}}}[2\widehat{\bC}_{XX} - \widehat{\bR}_{XX}]_{{\mathcal{B}}} ^\T~ (\oplus_{j\in [p]} \bG_{j}) ~[\bx]_{{\mathcal{B}}} \geq 0\\
 \Longleftrightarrow~  & [\bx]_{{\mathcal{B}}}^\T~ (\oplus_{j\in [p]} \bG_{j})^\T~ (\oplus_{j\in [p]} \widehat\bPhi_{j}) (2 \widehat\bTheta_C- \widehat\bTheta_{R}) (\oplus_{j\in [p]} \widehat\bPhi_{j})^\T~ (\oplus_{j\in [p]} \bG_{j}) ~[\bx]_{{\mathcal{B}}} \geq 0, 
\end{split}
\end{equation} 
where  $[\bx]_{\mathcal{B}} = \big([x_1]_{\mathcal{B}_1}^\T,\dots,[x_p]_{\mathcal{B}_p}^\T \big)^\T$, 
$\widehat\bTheta_C= (\widehat\bTheta_{C_{jk}})_{j,k\in[p]}$, 
$\widehat\bTheta_R = \diag (r \bI_{k_n}, \dots ,r \bI_{k_n})$ under E\ref{R.ex1}, 
$\widehat\bTheta_R = \diag (r_1 \bI_{k_n}, \dots ,r_p \bI_{k_n})\in \eR^{pk_n \times pk_n}$ under E\ref{R.ex2}, 
$\widehat\bTheta_R = \diag(r_{11}, \dots,r_{1k_n}, \dots, r_{p1}, \dots, r_{pk_n}) \in \eR^{pk_n \times pk_n}$ under E\ref{R.ex3},  
$\oplus_{j\in [p]} \widehat\bPhi_{j}  = \tdiag(\widehat\bPhi_1,\dots,\widehat\bPhi_p)$, and $\oplus_{j\in [p]} \bG_{j}  = \tdiag(\bG_1,\dots,\bG_p)$. 
(\ref{opt.R.corcond}) means that $2\widehat{\bC}_{XX} - \widehat{\bR}_{XX}\succeq 0$ if and only if 
\begin{equation}\label{opt.Theta}
  (\oplus_{j\in [p]} \bG_{j})^\T~ (\oplus_{j\in [p]} \widehat\bPhi_{j}) (2 \widehat\bTheta_C- \widehat\bTheta_{R}) (\oplus_{j\in [p]} \widehat\bPhi_{j})^\T~ (\oplus_{j\in [p]} \bG_{j}) \succeq 0, 
\end{equation} 
which shows that $2\widehat{\bC}_{XX} - \widehat{\bR}_{XX}\succeq 0$ is implied by $2 \widehat\bTheta_C - \widehat\bTheta_{R}  \succeq 0.$

\subsection{Algorithms} 
\label{smsec:algorithm}
The construction of functional Model-X knockoffs utilizes the Karhunen-Lo$\grave{\hbox{e}}$ve expansion. 
As outlined in Section \ref{sec:algorithm}, Algorithm \ref{Alg:2} comprises three main steps. 
In Step \ref{step:1}, three expressions of $\widehat\bTheta_{R}$ are obtained by determining the parameters of $r$, $(r_1, \dots, r_p)$,  and $(\Bar{\br}_1, \dots, \Bar{\br}_p)$ 
which involve solving the optimization problems in (\ref{opt.R.mat.1}), (\ref{opt.R.mat.2}), and (\ref{opt.R.mat.3}), respectively.
The second step involves constructing the FPC scores of $\widetilde{\bX}_i$ using the estimated FPC scores of ${\bX}_i$ and the procedure outlined in Algorithm \ref{Alg1}. The expression of $C_{X_jX_k}$ in Lemma \ref{lemma2} involves the correlations between the FPC scores, which are equivalent to the covariances between the normalized  FPC scores. Therefore, our focus is on the distribution of the estimated normalized  FPC scores. 
Considering $\big(\bX_i^\T, \widetilde \bX_i^\T\big)$ as a MGP for each $i \in [n]$ and using (\ref{coordinate.score}), (\ref{R.cor}), and (\ref{C.est.shrink}), we derive that the estimated normalized  FPC scores satisfy 
\begin{equation*}
\begin{split}
     &   \diag\big(\widehat\bW^{-1/2}\widehat\bA^\T, \widehat\bW^{-1/2}\widehat\bA^\T\big) \big([X_{i1}]_{\mathcal{B}_1}^\T,\dots,[X_{ip}]_{\mathcal{B}_p}^\T, [\widetilde X_{i1}]_{\mathcal{B}_1}^\T,\dots,[\widetilde X_{ip}]_{\mathcal{B}_p}^\T\big)^\T\\
    =  & \diag\big(\widehat\bW^{-1/2},\widehat\bW^{-1/2}\big)\big({\hat\xi}_{i1l},\dots,{\hat\xi}_{i1k_n},\dots,{\hat\xi}_{ipl},\dots,{\hat\xi}_{ipk_n},{\check\xi}_{i1l},\dots,
    {\check\xi}_{i1k_n},\dots,{\check\xi}_{ipl},\dots,{\check\xi}_{ipk_n}\big) \\
     \sim & \mathcal{N}(\bzero_{2pk_n},\widehat\bTheta),
\end{split}
\end{equation*}
 where $\widehat\bA = \tdiag(\bG_{1}\widehat\bPhi_{1},\dots,\bG_{p}\widehat\bPhi_{p})  
\in \eR^{pk_n \times pk_n}$,  $\widehat\bW = \tdiag(\hat\omega_{11}, \dots, \hat\omega_{1k_n},\dots, \hat\omega_{p1}, \dots, \hat\omega_{pk_n})$ is a normalization matrix, and 
\begin{equation} \label{Theta.all}
    \widehat\bTheta = \left(
\begin{array}{cc}
\widehat\bTheta_C & \widehat\bTheta_C - \widehat\bTheta_R\\
\widehat\bTheta_C - \widehat\bTheta_R & \widehat\bTheta_C 
\end{array}
\right). 
\end{equation}
Note $\widehat\bTheta_C \in \eR^{pk_n \times pk_n}$ and $\widehat\bTheta_R\in \eR^{pk_n \times pk_n}$ in (\ref{Theta.all}) are obtained from (\ref{opt.R.corcond}) and (\ref{opt.Theta}). Then for $i\in [n]$,  by conditional distribution under multivariate Gaussianity, we can obtain that
\begin{equation*}\label{cond.dis}
    \widehat\bW^{-1/2}\widehat\bA^\T \big( [\widetilde X_{i1}]_{\mathcal{B}_1}^\T,\dots,[\widetilde X_{ip}]_{\mathcal{B}_p}^\T\big)^{\T} \Big|   \widehat\bW^{-1/2}\widehat\bA^\T ([X_{i1}]_{\mathcal{B}_1}^\T,\dots,[X_{ip}]_{\mathcal{B}_p}^\T)^{\T} \sim \mathcal{N}(\widehat\bmu_{\tilde X|X},\widehat\bTheta_{\tilde X|X}),
\end{equation*} 
where 
$\widehat\bmu_{\tilde X|X} =  (  \bI_{pk_n} - \widehat\bTheta_R \widehat\bTheta_C^{-1} )\widehat\bW^{-1/2}\widehat\bA^\T([X_{i1}]_{\mathcal{B}_1}^\T,\dots,[X_{ip}]_{\mathcal{B}_p}^\T)^{\T}$ and
$\widehat\bTheta_{\tilde X|X} = 2\widehat\bTheta_{R}- \widehat\bTheta_{R} \widehat\bTheta_C^{-1}\widehat\bTheta_{R}.$
This conditional distribution forms the foundation for sampling the estimated  FPC scores of $\widetilde \bX_i$. Finally, we construct the functional knockoffs $\widecheck{\bX}_i = (\widecheck{X}_{i1}, \dots, \widecheck{X}_{ip})^\T$ from these estimated  FPC scores using the Karhunen-Loève expansion. 

\section{Additional empirical results}
\label{sec:additional.emp}
\subsection{Additional simulation results}
Table~\ref{tab:4} and \ref{tab:5} present the empirical power and FDR for partially observed functional data under the model settings of SFLR and FGGM in Section \ref{sec:sim}, respectively. 
Considering the similar conclusions drawn from SFLR and FFLR for fully observed functional data, 
we only apply comparison methods to SFLR for partially observed functional data due to computational efficiency. 

\begin{table}[htbp]
  \caption{The empirical power and FDR in SFLR for partially observed functional data.}
  \label{tab:4}
  \centering
     \resizebox{5.5in}{!}{
\begin{tabular}{*{12}{c}}
  \toprule
  \multirow{2}*{$p$} & \multirow{2}*{$n$} & \multicolumn{2}{c}{KF1} & \multicolumn{2}{c}{KF2} & \multicolumn{2}{c}{KF3} & \multicolumn{2}{c}{TF} &\multicolumn{2}{c}{GL}\\
  \cmidrule(lr){3-4}\cmidrule(lr){5-6}\cmidrule(lr){7-8}\cmidrule(lr){9-10}\cmidrule(lr){11-12}
  &  & FDR &Power & FDR &Power & FDR &Power & FDR &Power& FDR &Power\\
  \midrule
  50 &  100 & 0.18  & 0.94  &0.18   &0.96  & 0.19  
     &0.96  &0.21  &0.79  &0.28  & 1.00 \\


      &  200 & 0.16  &0.99   &   0.15& 0.98 &  0.16 &0.98  & 0.18  &0.85  &0.26 &1.00  \\

  100 &  100 & 0.17  & 0.92 &0.15   &0.92 &0.15  &0.92  &  0.16 &0.62  &0.47  &1.00  \\


      &  200 & 0.20  & 1.00  & 0.20  &1.00  &0.18   &1.00  & 0.12  & 0.93 & 0.31 &1.00  \\
      
 150 &  100 & 0.09  & 0.99  &0.09   &0.99  &0. 10 
      &0.99  &  0.14 &0.70  &0.68  &1.00  \\


      &  200 & 0.13  & 1.00  & 0.12  &1.00  &0.12   &1.00  & 0.08  & 0.96 & 0.37 &1.00  \\
      
  \bottomrule
\end{tabular}
}

\end{table}

\begin{table}[tbp]
  \caption{The empirical power and FDR in FGGM for partially observed functional data.}
  \label{tab:5}
  \centering
       \resizebox{4in}{!}{
\begin{tabular}{*{8}{c}}
  \toprule
  \multirow{2}*{$p$} & \multirow{2}*{$n$} & \multicolumn{2}{c}{KF1} & \multicolumn{2}{c}{TF} &\multicolumn{2}{c}{GL}\\
  \cmidrule(lr){3-4}\cmidrule(lr){5-6}\cmidrule(lr){7-8}
  &  & FDR &Power & FDR &Power& FDR &Power\\
  \midrule
   50   &  100 & 0.18   &0.74 &0.15   &0.50  &0.23  &0.75  \\

       &  200 & 0.20  &0.94 & 0.20  &0.79  &0.22  &0.94  \\

  100  &  100 & 0.17  & 0.63   &  0.12 &0.48  &0.20  &0.68  \\

       &  200 & 0.19  & 0.84    & 0.18  & 0.73 & 0.22 &0.86  \\
       
  \bottomrule
\end{tabular}
}

\end{table}

\subsection{Specific ROIs}
\label{Sec:tab}
Table~\ref{tab:regions} presents $34$ regions of interest (ROIs) and the associated labelling index for the emotion related fMRI dataset in Section \ref{real1}. 

\begin{table}
  \caption{The labeling index of ROIs in Section \ref{real1}.}
  \label{tab:regions}
  \centering
     \resizebox{5in}{!}{
\begin{tabular}{*{4}{c}}
  \toprule
  Index & Region & Index & Region\\
  \midrule
1 & bankssts                   &2& caudal anterior cingulate \\
3 & caudal middle frontal      &4& cuneus \\
5 & entorhinal                 &6& fusiform \\
7 & inferior parietal          &8& inferior temporal \\
9 & isthmus cingulate          &10& lateral occipital \\
11& lateral orbitofrontal      &12& lingual \\
13& medial orbito frontal      &14& middle temporal \\
15& parahippocampal            &16& paracentral \\
17& parsopercularis            &18& parsorbitalis \\
19& parstriangularis           &20& pericalcarine \\
21& postcentral                &22& posterior cingulate \\
23& precentral                 &24& precuneus \\
25& rostral anterior cingulate &26& rostral middle frontal \\
27& superior frontal           &28& superior parietal \\
29& superior temporal          &30& supramarginal \\
31& frontal pole               &32& temporal pole \\
33& transverse temporal        &34& insula \\
  \bottomrule
\end{tabular}
}

\end{table}

\newpage
\spacingset{1.2}
\section*{References}
\begin{description}
\item Cand{\`e}s, E., Fan, Y., Janson, L. and Lv, J. (2018). 
Panning for gold: `model-X' knockoffs forhigh dimensional controlled variable selection, 
{\it Journal of the Royal Statistical Society: Series B}  
{\bf 80}: 551–577. 
\item Fan, Y., Demirkaya, E., Li, G. and Lv, J. (2020a).
Rank: Large-scale inference with graphical nonlinear knockoffs, 
{\it Journal of the American Statistical Association}  
{\bf 115}: 362–379. 
\item Fang, Q., Guo, S. and Qiao, X. (2022). 
Finite sample theory for high-dimensional functional/scalar time series with applications, 
{\it Electronic Journal of Statistics}  {\bf 16}: 527–591.
\item Guo, S. and Qiao, X. (2023). 
On consistency and sparsity for high-dimensional functional time series with application to autoregressions, 
{\it Bernoulli}  
{\bf 29}: 451–472. 
\item Li, B. and Solea, E. (2018). 
A nonparametric graphical model for functional data with application to brain networks based on fMRI,
{\it Journal of the American Statistical Association}  
{\bf113}: 1637–1655. 
\item Li, J. and Maathuis, M. H. (2021). GGM knockoff filter: False discovery rate control for Gaussian graphical models, {\it Journal of the Royal Statistical Society: Series B}  
{\bf 83}: 534–558. 
\item Solea, E. and Li, B. (2022). 
Copula Gaussian graphical models for functional data, {\it Journal of the American Statistical Association}  
{\bf 117}: 781–793. 
\end{description}
\end{document}